\newcommand{\smalltitle}{}
\newcommand{\doctitle}{On the Existence of Dynamics of Wheeler-Feynman Electromagnetism}

\def\arxiv{1}
\def\kindle{0}
\def\draftmark{0}
\newcommand{\ifarxiv}[2]{{\if\arxiv 1 #1 \else #2 \fi}}

\if\arxiv 1 
  \documentclass[oneside,12pt]{article}

  \if\kindle 1
    \usepackage[paperwidth=16cm, paperheight=22cm, hmargin=0.2cm, vmargin=0.2cm]
      {geometry}
  \else
    \usepackage[letterpaper,left=1in,right=1in,top=1in,bottom=1in]{geometry}
  \fi
  \usepackage{fancyhdr}
  \usepackage{ifpdf}

  \usepackage{amsfonts,amsmath, amssymb, amsthm, bbm, subfigure, wasysym}

  \fancypagestyle{myheadsfoots}{
    \fancyhead[L]{\small\textit{\doctitle}}
    \fancyhead[R]{\small\thepage}
    \fancyfoot{}
  }
\else 
  \documentclass[oneside,10pt]{article}

  \usepackage[letterpaper,left=1in,right=1in,top=1in,bottom=1in]{geometry}
  \usepackage{fancyhdr}
  \usepackage{ifpdf}

  \usepackage{amsfonts,amsmath, amssymb, amsthm, bbm, subfigure, wasysym}

  \fancypagestyle{myheadsfoots}{
    \fancyhead[L]{\small\textit{\doctitle}}
    \fancyhead[R]{\small\thepage}
    \fancyfoot{}
  }
\fi

\usepackage{multicol,enumerate,bbm,txfonts,ifpdf,xrf,ifthen,refcount}

\usepackage[disable]{todonotes}

  \if\draftmark 1
    \usepackage{graphicx}
    \usepackage{type1cm}
    \usepackage{eso-pic}
    \usepackage{color}

    \makeatletter
    \AddToShipoutPicture{%
                \setlength{\@tempdimb}{.5\paperwidth}%
                \setlength{\@tempdimc}{.5\paperheight}%
                \setlength{\unitlength}{1pt}%
                \put(\strip@pt\@tempdimb,\strip@pt\@tempdimc){%
            \makebox(0,0){\rotatebox{45}{\textcolor[gray]{0.85}%
            {\fontsize{6cm}{6cm}\selectfont{DRAFT}}}}%
                }%
    }
    \makeatother

  \fi

\ifpdf
  \usepackage{pdfsync}
  \usepackage[pdftitle={\doctitle},
              pdfauthor={G. Bauer, D.-A. Deckert, D. Duerr},
              pdfpagemode=UseOutlines,
              pdfstartview=FitH,
              bookmarks=true,
              bookmarksopen=true,
              bookmarksnumbered=true,
              bookmarkstype=toc,
              colorlinks=true,
              linkcolor=blue,
              citecolor=blue,
              urlcolor=blue]{hyperref}
\fi

\newtheorem{theorem}{Theorem}[section]
\newtheorem{lemma}[theorem]{Lemma}
\newtheorem{corollary}[theorem]{Corollary}

\newtheorem{definition}[theorem]{Definition}

\newtheorem{notation}[theorem]{Notation}

\newtheorem{remark}[theorem]{REMARK}

\newcommand{\bounds}{\mathtt{Bounds}}                   
\if\arxiv 1
  \renewcommand{\cal}[1]{{\mathcal{#1}}}                    
\else
  \renewcommand{\cal}[1]{{\mathcal{#1}}}                    
\fi
\newcommand{\Ran}{{\operatorname{Range}\,}}             	
\newcommand{\supp}{{\operatorname{supp}\,}}             
\newcommand{\vect}[1]{{\mathbf{#1}}}                    
\newcommand{\bb}[1]{{\mathbb{#1}}}                      
\newcommand{\intdv}[1]{{\int d^3#1\;}}                  
\newcommand{\charf}[1]{\mathbbm{1}_{#1}}                

\newcommand{\WF}{{\operatorname{WF}}}                   
\newcommand{\bigoh}{\operatorname{O}}                   
\newcommand{\braket}[1]{\left\langle #1 \right\rangle}  

\newcommand{\id}[0]{{\operatorname{id}}}            
\newcommand{\pQ}{\mathtt{Q}}
\newcommand{\pP}{\mathtt{P}}
\newcommand{\pQP}{\mathtt{(Q+P)}}
\newcommand{\pF}{\mathtt{F}}
\newcommand{\mA}[0]{{\mathsf A}}
\newcommand{\mW}[0]{{\mathsf W}}
\newcommand{\mJ}[0]{{\mathsf J}}

\newcommand{\Lip}[0]{{\operatorname{Lip}}}

\begin{document}

\pagestyle{empty}

\title{\textbf{\doctitle}\\\large{\smalltitle}}

\author{G. Bauer\thanks{gernot.bauer@fh-muenster.de}, D.-A. Deckert\thanks{deckert@math.ucdavis.edu}, D. D\"urr\thanks{duerr@math.lmu.de}}
  
\date{September 16, 2010, rev. \today}

\maketitle

\begin{abstract} We study the equations of Wheeler-Feynman electrodynamics which is an action-at-a-distance theory about world-lines of charges that interact through their corresponding advanced and retarded  Li\'enard-Wiechert field terms. 
The equations are non-linear, neutral, and involve time-like advanced as well as retarded arguments of unbounded delay. Using a reformulation in terms of  Maxwell-Lorentz electrodynamics without self-interaction,  which we have introduced in  a preceding work, we are able to establish the existence of conditional solutions. These are solutions that solve the Wheeler-Feynman equations on any finite time interval with prescribed continuations outside of this interval.  As a byproduct we also prove existence and uniqueness of solutions to the Synge equations on the time half-line for a given history of charge trajectories.\\


\noindent\textbf{Acknowledgments:} The authors want to thank Martin Kolb for his valuable comments. D.-A.D. gratefully acknowledges financial support from the \emph{BayEFG} of the \emph{Freistaat Bayern} and the \emph{Universi\"at Bayern e.V.} as well as  from the post-doc program of the DAAD.
\end{abstract}

\if\arxiv 1
  \tableofcontents
\fi
  \pagestyle{myheadsfoots}

\section{Introduction}

%

 Wheeler-Feynman electrodynamics (WF) describes the classical, electromagnetic interaction of a number of $N$ charges by action-at-a-distance \cite{wheeler_classical_1949}. In contrast to Maxwell-Lorentz electrodynamics the theory contains no fields and is free from ultraviolet divergences originating from ill-defined self-fields. Electrodynamics without fields was considered as early as 1845 by Gauss \cite{gau_letter_1845}  and continued to be of interest, e.g. \cite{schwarzschild_zur_1903, tetrode_ber_1922, fokker_ein_1929}. In particular, it led Wheeler and Feynman \cite{wheeler_interaction_1945,wheeler_classical_1949} to a satisfactory description of radiation damping: Accelerated charges are supposed to radiate and to loose energy thereby. How can this be accounted for in a theory without fields?  To answer this question Wheeler and Feynman introduced a so-called \emph{absorber condition} which needs to be satisfied by the world-lines of all charges, and they argue that it is  satisfied on thermodynamic scales. Under the absorber condition it is straightforwardly seen that the motion of any selected charge is described effectively by the Lorentz-Dirac equation, an equation Dirac derived \cite{dirac_classical_1938} to explain the phenomena of radiation damping; see our short discussion in \cite{bauer_maxwell_2010}. The advantage in Wheeler and Feynman's derivation of the Lorentz-Dirac equation is that it bares no divergences  in the defining equations which provoke unphysical, so-called \emph{run-away}, solutions. At the same time Wheeler and Feynman's argument is able to explain the irreversible nature of radiation phenomena. These features make WF the most promising candidate for arriving at a mathematically well-defined theory of relativistic, classical electromagnetism.
 
However, mathematically WF is completely open. 
It is not an initial value problem for differential equations because the WF equations contain time-like advanced and retarded state-dependent arguments for which no theory of existence or uniqueness of solutions is available. Apart from two exceptions discussed below, it is not even known whether in general there are solutions at all.
 In tensor notation, WF is defined by: 
\begin{align}\label{eqn:cea wf}
  m\ddot z^\mu_i(\tau)=e_i \sum_{\stackrel{k=1,\ldots, N}{k\neq i}}\frac{1}{2}\left[F[z_k]_{+}^{\mu\nu}(z_i(\tau))+F[z_k]_{-}^{\mu\nu}(z_i(\tau))\right]\dot z_{i,\nu}(\tau), 
\end{align}
where
\begin{align}\label{eqn:cea lienard wiechert fields}
  F^{\mu\nu}=\partial^\mu A^{\nu}-\partial^\nu A^{\mu}, \qquad A[z_k]^\mu_{\pm}(x) &:= e_k\frac{\dot{ z}_k^\mu(\tau_{k,\pm}(x))}{(x- z_{k}(\tau_{k,\pm}(x)))_\nu\;\dot{ z}_k^\nu(\tau_{i,\pm}(x))},
\end{align}
and the world line parameters $\tau_{k,+},\tau_{k,-}:\bb M\to\bb R$ are implicitly defined through
\begin{align}\label{eqn:time ret}
   z_k^0(\tau_{k,+}(x)) = x^0+\|\vect x- \vect z_{k}(\tau_{k,+}(x))\|, \qquad
   z_k^0(\tau_{k,-}(x)) = x^0-\|\vect x- \vect z_{k}(\tau_{k,-}(x))\|.
\end{align}
Here, the world lines of the charges $z_i:\tau\mapsto z_i^\mu(\tau)$ for $1\leq i\leq N$  are parametrized by proper time $\tau\in\bb R$ and take values in Minkowski space $\bb M:=(\bb R^4,g)$ equipped with the metric tensor $g=\operatorname{diag}(1,-1,-1,-1)$.  We use Einstein's summation convention for Greek indices, i.e. $x_\mu y^\mu:=\sum_{\mu=0}^3g_{\mu\nu}x^\mu y^\nu$, and the notation $x=(x^0,\vect x)$ for an $x\in\bb M$ in order to distinguish the time component $x^0\in\bb R$ from the spatial components $\vect x\in\bb R^3$.  
The overset dot denotes a differentiation with respect to the world-line parametrization $\tau$. For simplicity each particle has the same inertial mass $m\neq 0$ (all presented results however hold for charges having different masses, too). The coupling constant $e_i$ denotes the charge of the $i$-th particle. 

If one were to insist on using field theoretic language then one may also say that
equations (\ref{eqn:cea wf}) describe the interaction between the charges via their advanced and retarded Li\'enard-Wiechert fields $F[z_k]_{+},F[z_k]_{-}$, $1\leq k\leq N$. These fields are special solutions of the Maxwell equations of classical electrodynamics corresponding to a prescribed world-line $z_k$. The functional dependence on $\tau\mapsto z_k(\tau)$ is emphasized by the square bracket notation $[z_k]$.
Given an $x\in\bb M$ and a time-like world line $\tau\mapsto z_k(\tau)$, i.e. one fulfilling $\dot z_{k,\mu}\dot z_{k}^\mu>0$, the solutions $\tau_{k,+}(x)$, $\tau_{k,-}(x)$, are unique and given by the intersection of the forward and backward light-cone of space-time point $x$ and the world-line $z_k$, respectively; see Figure \ref{fig:advance_delay}. The acceleration on the left-hand side of the WF equations depends through (\ref{eqn:time ret}) on time-like advanced as well as retarded data (with respect to $z_{i}^{0}(\tau)$) of all the other world lines; see Figure \ref{fig:wf ladder}. The delay is unbounded, and by (\ref{eqn:cea lienard wiechert fields})  the right-hand side of (\ref{eqn:cea wf}) again depends on the acceleration.

\begin{figure}[ht]
  \begin{center}
    \subfigure[\label{fig:advance_delay}]{
      \if\arxiv 1
        \includegraphics[scale=1]{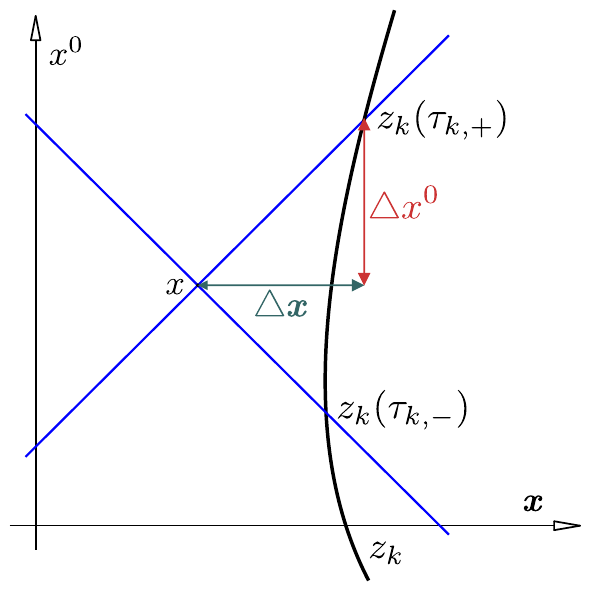}
      \else
        \includegraphics[scale=.7]{figures/advance_delay}
      \fi
    }
    \subfigure[\label{fig:wf ladder}]{
      \if\arxiv 1
        \includegraphics[scale=1]{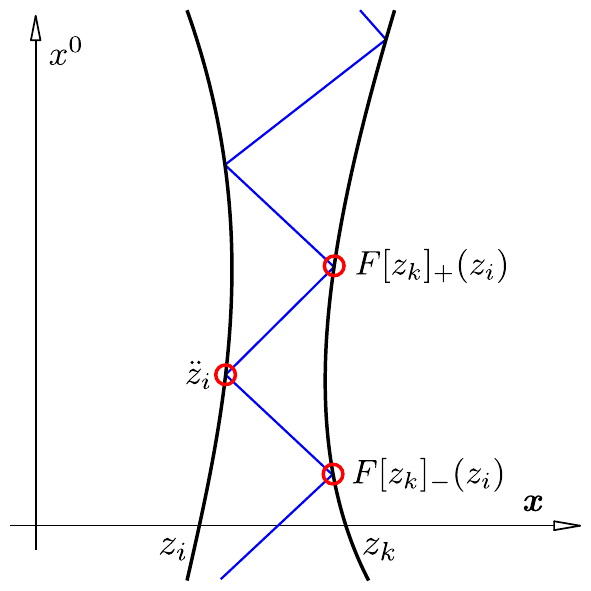}
      \else
        \includegraphics[scale=.7]{figures/wf_ladder}
       \fi
    }
    \caption{(a) Solutions of equations (\ref{eqn:time ret}) for $\triangle x^0:=z_k^0(\tau_{k,+}(x))-x^0$ and $\triangle x:=\|\vect x- \vect z_{k}(\tau_{k,+}(x))\|$. (b) Two WF world lines $z_i$ and $z_k$ interacting via a ladder of light-cones ($45^\circ$ lines since in our units the speed of light equals one). Hence, the value of $\ddot z_i$ depends on both advanced and retarded data $F[z_k]_{+}(z_i)$ and $F[z_k]_{+}(z_i)$, respectively.}
  \end{center}
\end{figure}
It is noteworthy that in 
early 1900 the mathematician and philosopher A.N. Whitehead developed a philosophical view on nature which rejects ``initial value problems''   as fundamental descriptions of nature \cite{whitehead_concept_1920}. He developed his own gravitational theory and motivated Synge's study of what is now referred to as Synge equations \cite{pauli_relativitltstheorie_1921,synge_electromagnetic_1940}, i.e.
\begin{align}\label{eqn:synge}
  m\ddot z^\mu_i(\tau)=e_i \sum_{\stackrel{k=1,\ldots, N}{k\neq i}}F[z_k]_{-}^{\mu\nu}(z_i(\tau))\dot z_{i,\nu}(\tau).
\end{align}
 The Synge equations share many difficulties with the WF equations but, as we shall show, are simpler to handle because they only depend on time-like retarded arguments. We would like to remark that independent of Whitehead's philosophy it seems to be the case that often fields are introduced to formulate a physical law, even though it may have a delay character, as initial value problem. Maxwell-Lorentz electrodynamics is a prime example. However, these very fields are then often the source of singularities of the theory, quantum or classical. Whitehead's idea might therefore point towards a fruitful new reflection about the character of physical laws.

 The books \cite{driver_ordinary_1977,diekmann_delay_1995} provide a beautiful overview on the topic of delay differential equations. However, for the WF equations as well as similar types of delay differential equations with advanced and retarded arguments of unbounded delay there are almost no mathematical results available. The problem  one usually deals with in the field of differential equations without delay is extension of  local solutions to a maximal domain and avoiding critical points by  introducing a notion of typicality of initial conditions. For WF the situation is dramatic. Because of the unbounded delay, the notion of local solutions does not make sense, so that the issue is not local versus global existence and also not explosion or running into singular points of the vector field.  The issue is simply: {\em Do solutions exist?} and  { \em What kind of data of the solutions is necessary and/or sufficient to characterize solutions uniquely?} 
 
 To put our work in perspective we call attention to the following literature: Angelov studied existence of Synge solutions in the case of two equal point-like charges and three dimensions \cite{angelov_synge_1990}. Under the assumption of an extra condition on the minimal distance between the charges to prevent collisions, he proved existence of Synge solutions on the positive time half-line. Uniqueness is only known in a special case in one dimension for two equal charges initially having sufficiently large opposite velocities and sufficiently large space-like separation. Under these conditions Driver has shown \cite{driver_backwards_1969} that the Synge solutions are uniquely characterized by initial positions and momenta. With regards to WF two types of special solutions are known to exist: First, the explicitly known Schild solutions \cite{schild_electromagnetic_1963} composed out of specially positioned charges revolving around each other on stable orbits, and second, the scattering solutions of two equal charges constrained on the straight line \cite{bauer_ein_1997}. The latter result rests on the fact that the asymptotic behavior of  world-lines on the straight line is well controllable (due to this special geometry the acceleration dependent term on the right-hand side of (\ref{eqn:cea wf}) vanishes). Uniqueness of WF solutions was proven in one dimension with zero initial velocity and sufficiently large separation of two equal charges  \cite{driver_canfuture_1979}. In a recent work \cite{de_luca_variational_2009} a well-defined analogue of the formal Fokker variational principle for two charges restricted to finite intervals was proposed. It is shown that its minima, if they exist, fulfill the WF equations on these finite times intervals. Furthermore, there are conjectures about uniqueness of WF solutions e.g. \cite{wheeler_classical_1949, van_dam_classical_1965, anderson_principles_1967, synge_relativity_1976}. While Driver's result \cite{driver_canfuture_1979} points to the possibility of uniqueness by initial positions and momenta, Bauer's \cite{bauer_ein_1997} work suggests to specify asymptotic positions and momenta. Furthermore, a WF toy model for two charges in three dimensions was given in \cite{deckert_electrodynamic_2010} for which a sufficient condition for a unique characterization of all its (sufficiently regular) solutions is the prescription of strips of time-like world lines long enough such that at least for one point on each strip the right-hand side of the WF equation is well-defined and the WF equation is fulfilled.

\section{Our Setup and Results}\label{sec:main}

 Our focus is on the bare existence of
solutions of WF, i.e. on the question: {\em Do solutions exist?}  For that question the issue that in a dynamical evolution of a system of point-like charges
catastrophic events may happen is secondary (compare the famous $n$-body problem of classical gravitation \cite{siegel_lectures_1971}).  More on target,  such 
considerations would have to invoke a notion of typicality of trajectories, so that catastrophic events can be shown to be atypical. 
But that would require not only existence of solutions but also a classification of solutions. We are far from that.
To avoid such issues at this early state of research we regard $\text{WF}_\varrho$ as introduced  in  \cite{bauer_maxwell_2010} instead of WF, i.e. we consider extended rigid charges described by the charge distributions $\varrho_{i}$, $1\leq i\leq N$, where singularities do not even occur when charges pass through each other.

For our mathematical analysis it is convenient  to express $\text{WF}_\varrho$  in coordinates where it takes the form
\begin{align}\label{eqn:WF equation written out}
  \begin{split}
    \partial_t\vect q_{i,t}&=\vect v(\vect p_{i,t}):=\frac{\vect p_{i,t}}{\sqrt{m^2+\vect p_{i,t}^2}}\\
    \partial_t\vect p_{i,t}&=\sum_{\stackrel{k=1,\ldots, N}{k\neq i}}\intdv x\varrho_i(\vect x-\vect q_{i,t})\left(\vect E_t[\vect q_k,\vect p_k](\vect x)+\vect v(\vect q_{i,t})\wedge\vect B_t[\vect q_k,\vect p_k](\vect x)\right)
  \end{split}
\end{align}
for $1\leq i\leq N$ and
\begin{align}\label{eqn:WF fields def}
  \begin{pmatrix}
    \vect E_t^{(e_{+},e_{-})}[\vect q_i,\vect p_i](\vect x)\\
    \vect B_t^{(e_{+},e_{-})}[\vect q_i,\vect p_i](\vect x)
  \end{pmatrix}=\sum_\pm 4\pi e_\pm \int ds\intdv y K_{t-s}^\pm(\vect x-\vect y)\begin{pmatrix}
    -\nabla\varrho_i(\vect y-\vect q_{i,s})-\partial_s\left[\vect v(\vect p_{i,s})\varrho_i(\vect y-\vect q_{i,s})\right]\\
    \nabla\wedge\left[\vect v(\vect p_{i,s})\varrho_i(\vect y-\vect q_{i,s})\right]
  \end{pmatrix}
\end{align}
where as in (\ref{eqn:WF equation written out}) most of the time we drop the superscript $^{(e_{+},e_{-})}$. Here, $K^\pm_t(\vect x):=\frac{\delta(\|\vect x\|\pm t)}{4\pi\|\vect x\|}$ are the advanced and retarded Green's functions of the d'Alembert operator. The partial derivative with respect to time $t$ is denoted by $\partial_t$, the gradient by $\nabla$, the divergence by $\nabla\cdot$, and the curl by $\nabla\wedge$.  At time $t$ the $i$th charge for $1\leq i\leq N$ is situated at position $\vect q_{i,t}$ in space $\bb R^3$, has momentum $\vect p_{i,t}\in\bb R^3$ and carries the classical mass $m\in\bb R\setminus\{0\}$. The geometry of the rigid charges are described by the smooth charge densities $\varrho_{i}$ of compact support , i.e. $\varrho_i\in\cal C^\infty_c(\bb R^3,\bb R)$, for $1\leq i\leq N$.

Using the notation $\vect E_{t}:=(F^{0i}(t,\cdot))_{1\leq i\leq 3}$ and $\vect B_{t}:=(F^{23}(t,\cdot), F^{31}(t,\cdot),F^{12}(t,\cdot))$ and replacing $\varrho_{i}$ by the three dimensional Dirac delta distribution $\delta^{(3)}$ times $e_{i}$ one retrieves from (\ref{eqn:WF equation written out}) the WF equations (\ref{eqn:cea wf}) for $e_+=\frac{1}{2}=e_-$ and the Synge equations (\ref{eqn:synge}) for $e_+=0,e_-=1$.  As discussed in Theorem \ref{thm:LWfields}, the expression (\ref{eqn:WF fields def}) for the choices for $e_+=1,e_-=0$ and $e_+=0,e_-=1$ is the advanced and retarded Li\'enard-Wiechert field, respectively. The square brackets $[\vect q_i,\vect p_i]$ emphasize that these fields are functionals of the charge trajectory $t\mapsto(\vect q_{i,t},\vect p_{i,t})$ and no dynamical degrees of freedom in their own.\\

The first idea to come to grips with existence of solutions is to adapt   fixed point arguments  from ordinary differential 
equations. That is not practical because of two difficulties.
The first difficulty is that in general in WF one cannot separate the second order derivative from lower order derivatives; see (\ref{eqn:LW_fields}), (\ref{eqn:LW_E_integrand}) and (\ref{eqn:LW_B_integrand}) for a more explicit expression of (\ref{eqn:WF fields def}).  Therefore one cannot rewrite the
WF equations in terms of an integral equation which is normally employed in the fixed points arguments.
The second difficulty is that the time-like advanced and retarded arguments   introduced by (\ref{eqn:WF fields def}) are of unbounded delay so that WF dynamics makes only sense for charge trajectories which are globally defined in time. One would thus have to find an appropriately normed space of functions on $\mathbb{R}$ on which the fixed point map can be defined -- which has not been found yet. 
One may circumvent this problem by introducing a  notion of conditional solution where outside a chosen time interval $[-T,T]$
the trajectories are prescribed by hand. The fixed point argument - if that could be formulated - would then run on the time interval $[-T,T]$ only. If successful, one may 
then try to construct a bonafide global solution by letting $T\to \infty$.  
In this work we show how one can formulate a fixed point procedure on intervals $[-T,T]$ for arbitrary large $T>0$, i.e. we show how one can circumvent
the first difficulty albeit gaining conditionally solutions only. The extension to global solutions would require good control
on the asymptotic behavior (as e.g. in \cite{bauer_ein_1997} in the case of the motion on the straight line), which we do not pursue here. We stress, however, that the extension to infinite time intervals is an interesting and worthwhile task, joining the results of this paper with the removal technique for $T\to\infty$  introduced in \cite{bauer_ein_1997}.

The key idea to define a fixed point map on time intervals $[-T,T]$ is a reformulation of the WF functional differential equations into a system of non-linear partial differential equations without delay, namely the Maxwell-Lorentz equations without self-interaction (abbrev. ML-SI) introduced in \cite[(4)-(7)]{bauer_maxwell_2010}\todo{check equations references}. Relying on the notation in \cite[(13)]{bauer_maxwell_2010} \todo{check equation references} the relation between WF and ML-SI can be expressed as an equality of sets of charge trajectories:
\begin{align}\label{eqn:crucial}
\text{WF}=\big\{\text{world lines of ML-SI}\upharpoonright\{F_0\equiv 0\} \big\}.
\end{align}
On the left-hand side we consider the set of trajectories of the charges that fulfill WF. On the right we have the set charge trajectories corresponding to solutions of ML-SI restricted to the subset for which there is no homogeneous field $F_0=F-\frac{1}{2}(F_{+}+F_{-})$, i.e. the actual fields $F$ coincide with the Wheeler-Feynman fields (\ref{eqn:WF fields def}). 

%

In the case of rigid charges, we shall use the relation (\ref{eqn:crucial}) in the following way:  Consider charge trajectories $t\mapsto(\vect q_{i,t},\vect p_{i,t})_{1\leq i\leq N}$  which solve $\text{WF}_\varrho$. By definition, the fields (\ref{eqn:WF fields def}) fulfill the Maxwell equations which implies that the map
\begin{align*}
 t\mapsto(\vect q_{i,t},\vect p_{i,t},\vect E_{i,t},\vect B_{i,t})_{1\leq i\leq N}:=(\vect q_{i,t},\vect p_{i,t},\vect E_t[\vect q_i,\vect p_i],\vect B_t[\vect q_i,\vect p_i])_{1\leq i\leq N}
\end{align*}
is  a solution of   $\text{ML-SI}_\varrho$, i.e. the \emph{Maxwell} equations:
\begin{align}\label{eqn:maxwell equations}
  \begin{split}
    \partial_t\vect E_{i,t} &= \nabla\wedge \vect B_{i,t} - 4\pi \vect v(\vect p_{i,t})\varrho_{i}(\cdot-\vect q_{i,t})\\
    \partial_t\vect B_{i,t} &= -\nabla\wedge \vect E_{i,t}
  \end{split}
  \begin{split}
    \nabla\cdot \vect E_{i,t} &= 4\pi \varrho_{i}(\cdot-\vect q_{t,i})\\
    \nabla\cdot\vect B_{i,t}&=0
  \end{split}
\end{align}
together with the \emph{Lorentz} equations (without self-interaction):
\begin{align}\label{eqn:lorentz force}
  \begin{split}
    \partial_t\vect q_{i,t} &= \vect v(\vect p_{i,t}) := \frac{\vect p_{i,t}}{\sqrt{m^2+\vect p_{i,t}^2}}\\
    \partial_t\vect p_{i,t} &= \sum_{\stackrel{k=1,\ldots, N}{k\neq i}}\intdv x\varrho_{i}(\vect x-\vect q_{i,t})\left[ \vect E_{k,t}(\vect x) + \vect v_{i,t} \wedge \vect B_{k,t}(\vect x) \right].
  \end{split}
\end{align}

On the other hand, global existence and uniqueness of solutions of  $\text{ML-SI}_\varrho$  for  initial  data $p:=(\vect q_i^0,\vect p_i^0)_{1\leq i\leq N}\in\bb R^{6N}$ and sufficiently regular initial fields $F:=(\vect E_i^0,\vect B_i^0)_{1\leq i\leq N}$, e.g. at time $t_0\in\bb R$, has been shown in \cite{bauer_maxwell_2010}; the needed definitions and results are summarized in the Section \ref{sec:summary of part I}. For any $(p,F)\in D_w(A^{\infty})$ the particular solution is then denoted by
\begin{align}\label{eqn:ML solutions}
t\mapsto M_L[p,F](t,t_0):=(\vect q_{i,t},\vect p_{i,t},\vect E_{i,t},\vect B_{i,t})_{1\leq i\leq N} .
\end{align}
In this sense we say that sufficiently regular $\text{WF}_{\varrho}$ charge trajectories give rise to $\text{ML-SI}_{\varrho}$ solutions.


Changing the point of view we now
 fix some Newtonian Cauchy data $p$ and ask our \\{\bf Crucial Question}: Do fields $F$ exist such that the corresponding  $\text{ML-SI}_\varrho$ solution $$t\mapsto(\vect q_{i,t},\vect p_{i,t},\vect E_{i,t},\vect B_{i,t})_{1\leq i\leq N}=:M_L[p,F](t,t_0)$$ fulfills
\begin{align}\label{eqn:special intial fields}
F=(\vect E_t[\vect q_i,\vect p_i],\vect B_t[\vect q_i,\vect p_i])_{1\leq i\leq N}|_{t=t_0}\quad?
\end{align}
Condition (\ref{eqn:special intial fields}) expresses that the initial fields $F$ equal the $\text{WF}_\varrho$ fields (\ref{eqn:WF fields def}) at initial time $t=t_0$. Equivalently, it ensures that the time-evolved fields $t\mapsto(\vect E_{i,t},\vect B_{i,t})_{1\leq i\leq N}$ of the $\text{ML-SI}_\varrho$  solution equal the $\text{WF}_\varrho$ fields $t\mapsto(\vect E_t[\vect q_i,\vect p_i],\vect B_t[\vect q_i,\vect p_i])_{1\leq i\leq N}$ for all times $t$ because their difference is a solution to the homogeneous Maxwell equations (i.e. (\ref{eqn:maxwell equations}) for $\varrho_i=0$) which is zero; compare (\ref{eqn:crucial}). Given the equality of fields for all times, equations (\ref{eqn:lorentz force}) turn into the $\text{WF}_\varrho$ equations (\ref{eqn:WF equation written out}), and hence, the charge trajectories of the  $\text{ML-SI}_\varrho$ solution fulfilling (\ref{eqn:special intial fields}) solve the $\text{WF}_\varrho$ equations. In other words, the subset of sufficiently regular solutions of $\text{ML-SI}_{\varrho}$ that correspond to initial conditions fulfilling (\ref{eqn:special intial fields}) have $\text{WF}_{\varrho}$ charge trajectories. We shall show that any once differentiable charge trajectory $t\mapsto(\vect p_{t},\vect q_{t})$ with bounded momenta and accelerations produces fields $\text{WF}_{\varrho}$ fields (\ref{eqn:WF fields def}) that are regular enough to serve as initial conditions for $\text{ML-SI}_{\varrho}$. This covers all physically interesting $\text{WF}_{\varrho}$ solutions, including the known Schild solutions. The advantage gained from this change of viewpoint is that $\text{ML-SI}_{\varrho}$ is given in terms of  an initial value problem. Therefore, instead of working directly with the $\text{WF}_{\varrho}$ functional equations it will be more convenient to formulate a fixed point procedure for $\text{ML-SI}_\varrho$ to find initial fields for which (\ref{eqn:special intial fields}) holds.


We now give an overview of our main results. Let $\cal T^{N}_{(e_+,e_-)}$ be the set of once differentiable charges trajectories $t\mapsto(\vect q_{i,t},\vect p_{i,t})_{1\leq i\leq N}$ having bounded momenta and accelerations and fulfilling $\text{WF}_{\varrho}$ equations (\ref{eqn:WF equation written out})-(\ref{eqn:WF fields def}); see Definition \ref{def:WF sols} below. Our first results is:
\begin{theorem}[Weak Uniqueness of Solutions]\label{thm:WF initial conditions}
For $e_{+},e_{-}\in\bb R$, $({\vect q}_i, {\vect p}_i)_{1\leq i\leq N}\in\cal T^{N}_{(e_+,e_-)}$ and $t\in\bb R$ we define
\begin{align}\label{eqn:WFsol}
 \varphi_t^{(e_+,e_-)}[({\vect q}_i, {\vect p}_i)_{1\leq i\leq N}]=(\vect q_{i,t},\vect p_{i,t},\vect E^{(e_+,e_-)}_{t}[\vect q_i,\vect p_i],\vect B^{(e_+,e_-)}_{t}[\vect q_i,\vect p_i])_{1\leq i\leq N}.
\end{align}
The following statements are true:
\begin{enumerate}[(i)]
 \item 
   For any $t_0\in\bb R$ we have $\varphi_{t_0}^{(e_+,e_-)}[({\vect q}_i, {\vect p}_i)_{1\leq i\leq N}]\in D_w(A^\infty)$.
   \item For all $t,t_0\in\bb R$ also $\varphi^{(e_+,e_-)}_t[({\vect q}_i, {\vect p}_i)_{1\leq i\leq N}]=M_L\left[\varphi_{t_0}^{(e_+,e_-)}[({\vect q}_i, {\vect p}_i)_{1\leq i\leq N}]\right](t,t_0)$ holds.
  \item For any $t_0\in\bb R$ the following map is injective:
  \begin{align*}
    i_{t_0}^{(e_+,e_-)}:\cal T^{N}_{(e_+,e_-)}\to D_w(A^\infty),\; ({\vect q}_i, {\vect p}_i)_{1\leq i\leq N}\mapsto\varphi^{(e_+,e_-)}_{t_0}[({\vect q}_i, {\vect p}_i)_{1\leq i\leq N}].
  \end{align*}
\end{enumerate}
\end{theorem}
Hence, for any choice of the coupling parameters $e_{+},e_{-}$ we know that: (i) The charge trajectories in $\cal T^{N}_{(e_+,e_-)}$ produce sufficiently regular initial fields for $\text{ML-SI}_{\varrho}$. (ii) The expression (\ref{eqn:WFsol}) coincides with a $\text{ML-SI}_{\varrho}$ solution. (iii) Each solutions of (\ref{eqn:WF equation written out})-(\ref{eqn:WF fields def}) can be identified by positions, momenta and fields $\vect E^{(e_+,e_-)}_{t}[\vect q_i,\vect p_i],\vect B^{(e_+,e_-)}_{t}[\vect q_i,\vect p_i])_{1\leq i\leq N}$ at an initial time $t_{0}$.

This gives us a good handle on the existence and uniqueness of the Synge solutions. We denote by $({\vect q}_i, {\vect p}_i)_{1\leq i\leq N}\in\cal T_{\text{\clock}!}^N(I)$ time-like and once differentiable charge trajectories $t\mapsto(\vect q_{i,t},\vect p_{i,t})_{1\leq i\leq N}$ with uniformly bounded momenta on the interval I, see Definition \ref{def:charge trajectory} below, and define further:
\begin{definition}[Synge Histories]\label{def:synge histories}
For $t_0\in\bb R$ we define the set $\mathfrak H(t_0)$ to consist of elements \ifarxiv{\linebreak}{}$({\vect q}_i, {\vect p}_i)_{1\leq i\leq N}\in \cal T_{\text{\clock}!}^N$ which fulfill:
  \begin{enumerate}[(i)]
    \item There exists an $\namel{amax}{a_{max}}<\infty$ such that for all $1\leq i\leq N$, $\sup_{t\in\bb R}\|\partial_t\vect v(\vect p_{i,t})\|\leq\namer{amax}$.
    \item $({\vect q}_i, {\vect p}_i)_{1\leq i\leq N}$ solve the equations (\ref{eqn:WF equation written out})-(\ref{eqn:WF fields def}) for $e_+=0, e_-=1$ at time $t=t_0$.
  \end{enumerate}
 Furthermore, 
$\mathfrak H(t_0)^+$ denotes the set $\mathfrak H(t_0)$ equipped with
\[
  (\vect q_{i},\vect p_{i})_{1\leq i\leq N}\overset{\mathfrak H^+(t_0)}{=}(\widetilde{\vect q}_{i},\widetilde{\vect p}_{i})_{1\leq i\leq N}:\Leftrightarrow\forall t\in[t_0,\infty):(\vect q_{i,t},\vect p_{i,t})_{1\leq i\leq N}=(\widetilde{\vect q}_{i,t},\widetilde{\vect p}_{i,t})_{1\leq i\leq N}
\]
while $\mathfrak H(t_0)^-$ denotes the set $\mathfrak H(t_0)$ equipped with
\[
  (\vect q_{i},\vect p_{i})_{1\leq i\leq N}\overset{\mathfrak H^-(t_0)}{=}(\widetilde{\vect q}_{i},\widetilde{\vect p}_{i})_{1\leq i\leq N}:\Leftrightarrow\forall t\in(-\infty,t_0]:(\vect q_{i,t},\vect p_{i,t})_{1\leq i\leq N}=(\widetilde{\vect q}_{i,t},\widetilde{\vect p}_{i,t})_{1\leq i\leq N}.
\]
\end{definition}
Given a history $(\vect q^{-}_{i},\vect p^{-}_{i})_{1\leq i\leq N}\in\mathfrak H^{-}(t_{0})$ one can simply compute the retarded Li\'enard-Wiechert fields $(\vect E_{t}^{(0,1)})[\vect q_{i}^{-},\vect p^{-}_{i}],\vect B_{t}^{(0,1)})[\vect q_{i}^{-},\vect p^{-}_{i}])_{1\leq i\leq N}$ at time $t=t_{0}$ and use them as initial fields for $\text{ML-SI}_{\varrho}$. The charge trajectories of the time-evolved $\text{ML-SI}_{\varrho}$ solutions then obey the Synge equations for times $t\geq t_{0}$. This way we shall prove:

\begin{theorem}[Existence and Uniqueness of Synge Solutions]\label{thm:exist and uni of synge}
Let $e_+=0, e_-=1$, $t_0\in\bb R$ and $(\vect q^-_{i},\vect p^-_{i})_{1\leq i\leq N}\in\mathfrak H(t_0)^-$.
\begin{enumerate}[(i)]
 \item (existence) There exists an extension $(\vect q^+_{i},\vect p^+_{i})_{1\leq i\leq N}\in\mathfrak H(t_0)^+$ such that the concatenation
\begin{align}\label{eqn:concat}
 (\vect q_{i},\vect p_{i})_{1\leq i\leq N}:t\mapsto(\vect q_{i,t},\vect p_{i,t})_{1\leq i\leq N}:=
\begin{cases}
 (\vect q^-_{i,t},\vect p^-_{i,t})_{1\leq i\leq N}\;\text{for}\;t\leq t_0\\
 (\vect q^+_{i,t},\vect p^+_{i,t})_{1\leq i\leq N}\;\text{for}\;t> t_0\\
\end{cases}
\end{align}
 is an element of $\cal T^N_{\text{\clock}!}((-\infty,T])$ for all $T\in\bb R$ and solves the equations (\ref{eqn:WF equation written out})-(\ref{eqn:WF fields def}) for all $t\geq t_0$.

 \item (uniqueness) Let $(\widetilde{\vect q}_{i},\widetilde{\vect p}_{i})_{1\leq i\leq N}\in\cal T^N_{\text{\clock}!}((-\infty,T])$ for any $T\in\bb R$ and suppose further that it solves the equations (\ref{eqn:WF equation written out})-(\ref{eqn:WF fields def}) for all times $t\geq t_{0}$. Then $(\widetilde{\vect q}_{i},\widetilde{\vect p}_{i})_{1\leq i\leq N}\overset{\mathfrak H^-(t_0)}{=}({\vect q}^-_{i},{\vect p}^-_{i})_{1\leq i\leq N}$ implies $(\widetilde{\vect q}_{i,t},\widetilde{\vect p}_{i,t})_{1\leq i\leq N}=({\vect q}_{i,t},{\vect p}_{i,t})_{1\leq i\leq N}$ for all $t\in\bb R$.
\end{enumerate}
\end{theorem}
Given Theorem \ref{thm:WF initial conditions} this existence and uniqueness result  is not hard to prove, and the reason for this is that we only ask for solutions on the half-line $[t_{0},\infty)$. In contrast to $\text{WF}_{\varrho}$, the notion of local solutions again makes sense since the histories simply act as prescribed external potentials. However, if we ask for solutions on whole $\bb R$ we again face the problem as in $\text{WF}_{\varrho}$, i.e. that by the unboundedness of the delay the notion of local solutions loses its meaning (a conceivable way around this without necessarily sacrificing uniqueness is to give initial conditions for $t_0\to\infty$).

We now come to the main part of this work where we discuss the existence of $\text{WF}_{\varrho}$ solutions. From now on we shall keep the choice $e_{+}=\frac{1}{2}$, $e_{-}+\frac{1}{2}$ fixed although all the results hold also for any choices of $0\leq e_{+},e_{-}\leq 1$. We take on the mentioned idea of conditional solutions: For given initial positions and momenta of the charges at $t=0$ we look for $\text{WF}_{\varrho}$ solutions on time intervals $[-T,T]$ for an arbitrary large but fixed $T>0$. To be able to regard only the time interval $[-T,T]$ of the $\text{WF}_{\varrho}$ dynamics we need to prescribe how the charge trajectories continue for times $|t|>T$ because due to the delay the dynamics within $[-T,T]$ will of course depend also on the trajectories at times $|t|>T$. This is done by specifying the advanced Li\'enard-Wiechert fields at time $T$ as well as the retarded Li\'enard-Wiechert fields and time $-T$ corresponding to each continuation of the charge trajectory for times $|t|>T$. We shall refer to these fields as \emph{boundary fields} and denote them by $X^+_{i,+T}$ and $X^-_{i,-T}$. 
The set of $\text{WF}_{\varrho}$ equations for $1\leq i\leq N$ with respect to the boundary fields $X^{\pm}_{i,\pm T}$ turn into 
\begin{align}\label{eqn:bWF equation written out}
  \begin{split}
    \partial_t\vect q_{i,t}&=\vect v(\vect p_{i,t}):=\frac{\vect p_{i,t}}{\sqrt{m^2+\vect p_{i,t}^2}}\\
    \partial_t\vect p_{i,t}&=\sum_{\stackrel{k=1,\ldots, N}{k\neq i}}\intdv x\varrho_i(\vect x-\vect q_{i,t})\left(\vect E^{X}_t[\vect q_k,\vect p_k](\vect x)+\vect v(\vect q_{i,t})\wedge\vect B^{X}_t[\vect q_k,\vect p_k](\vect x)\right)
  \end{split}
\end{align}
and
\begin{align}\label{eqn:WF with boundary fields}
       ({\vect E}^X_{i,t},
      {\vect B}^X_{i,t})
    = \frac{1}{2}\sum_\pm M_{\varrho_{i}}[X^\pm_{i,\pm T},(\vect q_i,\vect p_i)](t,\pm T),\;\text{for}\;1\leq i\leq N
\end{align}
where the $M_{\varrho}[F^{0},(\vect q,\vect p)](t,t_{0})$ denotes the solution of the Maxwell equations for initial fields $F^{0}$ at time $t_{0}$ corresponding to a prescribed trajectory $t\mapsto(\vect q_{t},\vect p_{t})$ with a charge distribution $\varrho$; see Definition \ref{def:Maxwell time evolution} below. Note that the above set of equations is a natural restriction of the $\text{WF}_{\varrho}$ dynamics onto the time interval $[-T,T]$ because: First, for the choice
\begin{align}\label{eqn:WF boundary fields}
  X^{\pm}_{i,\pm T}=4\pi \int ds\intdv y K_{\pm T-s}^\pm(\vect x-\vect y)\begin{pmatrix}
    -\nabla\varrho_i(\vect y-\vect q_{i,s})-\partial_s\left[\vect v(\vect p_{i,s})\varrho_i(\vect y-\vect q_{i,s})\right]\\
    \nabla\wedge\left[\vect v(\vect p_{i,s})\varrho_i(\vect y-\vect q_{i,s})\right]
  \end{pmatrix}
\end{align}
they turn into the $\text{WF}_{\varrho}$ set of equations (\ref{eqn:WF equation written out})-(\ref{eqn:WF fields def}). And second, it is well-known that for large $T$ the boundary fields are forgotten by the Maxwell time evolution $M$ in the point-wise sense; see Remark \ref{rem:mem loss} below. Based on this behavior one may expect to be able to study also unconditional existence  of $\text{WF}_{\varrho}$ solutions by considering the limit $T\to\infty$ for a convenient choice of controllable boundary fields.

For simplicity of our introductory discussion, let us choose
\begin{align*}
  X^{\pm}_{\pm T}&:=(\vect E^C_i(\cdot-\vect q_{i,\pm T}),0)_{1\leq i\leq N},\\
    (\vect E^C_{i},0)&:=M_{\varrho_{i}}[t\mapsto (0,0)](0,-\infty)=\left(\int d^3z\; \varrho_{i}(\cdot-\vect z)\frac{\vect z}{\|\vect z\|^3},0\right),
  \end{align*}
 i.e. the Coulomb fields corresponding to a charge at rest at $\vect q_{i,\pm T}$. With this prescription the conditional $\text{WF}_{\varrho}$ equations (\ref{eqn:bWF equation written out})-(\ref{eqn:WF with boundary fields}) are equivalent to $\text{WF}_{\varrho}$ dynamics for charges 
 initially being held at rest for times $t\leq -T$ and then instantaneously stopped at times $t\geq T$ by external mechanical forces; see Figure \ref{fig:wf}. The presented results, however, admit not only this particular case but a large class of boundary fields which also allow a continuous continuation of the momentum of the charges at times $t=\pm T$.

In view of our discussion of (\ref{eqn:crucial}) it seems natural to implement the following fixed point map in order to find solutions to the conditional $\text{WF}_{\varrho}$ equations (\ref{eqn:bWF equation written out})-(\ref{eqn:WF with boundary fields}) for initial positions and momenta $p\in\bb R^{6N}$ of the charges:\\

\noindent\textbf{INPUT:} $F=(\vect E_i^0,\vect B_i^0)_{1\leq i\leq N}$ for any fields such that $(p,F)\in D_w(A^\infty)$.
\begin{enumerate}[(i)]
  \item Compute the $\text{ML-SI}_{\varrho}$ solution $[-T,T]\ni t\mapsto (\vect q_{i,t},\vect p_{i,t},\vect E_{i,t},\vect B_{i,t})_{1\leq i\leq N}:=M_L[p,F](t,0)$.
  \item Compute the advanced and retarded fields 
    \begin{align*}
      (\widetilde{\vect E}_{i,t},
      \widetilde{\vect B}_{i,t})
    = \frac{1}{2}\sum_\pm M_{\varrho_{i}}[X^\pm_{i,\pm T}[p,F],(\vect q_i,\vect p_i)](t,\pm T)
  \end{align*}
  corresponding to the charge trajectories $t\mapsto(\vect q_{i},\vect p_{i})$ computed in (i) with prescribed initial fields $X^\pm_{i,\pm T}[p,F]$ at times $\pm T$. 
\end{enumerate}
\textbf{OUTPUT:} $S^{p,X^\pm}_T[F]:=(\widetilde{\vect E}_{i,t},\widetilde{\vect B}_{i,t})_{1\leq i\leq N}|_{t=0}$.\\

Note that the boundary fields $X^\pm_{i,\pm T}=X^\pm_{i,\pm T}[p,F]$ need to depend on the $\text{ML-SI}_{\varrho}$ initial values $(p,F)$. Otherwise, it would not be possible to continuously connect the charge trajectories with the prescribed continuation of the charge trajectories at times $t=\pm T$. The definition of $S_T^{p,X^\pm}$ is given in Definition \ref{def:STX} below. By construction, any fixed point $F^*$ of this map $S^{p,X^{\pm}}_{T}$ gives rise to a $\text{ML-SI}_{\varrho}$ solutions $t\mapsto M_L[p,F^*](t,0)$ whose charge trajectories fulfill the conditional $\text{WF}_{\varrho}$ equations (\ref{eqn:bWF equation written out})-(\ref{eqn:WF with boundary fields}); see Definition \ref{def:WF sol for finite times} and Theorem \ref{thm:the map ST} below. We prove:

\begin{theorem}[Existence of Conditional $\text{WF}_{\varrho}$ Solutions]\label{thm:ST has a fixed point}
  Let $p\in\bb R^{6N}$ be given. For each finite $T>0$ the map $S_T^{p,X^\pm}$ has a fixed point.
\end{theorem}

The essential ingredient in the proof of this result is the good nature of the $\text{ML-SI}_{\varrho}$ dynamics which implies Lemma \ref{lem:estimates for ST} below. Here we rely heavily on the work done in \cite{bauer_maxwell_2010}. 

We close with a discussion of these fixed points. Recall that the Synge solutions on the time half-line $[t_0,\infty)$ for times sufficiently close to $t_0$ give  rise to interaction with the given past trajectories on $(-\infty,t_0]$ only. For such small times one simply solves an external field problem. Not until larger times the interaction becomes truly retarded in the sense that the future charge trajectories interact with their just generated histories for times $t\geq t_{0}$. However, in an extreme situation a charge could approach the speed of light so fast  that the time coordinate of the intersection of its backward light-cone with another charge trajectory is bounded by, say, by $T^{max}\in\bb R$. This means that this charge will never interact with the part $t\geq T^{max}$ of the other charge trajectories. If $T^{max}\leq t_0$ one ends up solving a purely external field problem without seeing any truly retarded interaction. Such a scenario is of course so special that one would not expect it for all Synge solutions (recall that by Theorem \ref{thm:exist and uni of synge} one has existence and uniqueness on the time half-line for any sufficiently regular set of past trajectories). For the $\text{WF}_\varrho$ equations, however, we only have solutions on time intervals $[-T,T]$ yet and, therefore, one should be more curious as the described scenario in the case of the Synge equations could happen in the case of the WF equations in the past as well as the in future of $t_{0}$. If the $\text{WF}_{\varrho}$ solution on $[-T,T]$ behaves as badly as described above or the initial position are too far apart from each other in the space-like sense, we might end up solving only an external field problem as the charge trajectories on $[-T,T]$ only ``see'' the prescribed boundary fields; see Figure \ref{fig:wf extreme}. The following result makes sure that given $T$ at least for some solutions this is not the case because on an interval $[-L,L]$ with $0<L\leq T$ they interact exclusively with all other charge trajectories on $[-T,T]$ and not with the given boundary fields; i.e. the case as shown in Figure \ref{fig:wf}. 

\begin{figure}[h]
  \begin{center}
    \subfigure[\label{fig:wf}]{
      \if\arxiv 1
         \includegraphics[scale=1]{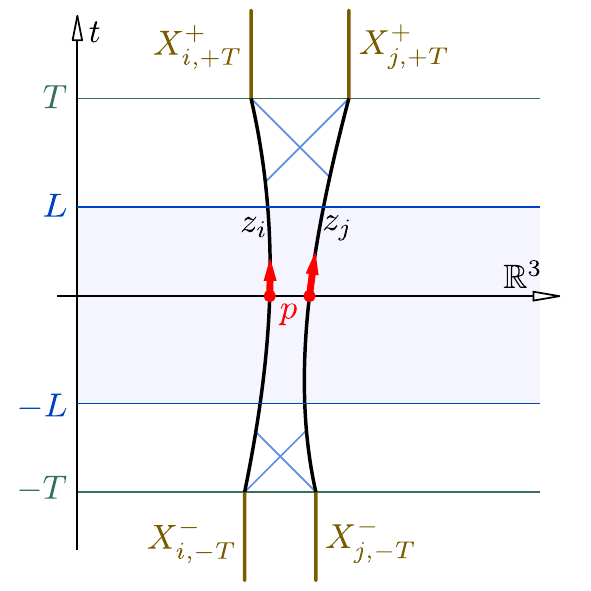}
       \else
         \includegraphics[scale=.7]{figures/wf}
       \fi
    }
    \subfigure[\label{fig:wf extreme}]{
      \if\arxiv 1
        \includegraphics[scale=1]{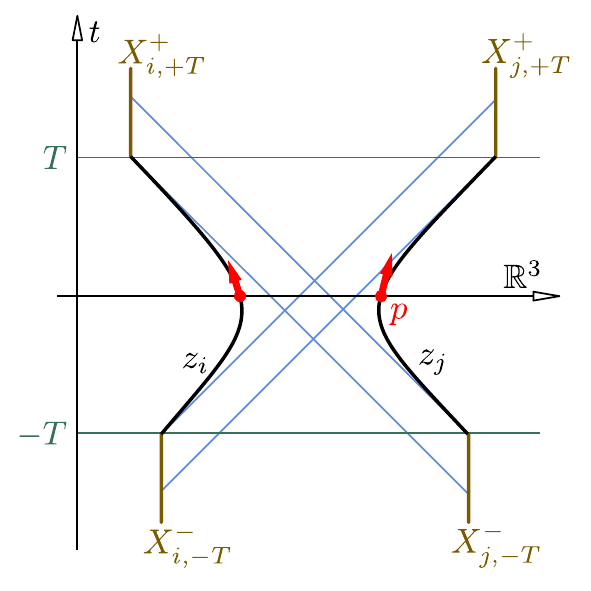}
	  \else
        \includegraphics[scale=.7]{figures/wf_extreme}
	  \fi
    }
     \caption{Two WF world lines $z_i$ and $z_j$ on time interval 
    $[-T,T]$ with Newtonian Cauchy data $p$. The straight lines for times $|t|>T$ are the prescribed asymptotes which generate the advanced and retarded Li\'enard-Wiechert fields $X^+_{i,+T}$ and $X^-_{i,-T}$. In (a) one observes true WF interaction between the charge trajectories on $[-T,T]$ within the time interval $[-L,L]$. In the extreme case (b) the charge trajectories on $[-T,T]$ interact only with  the given asymptotes (apart from the connection conditions at $\pm T$).}
  \end{center}
\end{figure}

We prove:
\begin{theorem}[True $\text{WF}_{\varrho}$ Interaction]\label{thm:existence of L}
  Choose $a,b,T>0$. Then:
\begin{enumerate}[(i)]
 \item The absolute values of the velocities $\vect v(\vect p_{i,t})$ of all charges of any $\text{ML-SI}_{\varrho}$ solution with any initial data $(p,F)$ such that
\begin{align*}
   \|p\|\leq a,\qquad \max_{1\leq i\leq N}\|\varrho_{i}\|_{L^2_w}+\max_{1\leq i\leq N}\|w^{-1/2}\varrho_{i}\|_{L^2}\leq b,\qquad F\in\Ran S^{p,X^\pm}_T
\end{align*}
 have an upper bound $v^{a,b}_T$ with $0\leq v^{a,b}_T<1$.

\item Let $R>0$ be the smallest radius such that the support of $\varrho_{i}$ lies within a ball around the origin with radius $R$, i.e. $\supp\varrho_{i}\subseteq B_{R}(0)$, for all $1\leq i\leq N$, and further $\namel{qmax}{\triangle q_{max}}(p):=\max_{1\leq i,j\leq N}\|\vect q_i^0-\vect q_j^0\|$. For sufficiently small $R$ there exist $p=(\vect q_i^0,\vect p_i^0)_{1\leq i\leq N}$ such that
\begin{align}\label{eqn:L}
 L:=\frac{(1-v^{a,b}_T)T-\namer{qmax}-2R}{1+v_T^{a,b}}>0
\end{align}
and any fixed point $F^*$ of $S^{p,X^\pm}_T$ gives rise to a $\text{ML-SI}_{\varrho}$ solution $t\mapsto M_L[p,F^*](t,0)$ whose charge trajectories for $t\in[-L,L]$ solve the $\text{WF}_{\varrho}$ equations (\ref{eqn:WF equation written out})-(\ref{eqn:WF fields def}).
\end{enumerate}
\end{theorem}

The form of $L$ in (\ref{eqn:L}) is a direct consequence of the geometry as displayed in Figure \ref{fig:wf} and the nature of the free Maxwell time evolution, see Lemma \ref{lem:shadows of boundary fields}, which can be seen from a direct computation using harmonic analysis. The proof further employs a very rough Gr\"onwall estimate coming from the $\text{ML-SI}_{\varrho}$ dynamics to estimate the velocities of the charges during the time interval $[-T,T]$, see Lemma \ref{lem:uni vel bound}. The conditions for the above result are therefore quite restrictive but merely technical. Any uniform velocity estimate, e.g. as given in \cite{bauer_ein_1997} for two charges of equal sign restricted to a straight line, makes this result redundant as then $T$ can just be chosen arbitrarily large to ensure an arbitrary large $L$, and hence charge trajectories that fulfill the $\text{WF}_{\varrho}$ equations on arbitrary large intervals. We expect such a bound also without the restriction to a straight line. However, even without such a uniform velocity bound the result above already ensures that in Theorem \ref{thm:ST has a fixed point} we do see truly advanced and retarded $\text{WF}_{\varrho}$ interaction between the charges. Furthermore, we remark that for the charge trajectories found in (ii) above one can already define the WF conservation laws \cite{wheeler_classical_1949} which we expect to be an important ingredient in order to control a limit procedure $T\to\infty$ to yield global $\text{WF}_{\varrho}$ solutions.

\section{Preliminaries}

In the proofs of the main results we will frequently rely on explicit expressions of the time-evolved electric and magnetic fields appearing in the Maxwell equations as well as in the $\text{ML-SI}_{\varrho}$ time evolution. The $\text{ML-SI}_{\varrho}$ equations are (\ref{eqn:maxwell equations})-(\ref{eqn:lorentz force}) while the Maxwell equations for a given charge-current density $t\mapsto(\rho_t,\vect j_t)$ have the form
\begin{align}\label{eqn:maxwell_equations_charge_current}
  \begin{split}
    \partial_{t}{\vect E}_{t} &= \nabla\wedge \vect B_{t} - 4\pi \vect j_t\\
    \partial_{t}{ \vect B}_{t} &= -\nabla\wedge \vect E_{t}
  \end{split}
  \begin{split}
    \nabla\cdot \vect E_{t} &= 4\pi \rho_t\\
    \nabla\cdot \vect B_{t} &= 0.
  \end{split}
\end{align}

 Although the presented results on the Maxwell equations are well-known in the physics community, we only found some of them in the mathematical literature; e.g. \cite{spohn_dynamics_2004}. Therefore, we give a mathematical review in Section \ref{sec:maxwellsolutions}. The proofs of all the claims are published separately in \cite{deckert_electrodynamic_2010}. Furthermore, $\text{ML-SI}_{\varrho}$ was studied in \cite{bauer_maxwell_2010}. In order to be self-contained we given an overview of the needed results in Section \ref{sec:summary of part I}.
 
 \paragraph{Notation.} Let $\bb N_{0}=\bb N\cup\{0\}$. $\bb R^{3}$ vectors and vector-valued functions have bold letters. We denote the ball of radius $R>0$ around the center $\vect x\in\bb R^{3}$ by $B_{R}(\vect x)\subset\bb R^{3}$ and its boundary by $\partial B_{R}(\vect x)$. We denote by $C\in\bounds$ any function $x\mapsto C(x)\in\bb R_+$ that depends continuously and non-decreasing on its argument $x$. Furthermore, let $\cal C^{n}(V,W)$ be the set of $n$-times continuously differentiable functions $V\to W$. $\cal C^{\infty}(V,W):=\bigcap_{n\in\bb N_{0}}\cal C^{n}(V,W)$. $\cal C^{n}_{c}(V,W)\subset\cal C^{n}(V,W)$ and $\cal C^{\infty}_{c}(V,W)\subset\cal C^{\infty}(V,W)$ are the respective subsets of functions with compact support. Where unambiguous we sometime drop the reference to $V$ and $W$.

\subsection{Strong Solutions to the Maxwell Equations}\label{sec:maxwellsolutions}
 
 We review the solution theory of the Maxwell equations (\ref{eqn:maxwell_equations_charge_current}).  The class of charge-current densities we treat is defined by:

\begin{definition}[Charge-Current Densities]\label{def:charge-current}\label{def:induced charge-current}
  We shall call any pair of maps $\rho:\bb R\times\bb R^3\to\bb R, (t,\vect x)\mapsto \rho_t(\vect x)$ and $\vect j:\bb R\times\bb R^3\to\bb R^3, (t,\vect x)\mapsto \vect j_t(\vect x)$ a charge-current density whenever:
  \begin{enumerate}[(i)]
    \item For all $\vect x\in\bb R^3$: $\rho_{(\cdot)}(\vect x)\in \cal C^1(\bb R,\bb R)$ and $\vect j_{(\cdot)}(\vect x)\in \cal C^1(\bb R,\bb R^3)$.
    \item For all $t\in\bb R$: $\rho_t,\partial_t\rho_t\in\cal C^\infty(\bb R^3,\bb R)$ and $\vect j_t,\partial_t\vect j_t\in\cal C^\infty(\bb R^3,\bb R^3)$.
    \item For all $(t,\vect x)\in\bb R\times\bb R^3$: $\partial_t\rho_t(\vect x)+\nabla\cdot\vect j_t(\vect x)=0$ which is referred to as continuity equation.
  \end{enumerate}
  We denote the set of such pairs $(\rho,\vect j)$ by $\cal D$.
\end{definition}

 We are interested in solutions to the Maxwell equations (\ref{eqn:maxwell_equations_charge_current}) in the following sense:
\begin{definition}[Strong Solution Sense]\label{def:initial fields F}
 We define the space of fields
\begin{align*}
 \cal F^1:=\cal C^\infty(\bb R^3,\bb R^3)\oplus\cal C^\infty(\bb R^3,\bb R^3).
\end{align*}
Let $t_0\in\bb R$ and $F^0\in\cal F^1$. Then any mapping $F:\bb R\to\cal F^1, t\mapsto F_t:=(\vect E_t,\vect B_t)$ that solves (\ref{eqn:maxwell_equations_charge_current}) in the point-wise sense for initial value $F_t|_{t=t_0}=F^0$ is called a strong solution to the Maxwell equations with $t_0$ initial value $F^0$.
\end{definition}
 Explicit formulas of those solutions are constructed with the help of:
\begin{definition}[Green's Functions of the d'Alembert]\label{def:greens_dalembert}
  We set
  \[
    K^\pm_t(\vect x):=\frac{\delta(\|\vect x\|\pm t)}{4\pi\|\vect x\|}
  \]
  where $\delta$ denotes the one-dimensional Dirac delta distribution.
\end{definition}
Note that for every $f\in\cal C^\infty(\bb R^3)$
  \[
    K_t^\pm*f(\vect x) = \begin{cases}
      0 & \text{ for }\pm t>0\\
      t\underset{\partial {B_{|t|}(\vect x)}}{\fint}d\sigma(y)f(\vect y):=\frac{t}{4\pi t^2}\int_{\partial {B_{|t|}}(\vect x)}d\sigma(y)f(\vect y) & \text{otherwise}
    \end{cases}
  \]
  holds, where $d\sigma$ denotes the surface element on $\partial B_{|t|}(\vect x)$. We introduce the notation $\triangle=\nabla\cdot\nabla$ and $\square=\partial_t^2-\triangle$.
  
\begin{lemma}[Green's Functions Properties]\label{lem:Greens_function_dalembert}
  Let $f\in \cal C^{\infty}(\bb R^3)$. Then:
  \begin{enumerate}[(i)]
    \item The following identities holds:
    \begin{align}
    \begin{split}
    \label{eqn:dtKt_f}
        K^\pm_t*f&=\mp t \underset{\partial B_{\mp t}(0)}\fint d\sigma(y) f(\cdot-\vect y)\\
    \partial_t K^\pm_t*f&= \mp \underset{\partial{B_{\mp t}}(0)}\fint d\sigma(y)\; f(\cdot-\vect y) \mp \frac{t^2}{3}\underset{{B_{\mp t}}(0)}\fint d^3y\; \triangle f(\cdot-\vect y)\\
    \partial_t^2 K_t^\pm*f&= K_t^\pm*\triangle f = \triangle K_t^\pm* f.
    \end{split}
  \end{align}
    \item Set $K_t=\sum_\pm \mp K_t^\pm$. The mapping $(t,\vect x)\mapsto [K_t*f ](\vect x)$ can uniquely be extended at $t=0$ to become a $C^\infty(\bb R\times\bb R^3)$ function such that for all $n\in\bb N$
\begin{align}\label{eqn:greens_dalembert_limits}
        \lim_{t\to 0\mp} \begin{pmatrix}
          \partial_t^{2n} K_t*f\\
          \partial_t^{2n+1} K_t*f
        \end{pmatrix}
        = \begin{pmatrix}
          0\\
          \triangle^n f
        \end{pmatrix}
      \end{align}
    and $\square K_t*f=0$ for all $t\in\bb R$.
  \end{enumerate}
\end{lemma}
\begin{remark}\label{rem:waveequation}
  In the future we will denote the unique extension of $K_t$ by the same symbol $K_t$. It is called the \emph{propagator} of the homogeneous wave equation.
\end{remark}
\if\arxiv 1

A direct consequence of this lemma is Kirchoff's formula:
\begin{corollary}[Kirchoff's Formula]\label{cor:Kirchoff}
  Let $A^0,\dot A^0\in\cal C^\infty(\bb R^3)$. The mapping $t\mapsto A_t$ defined by 
  \begin{align}\label{eqn:Kirchoff}
    A_t=\partial_t K_t*A^0+K_t*\dot{A}^0
  \end{align}
  solves the homogeneous wave equation $\square A_t=0$ in the strong sense and for initial values  $A_t|_{t=0}=A^0$ and $\partial_t A_t|_{t=0}=A^0$.
\end{corollary}
We construct explicit solutions of the Maxwell equations along the following line of thought: In the distribution sense every solution to the Maxwell equations (\ref{eqn:maxwell_equations_charge_current}) is also a solution to
\begin{align*}
  \square \begin{pmatrix}
    \vect E_t\\
    \vect B_t
  \end{pmatrix} = 4\pi\begin{pmatrix}
    -\nabla \rho_t - \partial_t \vect j_t\\
    \nabla\wedge \vect j_t
  \end{pmatrix}
\end{align*}
having initial values
\begin{align}\label{eqn:wave_equations_initial_values}
    (\vect E_t,\vect B_t)\big|_{t=t_0}&=(\vect E^0,\vect B^0),  &&
  \partial_t(\vect E_t,\vect B_t)\big|_{t=t_0}=(\nabla\wedge\vect B^0-4\pi \vect j_{t_0},-\nabla\wedge\vect E^0).
\end{align}
To make formulas more compact we sometimes abbreviate the pair of electric and magnetic fields in the form $F_t=(\vect E_t,\vect B_t)$ and let operators act thereon component-wisely. With the help of the Green's functions from Definition \ref{def:greens_dalembert} one may guess the general form of any solution to these equations:
\begin{align}\label{eqn:maxwell_solution}
  F_t = F^{\mathit{hom}}_t + \int_{-\infty}^{\infty} ds\; K^\pm_{t-t_0-s} * \begin{pmatrix}
    -\nabla \rho_{t_0+s} - \partial_s \vect j_{t_0+s}\\
    \nabla\wedge \vect j_{t_0+s}
  \end{pmatrix}
\end{align}
where $F^{\mathit{hom}}_t$ is a solution of the homogeneous wave equation, i.e. $\square F^{\mathit{hom}}_t=0$. Considering the forward as well as backward time evolution we regard two different kinds of initial value problems:
\begin{enumerate}[(i)]
  \item Initial fields $F^0$ are given at some time $t_0\in\bb R\cup\{-\infty\}$ and propagated to a time $t>t_0$.
  \item Initial fields $F^0$ are given at some time $t_0\in\bb R\cup\{+\infty\}$ and propagated to a time $t<t_0$.
\end{enumerate}
The kind of initial value problem posed will then determine $F^{\mathit{hom}}_t$ and the corresponding Green's function $K_t^{\pm}$. For (i) we use $K_t^-$ and for (ii) we use $K_t^+$ which are uniquely determined by $\square K_t^\pm(\vect x)=\delta(t)\delta^3(\vect x)$ and $K_t^\pm = 0 \text{ for } \pm t>0$. Furthermore, in the case of time-like charge trajectories and $\mp (t-t_0)>0$ Lemma \ref{lem:Greens_function_dalembert} implies
\[
  \square\int_{\pm \infty}^0 ds\; K^{\pm}_{t-t_0-s}*\begin{pmatrix}
    -\nabla \rho_{t_0+s} - \partial_s \vect j_{t_0+s}\\
    \nabla\wedge \vect j_{t_0+s}
  \end{pmatrix}=\int_{\pm \infty}^0 ds\; \square K^{\pm}_{t-t_0-s}*\begin{pmatrix}
    -\nabla \rho_{t_0+s} - \partial_s \vect j_{t_0+s}\\
    \nabla\wedge \vect j_{t_0+s}
  \end{pmatrix}=0.
\]
Terms of this kind can simply be included in the homogeneous part of the solution $F^{\mathit{hom}}_t$. This way we arrive at two solution formulas. One being suitable for our forward initial value problem, i.e. $t-t_0>0$,
\[
  F_t = F_t^{\mathit{hom}} + 4\pi\int_{0}^{t-t_0} ds\; K^{-}_{t-t_0-s} * \begin{pmatrix}
    -\nabla \rho_{t_0+s} - \partial_s \vect j_{t_0+s}\\
    \nabla\wedge \vect j_{t_0+s}
  \end{pmatrix},
\]
and the other suitable for the backward initial value problem, i.e. $t-t_0<0$,
\[
  F_t = F_t^{\mathit{hom}} + 4\pi\int_{t-t_0}^{0} ds\; K^{+}_{t-t_0-s} * \begin{pmatrix}
    -\nabla \rho_{t_0+s} - \partial_s \vect j_{t_0+s}\\
    \nabla\wedge \vect j_{t_0+s}
  \end{pmatrix}.
\]
As a last step one needs to identify the homogeneous solutions which satisfy the given initial conditions (\ref{eqn:wave_equations_initial_values}). Corollary \ref{cor:Kirchoff} provides the explicit formula:
\begin{align*}
  F^{\mathit{hom}}_t:=\begin{pmatrix}
    \partial_t & \nabla\wedge\\
    -\nabla\wedge & \partial_t
  \end{pmatrix}
  K_{t-t_0}*F^0=
\begin{pmatrix}
 \partial_t K_{t-t_0}*\vect E^0+\nabla\wedge K_{t-t_0}*\vect B^0\\
 -\nabla\wedge K_{t-t_0}*\vect E^0+\partial_t K_{t-t_0}*\vect B^0\\
\end{pmatrix}.
\end{align*}
Therefore, using the propagator $K_t$ and a substitution in the integration variable both initial value problems fulfill for all $t\in\bb R$:
\begin{align*}
  F_t &= \begin{pmatrix}
    \partial_t & \nabla\wedge\\
    -\nabla\wedge & \partial_t
  \end{pmatrix}
  K_{t-t_0}*F^0
  + K_{t-t_0}*\begin{pmatrix}
    -4\pi \vect j_{t_0}\\
    0
  \end{pmatrix}
  + 4\pi \int_{t_0}^{t} ds\; K_{t-s} * \begin{pmatrix}
    -\nabla \rho_{s} - \partial_s \vect j_{s}\\
    \nabla\wedge \vect j_{s}
  \end{pmatrix}.
\end{align*}
\fi

\begin{theorem}[Maxwell Solutions]\label{thm:maxwell_solutions} Let $(\rho,\vect j)\in\cal D$ .
\begin{enumerate}[(i)]
 \item Given $(\vect E^0,\vect B^0)\in\cal F^1$ fulfilling the Maxwell constraints $\nabla\cdot\vect E^0=4\pi\rho_{t_0}$ and $\nabla\cdot\vect B^0=0$ for any $t_0\in\bb R$, the mapping $t\mapsto F_t=(\vect E_t,\vect B_t)$ defined by
  \begin{align*}
    \begin{pmatrix}
      \vect E_t\\
      \vect B_t
    \end{pmatrix}:=
    \begin{pmatrix}
    \partial_t & \nabla\wedge\\
    -\nabla\wedge & \partial_t
    \end{pmatrix}
    K_{t-t_0}*\begin{pmatrix}
      \vect E^0\\
      \vect B^0
    \end{pmatrix}
    + K_{t-t_0}*\begin{pmatrix}
      -4\pi \vect j_{t_0}\\
      0
    \end{pmatrix}
    + 4\pi \int_{t_0}^{t} ds\; K_{t-s} * \begin{pmatrix}
      -\nabla & - \partial_s \\
      0 & \nabla\wedge
    \end{pmatrix}
    \begin{pmatrix}
      \rho_s\\
      \vect j_{s}
    \end{pmatrix}
  \end{align*}
  for all $t\in\bb R$ is $\cal F^1$ valued, infinitely often differentiable and a solution to (\ref{eqn:maxwell_equations_charge_current}) with $t_0$ initial value $F^0$.
  
  \item For all $t\in\bb R$ we have $\nabla\cdot\vect E_t=4\pi\rho_{t}$ and $\nabla\cdot\vect B_{t}=0$.


\end{enumerate}
\end{theorem}

\begin{remark}
  Clearly, one needs less regularity of the initial values in order to get a strong solution. With regard to $\text{WF}_{\varrho}$, however, we will only need to consider smooth initial values $\cal F^1$. The explicit formula of the solutions (after an additional partial integration) was already found in \cite{komech_longtime_2000}[(A.24),(A.25)] where it was derived with the help of the Fourier transform (there seems to be a misprint in equation (A.24). However, (A.20) from which it is derived is correct).
\end{remark}

For the rest of this paper the charge-current densities $(\rho,\vect j)$ we will consider are the ones generated by a moving rigid charge on time-like trajectories:
\begin{definition}[Charge Trajectories]\label{def:charge trajectory}
\mbox{}\\
\begin{enumerate}[(i)]
 
 \item We call any map
  \begin{align*}
    (\vect q,\vect p)\in\cal C^1(\bb R,\bb R^3\times\bb R^3),\; t\mapsto(\vect q_t,\vect p_t)
  \end{align*}
  a charge trajectory and denote with $\vect q_t$ and $\vect p_t$ the position and momentum of the charge, respectively. Its velocity at time $t$ is given by $\vect v(\vect p_{t}):=\frac{\vect p_{t}}{\sqrt{m^2+\vect p_{t}^2}}$.
\setcounter{enumi}{1}
  \item We collect all time-like charge trajectories in the set
  \begin{align*}
    \cal T_{\text{\clock}}^1 &:= \bigg\{(\vect q,\vect p)\in\cal C^1(\bb R,\bb R^3\times\bb R^3)\;\bigg|\;\left\|\vect v(\vect p_t)\right\|<1\;\text{for all}\;t\in\bb R\bigg\},
  \end{align*}
  
  \item and all strictly time-like charge trajectories in the set
  \begin{align*}
    \cal T_{\text{\clock}!}^1(I) &:= \bigg\{(\vect q,\vect p)\in\cal T_{\text{\clock}}^1\;\bigg|\;\exists\namel{vmax}{v_{max}}<1\;\text{such that}\;\sup_{t\in I}\left\|\vect v(\vect p_t)\right\|\leq\namer{vmax}\bigg\}.
  \end{align*}
  where we use the abbreviation $\cal T_{\text{\clock}!}^1:=\cal T_{\text{\clock}!}^1(\bb R)$.
\end{enumerate}
Furthermore, we use the notation
\[
  (\vect q,\vect p)=(\widetilde{\vect q},\widetilde{\vect p}):\Leftrightarrow\forall t\in\bb R:(\vect q_t,\vect p_t)=(\widetilde{\vect q}_t,\widetilde{\vect p}_t)
\]
and define the Cartesian products $\cal T_{\text{\clock}}^{N}:=(\cal T_{\text{\clock}}^{1})^{N}$ and $\cal T_{\text{\clock}!}^{N}:=(\cal T_{\text{\clock}!}^{1})^{N}$.
\end{definition}

\begin{definition}[The Charge-Current Density of a Charge Trajectory]
 For $\varrho\in\cal C^\infty_c(\bb R^3,\bb R)$ and $(\vect q,\vect p)\in\cal T_{\text{\clock}}^1$ define 
  \begin{align*}
    \rho_t(\vect x):=\varrho(\vect x-\vect q_t), && \vect j_t(\vect x):=\frac{\vect p_t}{\sqrt{m^2+\vect p_t^2}}\varrho(\vect x-\vect q_t)
  \end{align*}
  for all $(t,\vect x)\in\bb R\times\bb R^3$ which we call the induced charge-current density of $(\vect q,\vect p)$.
\end{definition}

Clearly, $(\rho,\vect j)\in\cal D$ so that Theorem \ref{thm:maxwell_solutions} applies:
\begin{definition}[Maxwell Time Evolution]\label{def:Maxwell time evolution}
  Given a charge trajectory $(\vect q,\vect p)\in\cal T_{\text{\clock}}^1$ which induces $(\rho,\vect j)\in\cal D$ we denote the solution $t\mapsto F_t$ of the Maxwell equations (\ref{eqn:maxwell_equations_charge_current}) given by Theorem \ref{thm:maxwell_solutions} and $t_0$ initial values $F^0=(\vect E^0,\vect B^0)\in\cal F^1$ by
  \[
    t\mapsto M_{\varrho}[F^0,(\vect q,\vect p)](t,t_0):=F_t.
  \]
\end{definition}
One finds the following special solutions:
\begin{theorem}[Li\'enard-Wiechert Fields]\label{thm:LWfields}
  Let $F^0=(\vect E^0,\vect B^0)\in\cal F^1$ such that $\nabla\cdot\vect E^0=4\pi\rho_{t_0}$ and $\nabla\cdot\vect B^{0}=0$ as well as
  \begin{align}\label{eqn:mem_loss}
    \|\vect E^0(\vect x)\|+\|\vect B^0(\vect x)\| + \|\vect x\|\sum_{i=1}^3\left(\|\partial_{\vect x_i}\vect E^0(\vect x)\|+\|\partial_{\vect x_i}\vect B^0(\vect x)\|\right) = \underset{\|\vect x\|\to\infty}\bigoh\left(\|\vect x\|^{-\epsilon}\right)
  \end{align}
  for some $\epsilon> 0$ and all $\vect x\in\bb R^3$ are fulfilled. We distinguish two cases denoted by $+$ or $-$ and assume that for all $t\in\bb R$, $(\vect q,\vect p)\in\cal T_{\text{\clock}!}^1([t,\infty))$ or $(\vect q,\vect p)\in\cal T_{\text{\clock}!}^1((-\infty,t])$ holds, respectively.
  Then the point-wise limit
  \begin{align}\label{eqn:LW_fields}
    \begin{split}
      M_{\varrho}[\vect q,\vect p](t,\pm\infty)&:=\operatorname{pw-lim}\limits_{t_0\to\pm\infty} M_{\varrho}[F^0,(\vect q,\vect p)](t,t_0)\\
      &=4\pi\int_{\pm\infty}^{t} ds\; \left[K_{t-s} * \begin{pmatrix}
        -\nabla & - \partial_s\\
        0 & \nabla\wedge
        \end{pmatrix}\begin{pmatrix}
          \rho_s\\
          \vect j_s
        \end{pmatrix}\right]
      =\int d^3z\; \varrho(\vect z)\begin{pmatrix}
        \vect E^{LW\pm}_t(\cdot-\vect z)\\
        \vect B^{LW\pm}_t(\cdot-\vect z)
      \end{pmatrix}
    \end{split}
  \end{align}
  exists in $\cal F^1$ where 
  \begin{align}\label{eqn:LW_E_integrand}
    \vect E^{LW\pm}_t(\vect x-\vect z)&:=\left[\frac{({\vect n}\pm\vect v)(1-\vect v^2)}{\|\vect x -\vect z- \vect q\|^2(1\pm {\vect n}\cdot\vect v)^3}+\frac{{\vect n}\wedge[({\vect n}\pm\vect v)\wedge \vect a]}{\|\vect x -\vect z -\vect q\|(1\pm {\vect n}\cdot\vect v)^3}\right]^\pm\\
    \vect B^{LW\pm}_t(\vect x-\vect z)&:=\mp[\vect n\wedge\vect E_t(\vect x-\vect z)]^\pm\label{eqn:LW_B_integrand}
  \end{align}
  and
  \begin{align}
    \begin{split}\label{eqn:LW_abbreviations}
      \begin{array}{rlrlc}
       \vect q^\pm &:= \vect q_{t^\pm} & \vect v^\pm &:= \vect v(\vect p_{t^\pm}) & \vect a^\pm := \partial_{t}{\vect v}(\vect p_{t})|_{t=t^{\pm}}\\
       \vect n^\pm &:= \frac{\vect x-\vect z-\vect q^\pm}{\|\vect x-\vect z-\vect q^\pm\|} & t^\pm &= t \pm \|\vect x-\vect z-\vect q^\pm\|.
      \end{array}
    \end{split}
  \end{align}
\end{theorem}

\begin{remark}\label{rem:mem loss}
  Condition (\ref{eqn:mem_loss}) guarantees that in the limit $t_0\to\pm\infty$ the initial value $F^0$ are forgotten by the time evolution of the Maxwell equations. The condition that $(\vect q,\vect p)$ are strictly time-like is only sufficient for the limits $t_0\to\pm\infty$ to exist but necessary to yield formulas (\ref{eqn:LW_E_integrand}) and (\ref{eqn:LW_B_integrand}); note the blowup of the denominators $(1\pm\vect n\cdot\vect v)$ for $\|\vect v\|\to 1$. 
\end{remark}
\begin{theorem}[Li\'enard-Wiechert Fields Solve the Maxwell Equations]\label{thm:LW solve M eq}
  Let $(\vect q,\vect p)\in\cal T_{\text{\clock}!}^1$, then the Li\'enard-Wiechert fields $M_{\varrho}[\vect q,\vect p](t,\pm\infty)$ are a solution to the Maxwell equations (\ref{eqn:maxwell_equations_charge_current}) including the Maxwell constraints for all $t\in\bb R$.
\end{theorem}
We immediately get a simple bound on the Li\'enard-Wiechert fields:
\begin{corollary}[Li\'enard-Wiechert Estimate]\label{cor:LW_estimate}
  Let $(\vect q,\vect p)\in\cal T_{\text{\clock}!}^1$. Furthermore, assume there exists an $\namel{amax}{a_{max}}<\infty$ such that $\sup_{t\in\bb R}\|\partial_t\vect v(\vect p_t)\|\leq\namer{amax}$. Then the Li\'enard-Wiechert fields (\ref{eqn:LW_E_integrand}) and (\ref{eqn:LW_B_integrand}) fulfill: For any multi-index $\alpha\in\bb N_0^3$ there exists a constant $\constl{LW_bound}^{(\alpha)}<\infty$ such that for all $\vect x\in\bb R^3$, $t\in\bb R$
  \begin{align*}
    \|D^\alpha\vect E^\pm_t(\vect x)\|+\|D^\alpha\vect B^\pm_t(\vect x)\|\leq \frac{\constr{LW_bound}^{(\alpha)}}{(1-\namer{vmax})^3}\left(\frac{1}{1+\|\vect x-\vect q_t\|^2}+\frac{\namer{amax}}{1+\|\vect x-\vect q_t\|}\right)
  \end{align*}
  holds.
\end{corollary}

\subsection{The $\text{ML-SI}_{\varrho}$ Time Evolution}\label{sec:summary of part I}

Next, we briefly summarize the results of \cite{bauer_maxwell_2010} on the $\text{ML-SI}_\varrho$ equations (\ref{eqn:maxwell equations})-(\ref{eqn:lorentz force}):
\begin{definition}[Weighted Square Integrable Functions]\label{def:weighted spaces}
  We define the class of weight functions
  \begin{align}\label{eqn:weightclass}
    \cal W:=\Big\{w\in\cal C^\infty(\bb R^3,\bb R^+\setminus\{0\}) \;\big|\;& \exists\; \namel{cw}{C_w}\in\bb R^+,\namel{pw}{P_w}\in\bb N: w(\vect x+\vect y)\leq (1+\namer{cw}\|\vect x\|)^\namer{pw} w(\vect y)\Big\}.
  \end{align}
   For any $w\in\cal W$ and open $\Omega\subseteq\bb R^3$ we define the space of weighted square integrable functions $\Omega\to\bb R^3$ by
  \[
    L^2_w(\Omega,\bb R):=\left\{\vect F:\Omega\to\bb R^3\;\text{measurable} \;\bigg|\; \intdv x w(\vect x)\|\vect F(\vect x)\|^2<\infty\right\}.
  \]
  For regularity arguments we need more conditions on the weight functions. For $k\in\bb N$ we define
  \begin{align}\label{eqn:weight function spaces}
    \cal W^k:=\Big\{w\in\cal W \;\big|\; \exists\; \namel{calpha}{C_\alpha}\in\bb R^+: |D^\alpha \sqrt w|\leq \namer{calpha}\sqrt w, |\alpha|\leq k\Big\}
  \end{align}
  and
  \[
    \cal W^\infty:=\bigcap_{k\in\bb N}\cal W^{k}.
  \]
\end{definition}
\begin{remark}
 As computed in \cite{bauer_maxwell_2010}, $\cal W\ni w(\vect x):=(1+\|\vect x\|^2)^{-1}$.
\end{remark}

The space of initial values is then given by:
\begin{definition}[Phase Space]\label{def:phasespace}
  We define
  \[
    \cal H_w:=\bigoplus_{i=1}^N \left(\bb R^3\oplus\bb R^3\oplus L^2_w(\bb R^3,\bb R^3) \oplus L^2_w(\bb R^3,\bb R^3)\right).
  \]
  Any element $\varphi\in\cal H_w$ consists of the components $\varphi=(\vect q_i,\vect p_i,\vect E_i,\vect B_i)_{1\leq i\leq N}$, i.e. positions $\vect q_i$, momenta $\vect p_i$ and electric and magnetic fields $\vect E_i$ and $\vect B_i$ for each of the $1\leq i\leq N$ charges.
\end{definition}
If not noted otherwise, any spatial derivative will be understood in the distribution sense, and the Latin indices $i,j,\ldots$ shall run over the charge labels $1,2,\ldots, N$. For $w\in\cal W$, open set $\Omega\subseteq\bb R^3$ and $k\geq 0$ we define the following Sobolev spaces
  \begin{align}
  \begin{split}\label{eqn:sobolev spaces}
    H^k_w(\Omega,\bb R^3)&:=\bigg\{\vect f\in L^2_w(\Omega,\bb R^3)\;\Big|\;D^\alpha \vect f\in L^2_w(\Omega,\bb R^3), |\alpha|\leq k\bigg\},\\
    H^{\triangle^k}_w(\Omega,\bb R^3)&:=\bigg\{\vect f\in L^2_w(\Omega,\bb R^3)\;\Big|\;\triangle^j \vect f\in L^2_w(\Omega,\bb R^3) \text{ for }0\leq j\leq k \bigg\},\\
    H^{curl}_w(\Omega,\bb R^3)&:=\bigg\{\vect f\in L^2_w(\Omega,\bb R^3)\;\Big|\;\nabla\wedge \vect f\in L^2_w(\Omega,\bb R^3)\bigg\}
  \end{split}
  \end{align}
  which are equipped with the inner products
  \begin{align*}
    \braket{\vect f,\vect g}_{H^k_w} := \sum_{|\alpha|\leq k}\braket{D^{\alpha} \vect f,D^{\alpha} \vect g}_{L^2_w(\Omega)},&&
    \braket{\vect f,\vect g}_{H^{\triangle}_w(\Omega)} := \sum_{j=0}^k\braket{\triangle^j \vect f,\triangle^j \vect g}_{L^2_w(\Omega)}
  \end{align*}
  \begin{align*}
    \braket{\vect f,\vect g}_{H^{curl}_w(\Omega)} := \braket{\vect f,\vect g}_{L^2_w(\Omega)}+\braket{\nabla\wedge \vect f,\nabla\wedge \vect g}_{L^2_w(\Omega)},
  \end{align*}
  respectively. We use the multi-index notation $\alpha=(\alpha_1,\alpha_2,\alpha_3)\in(\bb N_0)^3$, $|\alpha|:=\sum_{i=1}^3\alpha_i$, $D^\alpha=\partial_1^{\alpha_1}\partial_2^{\alpha_2}\partial_3^{\alpha_3}$ where $\partial_i$ denotes the derivative w.r.t. to the $i$-th standard unit vector in $\bb R^3$. In order to appreciate the structure of the ML equations we will rewrite them using the following operators $A$ and $J$:
\begin{definition}[Operator A]\label{def:operator_A}
  For a $\varphi=(\vect q_i,\vect p_i,\vect E_i,\vect B_i)_{1\leq i\leq N}$ we defined ${\mathtt A}$ and $A$ by the expression
  \[
      A\varphi = \Big(0,0,{\mathtt A}(\vect E_i,\vect B_i)\Big)_{1\leq i\leq N}
       :=\Big(0,0,-\nabla\wedge\vect E_i,\nabla\wedge\vect B_i)\Big)_{1\leq i\leq N}.
  \]
  on their natural domain
    \[
    D_w(A):=\bigoplus_{i=1}^N \left(\bb R^3 \oplus \bb R^3 \oplus H^{curl}_w(\bb R^3,\bb R^3) \oplus H^{curl}_w(\bb R^3,\bb R^3)\right)\subset \cal H_w.
  \]
  Furthermore, for any $n\in\bb N$ we define
  \begin{align*}
    D_w(A^n):=\big\{\varphi \in D_w(A)\;\big|\;A^k\varphi\in D_w(A)\text{ for }k=0,\ldots,n-1\big\}, && D_w(A^\infty):=\bigcap_{n=0}^\infty D_w(A^n).
  \end{align*}
\end{definition}
\begin{definition}[Operator J]\label{def:operator_J}
   Together with $\vect v(\vect p_i):=\frac{\vect p_i}{\sqrt{\vect p_i^2+m^2}}$ we define $J:\cal H_w\to D_w(A^\infty)$ by
  \[
    \varphi\mapsto J(\varphi) := \left(\vect v(\vect p_i),
      \sum_{j\neq i}^N \intdv x \varrho_{i}(\vect x-\vect q_{i})\left( \vect E_{j}(\vect x) + \vect v(\vect p_i) \wedge \vect B_{j}(x) \right),
      - 4\pi \vect v(\vect p_i) \varrho_i(\cdot-\vect q_{i}),
      0\right)_{1\leq i\leq N}
  \]
  for $\varphi=(\vect q_i,\vect p_i,\vect E_i,\vect B_i)_{1\leq i\leq N}$.
\end{definition}
Note that $J$ is well-defined because $\varrho_{i}\in\cal C^\infty_c(\bb R^3,\bb R)$, $1\leq i\leq N$. With these definitions, the Lorentz force law (\ref{eqn:lorentz force}), the Maxwell equations (\ref{eqn:maxwell equations}) without the Maxwell constraints take the form
\begin{align}\label{eqn:dynamic_maxwell}
   \partial_{t}\varphi_t = A\varphi_t + J(\varphi_t).
\end{align}
The two main theorems are:
\begin{theorem}[Global Existence and Uniqueness]\label{thm:globalexistenceanduniqueness}
   For $w\in\cal W^1$, $n\in\bb N$ and $\varphi^0\in D_w(A^n)$ the following holds:
  \begin{enumerate}[(i)]
    \item \emph{(global existence)} There exists an $n$-times continuously differentiable mapping
          \begin{align*}
            \varphi_{(\cdot)}:\bb R \to \cal H_w, &&
            t\mapsto\varphi_{t}=(\vect q_{i,t},\vect p_{i,t},\vect E_{i,t},\vect B_{i,t})_{1\leq i\leq N}
          \end{align*}
           which solves (\ref{eqn:dynamic_maxwell})
          for initial value $\varphi_t|_{t=0}=\varphi^0$. Furthermore, it holds  $\frac{d^j}{dt^j}\varphi_t\in D_w(A^{n-j})$ for all $t\in\bb R$ and $0\leq j\leq n$,
    \item \emph{(uniqueness and growth)}  Any once continuously differentiable function $\widetilde\varphi:\Lambda\to D_w(A)$ for some open interval $\Lambda\subseteq\bb R$ which fulfills $\widetilde\varphi_{t^*}=\varphi_{t^*}$ for an $t^*\in\Lambda$, and which also solves the equation (\ref{eqn:dynamic_maxwell}) on $\Lambda$, has the property that $\varphi_t=\widetilde \varphi_t$ holds for all $t\in \Lambda$. In particular, given $\varrho_i$, $1\leq i\leq N$ there exists $\constl{apriori lipschitz}\in\bounds$ such that for $T>0$ with $[-T,T]\subset\Lambda$ it holds
        \begin{align}\label{eqn:apriori lipschitz}
          \sup_{t\in[-T,T]}\|\varphi_t-\widetilde\varphi_t\|_{\cal H_w}\leq \constr{apriori lipschitz}(T,\|\varphi_{t_0}\|_{\cal H_w},\|\widetilde\varphi_{t_0}\|_{\cal H_w})\|\varphi_{t_0}-\widetilde\varphi_{t_0}\|_{\cal H_w}.
        \end{align}
        Furthermore, there is a $\constl{apriori ml rho}\in\bounds$ such that for all $\varrho_i$, $1\leq i\leq N$,
        \begin{align}\label{eqn:apriori lipschitz no diff}
          \sup_{t\in[-T,T]}\|\varphi_t\|_{\cal H_w} \leq \constr{apriori ml rho}\left(T,\|w^{-1/2}\varrho_i\|_{L^2},\|\varrho_i\|_{L^2_w}; 1\leq i\leq N\right)\; \|\varphi^0\|_{\cal H_w}.
        \end{align}
    \item \emph{(constraints)} If the solution $t\mapsto\varphi_{t}=(\vect q_{i,t},\vect p_{i,t},\vect E_{i,t},\vect B_{i,t})_{1\leq i\leq N}$ obeys the Maxwell constraints
        \begin{align}\label{eqn:ml constraints}
          \nabla\cdot \vect E_{i,t}=4\pi\varrho_{i}(\cdot-\vect q_{i,t}), && \nabla\cdot\vect B_{i,t}=0
        \end{align}
        for $1\leq i\leq N$ and one time instant $t\in\bb R$, then they are obeyed for all times $t\in\bb R$.
  \end{enumerate}
\end{theorem}

\begin{theorem}[Regularity]\label{thm:regularity}
  Assume the same conditions as in Theorem \ref{thm:globalexistenceanduniqueness} hold and let $t\mapsto\varphi_t=(\vect q_{i,t},\vect p_{i,t},\vect E_{i,t},\vect B_{i,t})_{1\leq i\leq N}$ be the solution to (\ref{eqn:dynamic_maxwell}) for initial value $\varphi^0\in D_w(A^n)$. In addition, let $w\in\cal W^2$ and $n=2m$ for $m\in\bb N$. Then for all $1\leq i\leq N$:
  \begin{enumerate}[(i)]
    \item It holds for any $t\in\bb R$ that $\vect E_{i,t},\vect B_{i,t}\in\cal H_w^{\triangle^m}$.
    \item The electromagnetic fields regarded as maps $\vect E_i:(t,\vect x)\mapsto \vect E_{i,t}(\vect x)$ and $\vect B_i:(t,\vect x)\mapsto \vect B_{i,t}(\vect x)$ are in $L^2_{loc}(\bb R^4,\bb R^3)$ and both have a representative in $\cal C^{n-2}(\bb R^4,\bb R^3)$ within their equivalence class.
    \item For $w\in \cal W^k$ for $k\geq 2$ and every $t\in\bb R$ we have also $\vect E_{i,t},\vect B_{i,t}\in H^n_w$ and $C<\infty$ such that:
          \begin{align}\label{eqn:fieldbound}
            \sup_{\vect x\in\bb R^3}\sum_{|\alpha|\leq k}\|D^\alpha\vect E_{i,t}(\vect x)\|\leq C\|\vect E_{i,t}\|_{H^k_w}, && \sup_{\vect x\in\bb R^3}\sum_{|\alpha|\leq k}\|D^\alpha\vect B_{i,t}(\vect x)\|\leq C\|\vect B_{i,t}\|_{H^k_w}.
          \end{align}
  \end{enumerate}
\end{theorem}
As shown in \cite[Lemma 2.19]{bauer_maxwell_2010}\todo{check reference}, $A$ on $D_w(A)$ is a closed operator that generates a $\gamma$-contractive group $(W_t)_{t\in\bb R}$:
\begin{definition}[Free Maxwell Time Evolution]\label{def:Wt}
  We denote by $(W_t)_{t\in\bb R}$ the $\gamma$-contractive group on $\cal H_w$ generated by $A$ on $D_w(A)$.
\end{definition}
\begin{remark}
  The $\gamma$-contractive group $(W_t)_{t\in\bb R}$ comes with a standard bound $\|W_t\varphi\|_{\cal H_w}\leq e^{\gamma|t|}\|\varphi\|_{\cal H_w}$ for all $\varphi\in\cal H_w$ for some $\gamma\geq 0$.
\end{remark}

The above existence and uniqueness result implicates:
\begin{definition}[ML Time Evolution]\label{def:ML time evolution}
  We define the non-linear operator
  \begin{align*}
    M_L:\bb R^2\times D_w(A) \to D_w(A), &&
    (t,t_0,\varphi^0) \mapsto M_L(t,t_0)[\varphi^0]=\varphi_t=W_{t-t_0}\varphi^0 + \int_{t_0}^t W_{t-s} J(\varphi_s)
  \end{align*}
  which encodes the ML time evolution from time $t_0$ to time $t$.
\end{definition}

Using the presented results in Section \ref{sec:maxwellsolutions} on the Maxwell equations we can give explicit expressions of the free Maxwell time evolution group $(W_{t})_{t\in\bb R}$ and the $\text{ML-SI}_\varrho$ time evolution for initial fields fulfilling both the regularity requirements of $D_{w}(A)$ and of $\cal F^{1}$. The following short-hand notation will be convenient:

\begin{notation}[Projectors $\pP,\pQ,\pF$]
  For any $\varphi=(\vect q_i,\vect p_i,\vect E_i,\vect B_i)_{1\leq i\leq N}\in\cal H_w$ we define
  \begin{align*}
    \pQ\varphi=(\vect q_i,0,0,0)_{1\leq i\leq N}, && \pP\varphi=(0,\vect p_i,0,0)_{1\leq i\leq N}, && \pF\varphi=(0,0,\vect E_i,\vect B_i)_{1\leq i\leq N}.
  \end{align*}
  where we sometime neglect the zeros and write for example
  \begin{align*}
    (\vect q_i,\vect p_i)_{1\leq i\leq N}=\pQP\varphi && \text{or} && (\vect q_i,\vect p_i,0,0)_{1\leq i\leq N}=\pQP(\vect q_i,\vect p_i)_{1\leq i\leq N}.
  \end{align*}
\end{notation}

\begin{definition}[Projection of $A,W_t,J$ to Field Space $\cal F_w$]\label{def:AWJ}
  For all $t\in\bb R$ and $\varphi\in\cal H_w$ we define
  \begin{align*}
    \cal F_{w}:=F\cal H_{w}, && \mA:=\pF A\pF, && \mW_t:=\pF W_t\pF, && \mJ:=\pF J(\varphi).
  \end{align*}
  The natural domain of $\mA$ is given by $D_w(\mA):=\pF D_w(A)\subset\cal F_w$. We shall also need $D_w(\mA^n):=\pF D_w(A^n)\subset\cal F_w$ for every $n\in\bb N$ and $D_{w}(\mA^{\infty}):=\pF D_{w}(A^{\infty})$. 
\end{definition}
Note the distinction between roman and sans serif letters, e.g. $A$ and $\mathsf{A}$. Clearly, $\cal F_{w}$ is a Hilbert space, the operator $\mA$ on $D_w(\mA)$ is again closed and inherits the resolvent properties from $A$ on $D_w(A)$. This implies $\pQP W_t=\id_{\cal P}$ and $\pF W_t=\mW_t$ so that $(\mW_t)_{t\in\bb R}$ is also a $\gamma$-contractive group generated by $\mA$ on $D_w(\mA)$. Finally, note also that by the definition of $J$ we have $\mJ(\varphi)=\mJ(\pQP \varphi)$ for all $\varphi\in\cal H_w$, i.e. $\mJ$ does not depend on the field components $\pF \varphi$.

We extend the space of fields $\cal F^1$, cf. Definition \ref{def:initial fields F}, to comprise $N$ electric and magnetic fields:
\begin{definition}[Space of $N$ Smooth Fields]
  $\cal F^{N}:=\bigoplus^N_{i=1}\cal F^1=\bigoplus^N_{i=1}\cal C^\infty(\bb R^3,\bb R^3)\oplus\cal C^\infty(\bb R^3,\bb R^3)$.
\end{definition}

The following corollary gives an explicit expression for the action of the group $(\mW_t)_{t\in\bb R}$ using the results from in Section \ref{sec:maxwellsolutions} about the free Maxwell equation.
\begin{corollary}[Kirchoff's Formulas for $(W_t)_{t\in\bb R}$]\label{cor:connection m wt}
  Let $w\in\cal W^1$, $F\in D_w(\mA^n)\cap\cal F^{N}$ for some $n\in\bb N$, and
  \begin{align*}
    (\vect E_{i,t},\vect B_{i,t})_{1\leq i\leq N}:=\mW_{t}F,\qquad t\in\bb R.
  \end{align*}
  Then
  \begin{align*}
    \begin{pmatrix}\widetilde{\vect E}_{i,t}\\\widetilde{\vect B}_{i,t}\end{pmatrix}=\begin{pmatrix}
      \partial_t & \nabla\wedge\\
      -\nabla\wedge & \partial_t
    \end{pmatrix}
    K_{t}*\begin{pmatrix}\vect E_{i,0}\\\vect B_{i,0}\end{pmatrix}-\int_{0}^tds\;K_{t-s}*\begin{pmatrix}
      \nabla\nabla\cdot\vect E_{i,0}\\
      \nabla\nabla\cdot\vect B_{i,0}
    \end{pmatrix}
  \end{align*}
  fulfill $\vect E_{i,t}=\widetilde{\vect E}_{i,t}$ and $\vect B_{i,t}=\widetilde{\vect B}_{i,t}$ for all $t\in\bb R$ and $1\leq i\leq N$ in the $L^2_w$ sense. Furthermore, for all $t\in\bb R$ it holds also that $(\vect E_{i,t},\vect B_{i,t})_{1\leq i\leq N}\in D_w(\mA^n)\cap\cal F^{N}$.
\end{corollary}
\begin{proof}
  A direct application of Lemma \ref{lem:Greens_function_dalembert} and Definition \ref{def:Wt}.
\end{proof}

From this corollary we can also express the inhomogeneous Maxwell time-evolution, cf. Definition \ref{def:Maxwell time evolution}, in terms of $(\mW_t)_{t\in\bb R}$ and $\mJ$.

\begin{lemma}[The Maxwell Solutions in Terms of $(\mW_t)_{t\in\bb R}$ and $\mJ$]\label{lem:connection maxwell time and Wt J}
  Let times $t,t_0\in\bb R$ be given, $F=(F_i)_{1\leq i\leq N}\in D_w(\mA^n)\cap\cal F^{N}$ for some $n\in\bb N$ be given initial fields, and $(\vect q_i,\vect p_i)\in\cal T_{\text{\clock}}^1$ time-like charge trajectories for $1\leq i\leq N$. In addition suppose the initial fields $F_i=(\vect E_i,\vect B_i)$, $1\leq i\leq N$,  fulfill the Maxwell constraints
  \begin{align*}
    \nabla\cdot\vect E_i=4\pi\varrho_{i}(\cdot-\vect q_{i,t_0}), && \nabla\cdot\vect B_i=0.
  \end{align*}
Then for all $t\in\bb R$
  \begin{align*}
     F_t:=\mW_{t-t_0}F+\int_{t_0}^tds\;\mW_{t-s}\mJ(\varphi_s)\in D_w(\mA^n)\qquad=\qquad \big(M_{\varrho_{i}}[F_i,(\vect q_i,\vect p_i)](t,t_0)\big)_{1\leq i\leq N}
  \end{align*}
  holds in the $L^2_w$ sense where $\varphi_s:=\pQP(\vect q_{i,s},\vect p_{i,s})_{1\leq i\leq N}$ for $s\in\bb R$. Furthermore, $F_t\in D_w(\mA^n)\cap\cal F^{N}$ for all $t\in\bb R$.
\end{lemma}
\begin{proof}
  This can be computed by applying Corollary \ref{cor:connection m wt} twice and using one partial integration.
\end{proof}

\section{Proofs}

\subsection{Weak Uniqueness of $\text{WF}_{\varrho}$ and Synge Solutions by $\text{ML-SI}_{\varrho}$ Cauchy Data}\label{sec:WF initial fields}

Our first goal is to prove Theorem \ref{thm:WF initial conditions}. Using the results of Section \ref{sec:maxwellsolutions} we can give a sensible definition of what we mean by solutions to equations (\ref{eqn:WF equation written out})-(\ref{eqn:WF fields def}) for particular choices of $e_+$ and $e_-$. Recall that $e_+=\frac{1}{2},e_-=\frac{1}{2}$ and $e_+=0,e_-=1$ corresponds to the $\text{WF}_{\varrho}$ equations and the Synge equations, respectively.

\begin{definition}[Class of Solutions]\label{def:WF sols}
  We define $\cal T^{N}_{(e_+,e_-)}$ to consist of elements $({\vect q}_i, {\vect p}_i)_{1\leq i\leq N}\in \cal T_{\text{\clock}!}^N$ which fulfill:
  \begin{enumerate}[(i)]
    \item There exists an $\namel{amax}{a_{max}}<\infty$ such that for all $1\leq i\leq N$, $\sup_{t\in\bb R}\|\partial_t\vect v(\vect p_{i,t})\|\leq\namer{amax}$.
    \item $({\vect q}_i, {\vect p}_i)_{1\leq i\leq N}$ solve the equations (\ref{eqn:WF equation written out})-(\ref{eqn:WF fields def}) for all times $t\in\bb R$ and the particular choice of $e_{+},e_{-}$.
  \end{enumerate}
\end{definition}
\begin{remark}
  (1) Note that this definition is sensible because with $(\vect q_i,\vect p_i)_{1\leq i\leq N}\in\cal T^{N}_{\text{\clock}!}$, equations (\ref{eqn:WF fields def}) for $1\leq i\leq N$ coincide with
  \begin{align*}
    (\vect E^{(e_+,e_-)}_{i,t},
      \vect B^{(e_+,e_-)}_{i,t})
    = \sum_\pm e_\pm M_{\varrho_{i}}[\vect q_i,\vect p_i](t,\pm\infty)
  \end{align*}
by definition in (\ref{eqn:LW_fields}). Theorem \ref{thm:LWfields} guarantees that the right-hand side is well-defined, and charge trajectories in $\cal T_{\text{\clock}!}^1$ are once continuously differentiable so that the left-hand side of (\ref{eqn:WF equation written out}) is also well-defined. The bound on the acceleration will give us a bound on the $\text{WF}_\varrho$ fields in a suitable norm; cf. Lemma \ref{lem:LW fields are in DwA}. (2) Furthermore, there is no doubt that $\cal T^{N}_{(e_+,e_-)}$ is non-empty because in the point-particle case the Schild solutions \cite{schild_electromagnetic_1963} as well as the solutions of Bauer's existence theorem \cite{bauer_ein_1997} have smooth and strictly time-like charge trajectories with bounded accelerations.
\end{remark}

As discussed in Section \ref{sec:summary of part I}, the electric and magnetic fields live in the $L^{2}_{w}$ space for a conveniently chosen weight $w\in\cal W^{\infty}$, cf. Definitions \ref{def:weighted spaces} and \ref{def:phasespace}. In the following we give an example weight $w$ and show that with it the Li\'enard-Wiechert fields of charge trajectories in $\cal T_{\text{\clock}!}^N$ with uniformly bounded accelerations are admissible as $\text{ML-SI}_{\varrho}$ initial data; cf. Theorem \ref{thm:globalexistenceanduniqueness} and Definition \ref{def:operator_A}.

\begin{definition}[Example Weight]
We define the function
 \begin{align}\label{eqn:weight}
w:\bb R^3\to\bb R^+\setminus\{0\},\quad\vect x\mapsto w(\vect x):=(1+\|\vect x\|^2)^{-1}.
\end{align}
\end{definition}
A straight-forward computation given in \cite{deckert_electrodynamic_2010} yields:
\begin{lemma}\label{lem:weight}
  The function $w$ is an element of $\cal W^\infty$.
\end{lemma}
\begin{lemma}[Regularity of the Li\'enard-Wiechert Fields]\label{lem:LW fields are in DwA}
  Let $( {\vect q}_i, {\vect p}_i)_{1\leq i\leq N}\in\cal T_{\text{\clock}!}^N$ and assume there exists a constant $\namel{amax}{a_{max}}<\infty$ such that for all $1\leq i\leq N$, $\sup_{t\in\bb R}\|\partial_t\vect v(\vect p_{i,t})\|\leq\namer{amax}$. Define $t\mapsto(\vect E_{i,t},\vect B_{i,t}):=M_{\varrho_{i}}[\vect q_i,\vect p_i](t,\pm\infty)$. Then for all $t\in\bb R$
  \begin{align*}
    (\vect q_{i,t},\vect p_{i,t},\vect E_{i,t},\vect B_{i,t})_{1\leq i\leq N}\in D_w(A^\infty),
  \end{align*}
  holds true.
\end{lemma}
\begin{proof}
  By Corollary \ref{cor:LW_estimate}, for $1\leq i\leq N$ and each multi-index $\alpha\in\bb N_{0}^3$ there exists a constant $\constr{LW_bound}^{(\alpha)}<\infty$ such that
  \begin{align*}
    \|D^\alpha\vect E^\pm_{i,t}(\vect x)\|+\|D^\alpha\vect B^\pm_{i,t}(\vect x)\|\leq \frac{\constr{LW_bound}^{(\alpha)}}{(1-\namer{vmax})^3}\left(\frac{1}{1+\|\vect x-\vect q_t\|^2}+\frac{\namer{amax}}{1+\|\vect x-\vect q_t\|}\right).
  \end{align*}
  Hence, $w(\vect x)=\frac{1}{1+\|\vect x\|^2}$ ensures that
  \begin{align*}
    \left\|A^n(\vect q_{i,t},\vect p_{i,t},\vect E_{i,t}^\pm,\vect B^\pm_{i,t})\right\|_{\cal H_w} &\leq \sum_{i=1}^N\sum_{|\alpha|\leq n}\left(\|\vect q_{i,t}\|+\|\vect p_{i,t}\|+\intdv x w(\vect x)\left(\|D^\alpha\vect E^\pm_{i,t}(\vect x)\|^2+\|D^\alpha\vect B^\pm_{i,t}(\vect x)\|^2\right)\right)
  \end{align*}
  is finite for all $n\in\bb N_{0}$, $t\in\bb R$. We conclude that $\varphi_t\in D_w(A^\infty)$ for all $t\in\bb R$.
\end{proof}
We prove the first main result:

\begin{proof}[\textbf{Proof of Theorem \ref{thm:WF initial conditions} (Weak Uniqueness of Solutions)}]
(i) Since $(\vect q_{i},\vect p_{i})_{1\leq i\leq N}\in\cal T^{N}_{(e_+,e_-)}$, Lemma \ref{lem:LW fields are in DwA} guarantees $\varphi_{t_0}\in D_w(A^\infty)$ for all $t_0\in\bb R$.

(ii) First, the charge trajectories $(\vect q_{i},\vect p_{i})_{1\leq i\leq N}\in\cal T^{N}_{(e_+,e_-)}$ are once continuously differentiable and fulfill the $\text{WF}_{\varrho}$ equations (\ref{eqn:WF equation written out})-(\ref{eqn:WF fields def}). Second, by Theorem \ref{thm:LW solve M eq} the fields $(\vect E^{(e_+,e_-)}_t[\vect q_i,\vect p_i],\vect B^{(e_+,e_-)}_t[\vect q_i,\vect p_i])$ given in (\ref{eqn:WF fields def}) fulfill the Maxwell equations (\ref{eqn:maxwell equations}) including the Maxwell constraints for all $t\in\bb R$ and $1\leq i\leq N$. Hence, using (i), the equality 
\[
  \frac{d}{dt}\varphi^{(e_+,e_-)}_t[(\vect q_{i},\vect p_{i})_{1\leq i\leq N}]=A\varphi^{(e_+,e_-)}_t[(\vect q_{i},\vect p_{i})_{1\leq i\leq N}]+J(\varphi^{(e_+,e_-)}_t[(\vect q_{i},\vect p_{i})_{1\leq i\leq N}]),
\]
holds true (recall the notation in Section \ref{sec:summary of part I} before equation (\ref{eqn:dynamic_maxwell})). Due to (i) also \ifarxiv{\linebreak}{}$\phi_t:=M_L\left[\varphi^{(e_+,e_-)}_{t_0}[(\vect q_{i},\vect p_{i})_{1\leq i\leq N}]\right]$, $t\in\bb R$, is well-defined; cf. Definition \ref{def:ML time evolution}. Theorem \ref{thm:globalexistenceanduniqueness} states that $\phi_t$ is the only solution of $\partial_{t}\phi_t=A\phi_t+J(\phi_t)$ which fulfills $\phi_{t_0}=\varphi^{(e_+,e_-)}_{t_0}[(\vect q_{i},\vect p_{i})_{1\leq i\leq N}]$. Hence, $\phi_t=\varphi^{(e_+,e_-)}_{t}[(\vect q_{i},\vect p_{i})_{1\leq i\leq N}]$ holds for all $t\in\bb R$.

(iii) Suppose $(\vect q_{i},\vect p_{i})_{1\leq i\leq N}, (\widetilde{\vect q}_{i},\widetilde{\vect p}_{i})_{1\leq i\leq N}\in\cal T^{N}_{(e_+,e_-)}$ and define \ifarxiv{\linebreak}{}$\varphi:=i_{t_0}((\vect q_{i},\vect p_{i})_{1\leq i\leq N})$, $\widetilde\varphi:=i_{t_0}((\widetilde{\vect q}_{i},\widetilde{\vect p}_{i})_{1\leq i\leq N})$ for some $t_0\in\bb R$. According to (ii) we also set
\begin{align*}
 \varphi_t&=(\vect q_{i,t},\vect p_{i,t},\vect E_{i,t},\vect B_{i,t})_{1\leq i\leq N}:=M_L[\varphi](t,t_0),\\
\widetilde\varphi_t&=(\widetilde{\vect q}_{i,t},\widetilde{\vect p}_{i,t},\widetilde{\vect E}_{i,t},\widetilde{\vect B}_{i,t})_{1\leq i\leq N}:=M_L[\widetilde\varphi](t,t_0)
\end{align*}
for all $t\in\bb R$. Now, $\varphi=\widetilde\varphi$ implies $\varphi_t=\widetilde\varphi_t$ for all $t\in\bb R$. Hence, $(\vect q_{i,t},\vect p_{i,t})_{1\leq i\leq N}=(\widetilde{\vect q}_{i,t},\widetilde{\vect p}_{i,t})_{1\leq i\leq N}$ for all $t\in\bb R$, i.e. $(\vect q_{i},\vect p_{i})_{1\leq i\leq N}=(\widetilde{\vect q}_{i},\widetilde{\vect p}_{i})_{1\leq i\leq N}$ by Definition \ref{def:charge trajectory}.
\end{proof}

\if\arxiv 1
\begin{remark}
  Note that the weight function $w$ could be chosen to decay faster than the choice in Lemma \ref{lem:weight}. This freedom allows to generalize Theorem \ref{thm:WF initial conditions} also to include charge trajectories whose accelerations are not bounded but may grow with $t\to\pm\infty$ because the growth of the acceleration $\vect a$ in equations (\ref{eqn:LW_fields}) can be suppressed by the weight $w$. However, Theorem \ref{thm:globalexistenceanduniqueness} applies only for weights $w\in\cal W^1$ which means that $w$ is not allowed to decay faster than an inverse of any polynomial.
\end{remark}
\fi

\subsection{Existence and Uniqueness of Synge Solutions for given Histories}

We continue with the proof of our second main result:

\begin{proof}[\textbf{Proof of Theorem \ref{thm:exist and uni of synge} (Existence and Uniqueness of Synge Solutions)}] 
Set $e_+=0$ and $e_-=1$. (i)  By definition $(\vect q^-_i,\vect p_i^-)_{1\leq i\leq N}\in\cal T^{N}_{\text{\clock}!}$, so that due to Theorem \ref{thm:LWfields} and (\ref{eqn:WF fields def}) for all $t\leq t_0$ we can define
\[
 \varphi^-_{t}=(\vect q_{i,t}^-,\vect p_{i,t}^-,\vect E_{t}[\vect q^-_i,\vect p^-_i],\vect B_{t}[\vect q^-_i,\vect p^-_i])_{1\leq i\leq N}
\]
where the fields are given by the retarded Li\'enard-Wiechert fields of the past history $(\vect q^-,\vect p^-)\in\mathfrak H^-(t_0)$, i.e.
\[
(\vect E_{t}[\vect q^-_i,\vect p^-_i],\vect B_{t}[\vect q^-_i,\vect p^-_i])=M_{\varrho_{i}}[\vect q^-_i,\vect p^-_i](t,-\infty)).
\]
 Lemma \ref{lem:LW fields are in DwA} states $\varphi^-_{t_0}\in D_w(A^\infty)$. Hence, by Theorem \ref{thm:globalexistenceanduniqueness} there is a unique mapping
\begin{align}\label{eqn:futer_MLSI}
  t\mapsto(\vect q_{i,t}^+,\vect p_{i,t}^+,\vect E^+_{i,t},\vect B^+_{i,t})_{1\leq i\leq N}=\varphi^+_t:=M_L[\varphi^-_{t_0}](t,t_0)
\end{align}
such that $\varphi^+_{t_0}=\varphi^-_{t_0}$. Let $(\vect q_{i},\vect p_{i})_{1\leq i\leq N}$ be the concatenation defined in (\ref{eqn:concat}). We consider now 
\[
 \varphi_{t}=(\vect q_{i,t},\vect p_{i,t},\vect E_{t}[\vect q_i,\vect p_i],\vect B_{t}[\vect q_i,\vect p_i])_{1\leq i\leq N}
\]
for all $t\in\bb R$ with the retarded Li\'enard-Wiechert fields of $(\vect q_i,\vect p_i)$ given by
\begin{align}\label{eqn:ret fields}
(\vect E_{t}[\vect q_i,\vect p_i],\vect B_{t}[\vect q_i,\vect p_i]):=M_{\varrho_{i}}[\vect q_i,\vect p_i](t,-\infty),
\end{align}
which are well-defined if $(\vect q_i,\vect p_i)$ would be in $\cal T^{1}_{\text{\clock}!}((-\infty,T])$ for all $T\in\bb R$. However, $(\vect q^-_{i},\vect q^-_{i})_{1\leq i\leq N}\in\cal T^{N}_{\text{\clock}!}$, and $(\vect q^+_{i},\vect p^+_{i})_{1\leq i\leq N}$ is continuously differentiable so that we only need to check that $(\vect q_{i},\vect p_{i})_{1\leq i\leq N}$ is continuously differentiable at $t=t_0$. Now according to the assumption, at time $t=t_0$ the past history $(\vect q^-_{i},\vect p^-_{i})_{1\leq i\leq N}$ solves equations (\ref{eqn:WF equation written out})-(\ref{eqn:WF fields def}) for $e_+=0$ and $e_-=1$, and furthermore, Theorem \ref{thm:LW solve M eq} states that $(\vect E_{t}[\vect q^-_i,\vect p^-_i],\vect B_{t}[\vect q^-_i,\vect p^-_i])_{1\leq i\leq N}$ solve the Maxwell equations at $t=t_{0}$. Hence, we have
\[
  \lim_{t\nearrow t_0}\frac{d}{dt}\varphi^-_t=A\varphi^-_{t_0}+J(\varphi^-_{t_0})=\lim_{t\searrow t_0}\frac{d}{dt}\varphi^+_t,
\]
i.e. $(\vect q_i,\vect p_i)\in\cal T^{1}_{\text{\clock}!}((-\infty,T])$ for $1\leq i\leq N$ and any $T\in\bb R$ so that (\ref{eqn:ret fields}) is well-defined.  With the help of Theorem \ref{thm:globalexistenceanduniqueness} for all $t\geq t_0$ we compute
\begin{align}\label{eqn:uniqueness synge}
  \frac{d}{dt}(\varphi_t-\varphi^+_t)=A(\varphi_{t}-\varphi^+_{t})+ \left[J(\varphi_s)-J(\varphi^+_s)\right]=A(\varphi_{t}-\varphi^+_{t})
\end{align}
 because $J$ does only depend on the charge trajectories. The only solution to this equation is $W_t(\varphi_{t_{0}}-\varphi^+_{t_0})=0$; cf. Definition \ref{def:Wt}. Hence, $(\vect E^+_{i,t},\vect B^+_{i,t})=(\vect E_{t}[\vect q_i,\vect p_i],\vect B_{t}[\vect q_i,\vect p_i])$ for $1\leq i\leq N$ and all $t\geq t_0$, i.e. the fields generated by the $\text{ML-SI}_\varrho$ time evolution equal the retarded Li\'enard-Wiechert fields corresponding to the charge trajectories generated by $\text{ML-SI}_\varrho$ time evolution. This implies that $(\vect q_{i},\vect p_{i})_{1\leq i\leq N}$ solve the $\text{WF}_{\varrho}$ equations (\ref{eqn:WF equation written out})-(\ref{eqn:WF fields def}) for $e_+=0$ and $e_-=1$ and all $t\geq t_0$.

(ii) Since $(\widetilde{\vect q}_{i,t},\widetilde{\vect p}_{i,t})=({\vect q}^-_{i,t},{\vect p}^-_{i,t})_{1\leq i\leq N}$ for all $t\leq t_{0}$ the claim follows from the uniqueness of the map (\ref{eqn:futer_MLSI}).
\end{proof}

\begin{remark}
(1) Condition (ii) in Definition \ref{def:synge histories} is only needed to ensure continuity of the derivative of the charge trajectories at $t_0$. Theorem \ref{thm:LWfields} can be generalized to piecewise $\cal C^1$ charge trajectories. Using this generalization Theorem \ref{thm:exist and uni of synge} can be proven without this condition, ensuring the existence of piecewise $\cal C^1$ Synge solutions for $t\geq t_0$. However, this condition is not restrictive in the sense that one had to fear $\mathfrak H(t_0)$ could be empty. Elements of $\mathfrak H(t_0)$ can be constructed with the following algorithm:
\begin{enumerate}
 \item Choose positions and momenta $(\vect q^-_{i,t_0},\vect q^-_{i,t_0})$ for $1\leq i\leq N$ particles at time $t_0$.
 \item For $1\leq i\leq N$ choose $(\vect q^-_{i,t},\vect p^-_{i,t})$ on time intervals from $-\infty$ up to the latest intersection of the backward light-cones of space-time points $(t_0,\vect q^-_{j,t_0})$, $j\neq i$, before time $t_0$.
 \item Use the Synge equations to compute the acceleration for all $1\leq i\leq N$ charges at $t_0$.
  \item For $1\leq i\leq N$ extend $(\vect q^-_{i,t},\vect p^-_{i,t})$ up to time $t_0$ smoothly such that they connect to the chosen $(\vect q^-_{i,t_0},\vect q^-_{i,t_0})$ with the correct acceleration computed in step 3.
\end{enumerate}

 (2) From the geometry of the Li\'enard-Wiechert fields it is clear that the whole history \ifarxiv{\linebreak}{} $(\vect q^-_{i,t},\vect p^-_{i,t})_{1\leq i\leq N}$ for $t\leq t_0$ is sufficient for uniqueness but not necessary. The necessary data for the charge trajectories $(\vect q^-_{i,t},\vect p^-_{i,t})_{1\leq i\leq N}$ that identify a Synge solution for $t\geq t_0$ uniquely are the shortest trajectory strips, so that the backward light-cone of each space-time point $(t,\vect q^-_{i,t_0})$ intersects all other charge trajectories $(\vect q^-_j,\vect p^-_j)$, $j\neq i$.
\end{remark}

\subsection{Existence of $\text{WF}_{\varrho}$ Solutions on Finite Time Intervals}
\label{sec:existence of initial fields}

We shall now prove the remaining main results Theorem \ref{thm:ST has a fixed point} and Theorem \ref{thm:existence of L}.  For the rest of this work we keep the choice $e_{+}=\frac{1}{2}$, $e_{-}=\frac{1}{2}$ fixed. The results, however, hold also for any choices of $0\leq e_{+},e_{-}\leq 1$. The strategy will be to use Schauder's fixed point theorem to prove the existence of a fixed point of $S^{X^{\pm}}_T$. Recall the distinction between roman and sans serif letters in Definition \ref{def:AWJ}. We generalize the definition of $\cal F_{w}$:
\begin{definition}[Hilbert Spaces for the Fixed Point Theorem]\label{def:Fwn}
  Given $n\in\bb N$ we define $\cal F_w^n$ to be the linear space of elements $F\in D_w(\mA^n)$ equipped with the inner product
  \begin{align*}
    \braket{F,G}_{\cal F_{w}^{n}}:=\sum_{k=0}^{n}\braket{\mA^{k}F,\mA^{k}G}_{\cal F_{w}}.
  \end{align*}
  The corresponding norm is denoted by $\|\cdot\|_{\cal F^{n}_{w}}$ and we shall use the notation 
  \[\|F\|_{\cal F_w^n(B)}:=\left(\sum_{k=0}^{n}\|A^{k}F\|^{2}_{L^{2}_{w}(B)}\right)^{1/2}\] to denote the restriction of the norm to a subset $B\subset\bb R^3$.
\end{definition}
\begin{lemma}
  For $n\in\bb N$, $\cal F_w^n$ is a Hilbert space.
\end{lemma}
\begin{proof}
  This is an immediate consequence of \cite[Theorem 2.10]{bauer_maxwell_2010} \todo{check reference} and relies on the fact that $\mA$ is closed on $D_w(\mA)$.
\end{proof}

As explained in Section \ref{sec:main} we encode the continuation of the charge trajectories for times $|t|\geq T$ in terms of advanced and retarded Li\'enard-Wiechert fields $X^+_{i,+T}$ and $X^-_{i,-T}$, respectively. These fields are generated by the prescribed charge trajectories for times $|t|\geq T$ and evaluated at time $T$. They must depend on the charge trajectories within $[-T,T]$ because we want to impose certain regularity conditions at the connection times $t=\pm T$. Since these trajectories will be generated within the iteration of $S^{p,X^\pm}_T$ by the $\text{ML-SI}_\varrho$ time evolution this dependence can be expressed simply by the dependence on the $\text{ML-SI}_\varrho$ initial data $(p,F)=\varphi\in D_w(A^\infty)$. We shall therefore use the notation $X^\pm_{i,\pm T}[\varphi]$ for the boundary fields.

Next, we introduce three classes of such boundary fields for our discussion, namely $\cal A^n_w\supset\widetilde{\cal A}^n_w\supset\cal A^\Lip$. The class $\cal A^n_w$ will allow to define what we mean by a conditional $\text{WF}_\varrho$ solution (see Definition \ref{def:WF sol for finite times} below). 
The existence of conditional $\text{WF}_\varrho$ solutions is then shown for the class $\widetilde{\cal A}^n_w$ with $n=3$. The third class,  $\cal A^\Lip$, is only needed for Remark \ref{rem:uniqueness} where we discuss uniqueness of the conditional $\text{WF}_\varrho$ solution for small enough $T$. We define:

\begin{definition}[Boundary Fields Classes $\cal A_w^n$, $\widetilde{\cal A}_w^n$ and $\cal A^\Lip_w$]\label{def:boundary fields}
  For weight $w\in\cal W$ and $n\in\bb N$ we define $\cal A_w^n$ to be the set of maps
  \begin{align*}
    X:\bb R\times D_w(A)\to D_w(\mA^\infty)\cap\cal F^{N},\qquad
    (T,\varphi)\mapsto X_T[\varphi]
  \end{align*}
  which have the following properties for all $p\in\cal P$ and $T\in\bb R$:
  \begin{enumerate}[(i)]
    \item There is a $\constl{c hwn norm}^{(n)}\in\bounds$ such that for all $\varphi\in D_w(A)$ with $\pQP \varphi=p$ it is true that $\|X_T[\varphi]\|_{\cal F_w^n}\leq \constr{c hwn norm}^{(n)}(|T|,\|p\|)$.
    \item The map $F\mapsto X_T[p,F]$ as $\cal F_w^1 \to\cal F_w^{1}$ is continuous.
    \item For $(\vect E_{i,T},\vect B_{i,T})_{1\leq i\leq N}:=X_T[\varphi]$ and $(\vect q_{i,T},\vect p_{i,T})_{1\leq i\leq N}:=\pQP M_L[\varphi](T,0)$ one has
\begin{align*}
 \nabla\cdot\vect E_{i,T}=4\pi\varrho_{i}(\cdot-\vect q_{i,T}), && \nabla\cdot\vect B_{i,T}=0.
\end{align*}
  \end{enumerate}
  The subset $\widetilde{\cal A}^{n}_w$ comprises maps $X\in\cal A^n_w$ that fulfill:
  \begin{enumerate}[(i)]
    \setcounter{enumi}{3}
    \item For balls $B_\tau:=B_\tau(0)\subset\bb R^3$ with radius $\tau>0$, $B_\tau^c:=\bb R^{3}\setminus B_\tau$,  and any bounded set $M\subset D_w(\mA)$ it holds that
    \begin{align*}
\lim_{\tau\to\infty}\sup_{F\in M}\|X_T[p,F]\|_{\cal F_w^n(B_\tau^c)}=0.
\end{align*}
  \end{enumerate}
  Furthermore, $\cal A^{\Lip}_w$ comprise such maps $X\in\subset\cal A^1_w$ that fulfill:
  \begin{enumerate}[(i)]
    \setcounter{enumi}{4}
    \item There is a $\constl{c strong lipschitz}\in\bounds$ such that for all $\varphi,\widetilde\varphi\in D_w(A)$ with $\pQP\varphi=p=\pQP\widetilde\varphi$ it is true that
\begin{align*}
 \|X_T[\varphi]-X_T[\widetilde\varphi]\|_{\cal F_w^1}\leq |T|\constr{c strong lipschitz}(|T|,\|\varphi\|_{\cal H_w},\|\widetilde\varphi\|_{\cal H_w})\;\|\varphi-\widetilde\varphi\|_{\cal H_w}.
\end{align*}
  \end{enumerate}
\end{definition}

\begin{remark}
(1) Note also that $\cal A^{n+1}_w\subset\cal A^n_w$ as well as $\widetilde{\cal A}^{n+1}_w\subset\widetilde{\cal A}^n_w$ for $n\in\bb N$.
(2) In Lemma \ref{lem:boundary fields not empty} we shall show that these classes are not empty. In fact the definitions are intended to allow Li\'enard-Wiechert fields generated by any once continuously differentiable asymptotes with strictly time-like and uniformly bounded accelerations.
\end{remark}

With this definition we can formalize the term ``conditional $\text{WF}_\varrho$ solution'' for given Newtonian Cauchy data and prescribed boundary fields which we have discussed in Section \ref{sec:main}:
\begin{definition}[Conditional $\text{WF}_\varrho$ Solutions]\label{def:WF sol for finite times}
  Let $T>0$, $p\in\cal P$ and $X^\pm\in\cal A^1_w$ be given. The set $\cal T_T^{p,X^\pm}$ consists of elements $(\vect q_i,\vect p_i)_{1\leq i\leq N}\in\cal T_{\text{\clock}}^{N}$ that solve the conditional $\text{WF}_{\varrho}$ equations  (\ref{eqn:bWF equation written out})-(\ref{eqn:WF with boundary fields}) for Newtonian Cauchy data $p=(\vect q_{i,t},\vect q_{i,t})_{1\leq i\leq N}|_{t=0}$. We shall refer to elements in $\cal T^{p,X^\pm}_T$ as conditional $\text{WF}_\varrho$ solutions for initial value $p$ and boundary fields $X^\pm_T$.
\end{definition}

Furthermore, we define the potential fixed point map $S^{p,X^\pm}_T$ as discussed in Section \ref{sec:main} where we make use of the notation and results presented in Section \ref{sec:maxwellsolutions} and Section \ref{sec:summary of part I}.

\begin{definition}[Fixed Point Map $S_T^{p,X^{\pm}}$]\label{def:STX}
  For any given finite $T>0$, $p\in\cal P$ and $X^\pm\in\cal A_w^1$, we define
  \begin{align*}
    S_T^{p,X^\pm}:D_w(\mA)\to D_w(\mA^\infty), \qquad F\mapsto S_T^{p,X^\pm}[F]
  \end{align*}
  by
  \begin{align*}
    S_T^{p,X^\pm}[F]:=\frac{1}{2}\sum_{\pm}\left[\mW_{\mp T}X^\pm_{\pm T}[p,F]+\int_{\pm T}^tds\;\mW_{-s}\mJ(\varphi_s[p,F])\right]
  \end{align*}
  where $s\mapsto\varphi_s[p,F]:=M_L[p,F](s,0)$ denotes the $\text{ML-SI}_\varrho$ solution, cf. Definition \ref{def:ML time evolution}, for initial value $(p,F)\in D_w(A)$.
\end{definition}

Next we make sure that this map is well-defined and that its fixed points, if they exist, have corresponding charge trajectories in $\cal T^{p,X^\pm}_T$, i.e. the conditional $\text{WF}_\varrho$ solutions.

\begin{theorem}[$S_T^{p,X^{\pm}}$ and its Fixed Points]\label{thm:the map ST}
  For any finite $T>0$, $p\in\cal P$ and $X^\pm\in\cal A_w^1$ the following is true:
  \begin{enumerate}[(i)]
  \item The map $S_T^{p,X^\pm}$ is well-defined.
  \item Given $F\in D_w(\mA)$, setting $(X^\pm_{i,\pm T})_{1\leq i\leq N}:=X^\pm_{\pm T}[p,F]$ and denoting the $\text{ML-SI}_\varrho$ charge trajectories
  \begin{align}\label{eqn:fixed point charge trajectories}
    t\mapsto(\vect q_{i,t},\vect q_{i,t})_{1\leq i\leq N}:=\pQP M_L[p,F](t,0)
  \end{align}
  by $(\vect q_i,\vect p_i)_{1\leq i\leq N}$ we have
  \begin{align*}
    S_T^{p,X^\pm}[F]=\frac{1}{2}\sum_\pm\bigg(M_{\varrho_{i}}[X^\pm_{i,\pm T},(\vect q_i,\vect p_i)](0,\pm T)\bigg)_{1\leq i\leq N}
  \end{align*}
  as well as $S_T^{p,X^\pm}[F]\in D_w(\mA^\infty)\cap\cal F^{N}$.
  \item For any $F\in D_{w}(\mA)$ such that $F=S_T^{p,X^\pm}[F]$ the corresponding charge trajectories (\ref{eqn:fixed point charge trajectories}) are in $\cal T^{p,X^\pm}_T$.
  \end{enumerate}
\end{theorem}
\begin{proof}
  (i) Let $F\in D_w(\mA)$, then $(p,F)\in D_w(A)$, and hence, by Theorem \ref{thm:globalexistenceanduniqueness} the map $t\mapsto\varphi_t:=M_L[\varphi](t,0)$ is a once continuously differentiable map $\bb R\to D_w(A)\subset\cal H_w$. By properties of $J$ stated in \cite[Lemma 2.22]{bauer_maxwell_2010} \todo{check reference} we know that $A^kJ:\cal H_w\to D_w(A^\infty)\subset\cal H_w$ is locally Lipschitz continuous for any $k\in\bb N$. By projecting onto field space $\cal F_w$, cf. Definition \ref{def:AWJ}, we obtain that also $\mA^k\mJ:\cal H_w\to D_w(\mA^\infty)\subset\cal F_w$ is locally Lipschitz continuous. Hence, by the group properties of $(\mW_t)_{t\in\bb R}$ we know that $s\mapsto \mW_{-s}\mA^k\mJ(\varphi_s)$ for any $k\in\bb N$ is continuous. Furthermore, $\mA$ is closed. This implies the commutation
  \begin{align*}
    \mA^k\int_{\pm T}^0ds\;\mW_{-s}\mJ(\varphi_s)=\int_{\pm T}^0ds\;\mW_{-s}\mA^k\mJ(\varphi_s).
  \end{align*}
 As this holds for any $k\in\bb N$, $\int_{\pm T}^0ds\;\mW_{-s}\mJ(\varphi_s)\in D_w(\mA^\infty)$. Furthermore, by Definition \ref{def:boundary fields} the term $X_{\pm T}^\pm[p,F]$ is in $D_w(\mA^\infty)$ and therefore $\mW_{\mp T}X_{\pm T}^\pm[p,F]\in D_w(\mA^\infty)$ by the group properties. Hence, the map $S_T^{p,X^\pm}$ is well-defined as a map $D_w(A)\to D_w(\mA^\infty)$.

  (ii) For $F\in D_w(\mA)$ let $(\vect q_i,\vect p_i)_{1\leq i\leq N}$ denote the charge trajectories $t\mapsto(\vect q_{i,t},\vect p_{i,t})_{1\leq i\leq N}=\pQP \varphi_t$ of $t\mapsto\varphi_t:=M_L[p,F](t,0)$, which by $(p,F)\in D_w(A)$ and Theorem \ref{thm:globalexistenceanduniqueness} are once continuously differentiable. Since the absolute value of the velocity is given by $\|\vect v(\vect p_{i,t})\|=\frac{\|\vect p_{i,t}\|}{\sqrt{m^2+\vect p_{i,t}^2}}<1$, we conclude that $(\vect q_i,\vect p_i)_{1\leq i\leq N}$ are also time-like and therefore in $\cal T^{N}_{\text{\clock}}$, cf. Definition \ref{def:charge trajectory}. Furthermore, the boundary fields $X^\pm_{\pm T}[p,F]$ are in $D_w(\mA^\infty)\cap\cal F^{N}$ and obey the Maxwell constraints by the definition of $\cal A_w^n$. So we can apply Lemma \ref{lem:connection maxwell time and Wt J} which states for $(X^\pm_{i,\pm T})_{1\leq i\leq N}:=X^\pm_{\pm T}[p,F]$ that
  \begin{align}\label{eqn:SpXT field}
    \big(M_{\varrho_{i}}[X_{i,\pm T}^{\pm},(\vect q_i,\vect p_i](t,\pm T)\big)_{1\leq i\leq N}=\mW_{t\mp T}X^\pm_{\pm T}[p,F]+\int_{\pm T}^tds\;\mW_{t-s}\mJ(\varphi_s)\in D_w(\mA)\cap\cal F^{N}.
  \end{align}
  For $t=0$ this proves claim (ii).

  (iii) Finally, assume there is an $F\in D_{w}(\mA)$ such that $F=S_T^{p,X^\pm}[F]$. By (ii) this implies $F\in D_w(\mA^\infty)\cap\cal F^{N}$. Let $(\vect q_i,\vect p_i)_{1\leq i\leq N}$ and  $t\mapsto\varphi_t$ be defined as in the proof of (ii) which now is infinitely often differentiable as $\bb R\to\cal H_w$ since $(p,F)\in D_w(A^\infty)$. We shall show later that the following integral equality holds
  \begin{align}\label{eqn:SpXT integral equation}
    \varphi_t=(p,0)+\int_0^tds\; \pQP J(\varphi_s)+\frac{1}{2}\sum_{\pm}\left[W_{t\mp T}(0, X^\pm_{\pm T}[p,F])+\int_{\pm T}^tds\;W_{t-s}\pF J(\varphi_s)\right]
  \end{align}
  for all $t\in\bb R$; note that $t\mapsto\varphi_t:=M_L[p,F](t,0)$ depends also on $(p,F)$. For now, suppose (\ref{eqn:SpXT integral equation}) holds. Then the differentiation with respect to time $t$ of the phase space components of $(\vect q_{i,t},\vect p_{i,t},\vect E_{i,t},\vect B_{i,t})_{1\leq i\leq N}:=\varphi_t$ yields $\partial_t\pQP\varphi_t=\pQP J(\varphi_t)$, which by definition of $J$ gives
  \begin{align}\label{eqn:SpXT lorentz}
    \begin{split}
      \partial_t\vect q_{i,t}&=\vect v(\vect p_{i,t}):=\frac{\vect p_{i,t}}{\sqrt{m^2+\vect p_{i,t}^2}}\\
      \partial_t\vect p_{i,t}&=\sum_{j\neq i}\intdv x\varrho_{i}(\vect x-\vect q_{i,t})\left(\vect E_{j,t}(\vect x)+\vect v(\vect q_{i,t})\wedge\vect B_{j,t}(\vect x)\right).
    \end{split}
  \end{align}
  Furthermore, the field components fulfill
  \begin{align*}
    \pF \varphi_t&=\pF\frac{1}{2}\sum_{\pm}\left[ W_{t\mp T}(0, X^\pm_{\pm T}[\varphi])+\int_{\pm T}^tds\;W_{t-s}\pF J(\varphi_s)\right]\\
    &=\frac{1}{2}\sum_{\pm}\left[\mW_{t\mp T} X^\pm_{\pm T}[p,F]+\int_{\pm T}^tds\;\mW_{t-s}\mJ(\varphi_s)\right]
  \end{align*}
  where we only used the definition of the projectors, cf. Definition \ref{def:AWJ}. Hence, by (\ref{eqn:SpXT field}) we know
  \begin{align}\label{eqn:SpXT wf field}
    (\vect E_{i,t},\vect B_{i,t})=\frac{1}{2}\sum_\pm M_{\varrho_{i}}[F_i,(\vect q_i,\vect p_i](t,\pm T).
  \end{align}
  Furthermore, we have
  \begin{align}\label{eqn:SpXT initial value}
    (\vect q_{i,t},\vect p_{i,t})_{1\leq i\leq N}\big|_{t=0}=p=(\vect q_i^0,\vect p_i^0)_{1\leq i\leq N}.
  \end{align}
  Now, equations (\ref{eqn:SpXT lorentz}), (\ref{eqn:SpXT wf field}) and (\ref{eqn:SpXT initial value}) are exactly the conditional $\text{WF}_\varrho$ equations (\ref{eqn:bWF equation written out})-(\ref{eqn:WF with boundary fields}) for Newtonian Cauchy data $p$ and boundary fields $X^\pm$. Hence, since in (ii) we proved that $(\vect q_i,\vect p_i)_{1\leq i\leq N}$ are in $\cal T^{N}_{\text{\clock}}$, we conclude that they are also in $\cal T^{p,X^\pm}_T$, cf. Definition \ref{def:WF sol for finite times}.

  Finally, it is only left to prove that the integral equation (\ref{eqn:SpXT integral equation}) holds. By Definition \ref{def:ML time evolution}, $\varphi_t$ fulfills
  \begin{align*}
    \varphi_t=W_t(p,F)+\int_0^tds\;W_{t-s}J(\varphi_s)
  \end{align*}
  for all $t\in\bb R$. Inserting the fixed point equation $F=S_T^{p,X^\pm}[F]$, i.e.
  \begin{align*}
    F=\frac{1}{2}\sum_{\pm}\left[\mW_{\mp T}X^\pm_{\pm T}[p,F]+\int_{\pm T}^tds\;\mW_{-s}\mJ(\varphi_s)\right],
  \end{align*}
  we find
  \begin{align*}
    \varphi_t = (p,0)+\frac{1}{2}\sum_\pm W_{t\mp T}\big(0,X^\pm_{\pm T}[p,F]\big) + \frac{1}{2}\sum_\pm W_t\int_{\pm T}^0ds\;W_{-s}\big(0,\mJ(\varphi_s)\big)+\int_0^tds\;W_{t-s}J(\varphi_s).
  \end{align*}
  By the same reasoning as in (i) we may commute $W_t$ with the integral. This together with $J=\pQP J+\pF J$ and $\pQP W_t=\id_{\cal P}$ proves the equality (\ref{eqn:SpXT integral equation}) for all $t\in\bb R$ which concludes the proof.
\end{proof}

Next, we give a simple but physically meaningful element $C\in\widetilde{\cal A}^n_w\cap\cal A^\Lip_w$ to show that neither $\widetilde{\cal A}^n_w$ nor $\cal A^\Lip_w$ is empty.
\begin{definition}[Coulomb Boundary Field]\label{def:coulomb field}
  Define $C:\bb R\times D_w(A) \to D_w(\mA^{\infty})$, $(T,\varphi)\mapsto C_T[\varphi]$ by
  \begin{align*}
    C_T[\varphi]:=\left(\vect E^C_{i}(\cdot-\vect q_{i,T}),0\right)_{1\leq i\leq N}
  \end{align*}
  where $(\vect q_{i,T})_{1\leq i\leq N}:=\pQ M_L[\varphi](T,0)$ and
  \begin{align}\label{eqn:exp coulomb}
    (\vect E^C_{i},0):=M_{\varrho_{i}}[t\mapsto (0,0)](0,-\infty)=\left(\int d^3z\; \varrho_{i}(\cdot-\vect z)\frac{\vect z}{\|\vect z\|^3},0\right).
  \end{align}
  Note that the equality on the right-hand side of (\ref{eqn:exp coulomb}) holds by Theorem \ref{thm:LWfields}.
\end{definition}

\begin{lemma}[$\widetilde{\cal A}^n_w\cap\cal A^\Lip_w$ is Non-Empty]\label{lem:boundary fields not empty}
  Let $n\in\bb N$ and $w\in\cal W$. The map $C$ given in Definition \ref{def:coulomb field} is an element of $\widetilde{\cal A}_w^n\cap\cal A^\Lip_w$.
\end{lemma}
\begin{proof}
  We need to show the properties (i)-(v) given in Definition \ref{def:boundary fields}. Fix $T>0$ and $p\in\cal P$. Let $\varphi\in D_w(A)$ such that $\pQP \varphi=p$ and set $F:=\pF\varphi$. Furthermore, we define $(\vect q_{i,T})_{1\leq i\leq N}:=\pQ M_L[\varphi](T,0)$. Since $\vect E^C_{i}$ is a Li\'enard-Wiechert field of the constant charge trajectory $t\mapsto(\vect q_{i,T},0)$ in $\cal T_{\text{\clock}!}^1$, we can apply Corollary \ref{cor:LW_estimate} to yield the following estimate for any multi-index $\alpha\in\bb N^3_{0}$ and $\vect x\in\bb R^3$
  \begin{align}\label{eqn:coulomb est}
    \left\|D^\alpha\vect E^C_{i}(\vect x)\right\|_{\bb R^3}\leq \frac{\constr{LW_bound}^{(\alpha)}}{1+\|\vect x\|^2}.
  \end{align}
  which allows to define the finite constants $\constl{coulomb const}^{(\alpha)}:=\left\|D^\alpha\vect E^C_{i}\right\|_{L^2_w}$. Using the properties of the weight $w\in\cal W$, see (\ref{eqn:weightclass}), we find
  \begin{align*}
    &\|C_T[\varphi]\|_{\cal F_w^n}^2\leq\sum_{k=0}^n\|\mA^k C_T[\varphi]\|_{\cal F_w}\leq\sum_{k=0}^n\sum_{i=1}^N\left\|(\nabla\wedge)^k \vect E^C_{i}(\cdot-\vect q_{i,T})\right\|_{L^2_w}\leq\sum_{k=0}^n\sum_{|\alpha|\leq k}\sum_{i=1}^N\left\|D^\alpha \vect E^C_{i}(\cdot-\vect q_{i,T})\right\|_{L^2_w}\\
    &\leq \sum_{k=0}^n\sum_{|\alpha|\leq k}\sum_{i=1}^N\left(1+\namer{cw}\left\|\vect q_{i,T}\right\|\right)^{\frac{\namer{pw}}{2}}\left\|D^\alpha \vect E^C\right\|_{L^2_w}\leq \sum_{k=0}^n\sum_{|\alpha|\leq k}\sum_{i=1}^N\left(1+\namer{cw}\left\|\vect q_{i,T}\right\|\right)^{\frac{\namer{pw}}{2}}\constr{coulomb const}^{(\alpha)}<\infty.
  \end{align*}
  This implies $C_T[\varphi]\in D_w(\mA^\infty)\cap\cal F^{N}$ and that $C:\bb R\times D_w(A)\to D_w(\mA^\infty)\cap\cal F^{N}$ is well-defined. Note that the right-hand side depends only on $\left\|\vect q_{i,T}\right\|$ which is bounded by
  \begin{align}\label{eqn:vel_bound}
  \left\|\vect q_{i,T}\right\|\leq \|\pQ p\|+|T|
  \end{align}
 since the maximal velocity is bounded by one, i.e. the speed of light. Hence, property (i) holds for
  \begin{align*}
    \constr{c hwn norm}^{(n)}(|T|,\|p\|):=\sum_{k=0}^n\sum_{|\alpha|\leq k}\sum_{i=1}^N\left(1+\namer{cw}\left(\|\pQ p\|+|T|\right)\right)^{\frac{\namer{pw}}{2}}\constr{coulomb const}^{(\alpha)}.
  \end{align*}

  Instead of showing property (ii), we prove the stronger property (v). For this let $\widetilde\varphi\in D_w(A)$ such that $\pQP\varphi=\pQP\widetilde\varphi$ and set $(\widetilde{\vect q}_{i,T})_{1\leq i\leq N}:=\pQ M_L[\widetilde\varphi](T,0)$. Starting with
  \begin{align*}
    \|C_T[\varphi]-C_T[\widetilde\varphi]\|_{\cal F_w^1}\leq\sum_{i=1}^N\sum_{|\alpha|\leq 1}\left\|D^\alpha\left(\vect E^C(\cdot-\vect q_{i,T})-\vect E^C(\cdot-\widetilde{\vect q}_{i,T})\right)\right\|_{L^2_w}
  \end{align*}
  we compute
  \begin{align*}
    \left\|D^\alpha\left(\vect E^C(\cdot-\vect q_{i,T})-\vect E^C(\cdot-\widetilde{\vect q}_{i,T})\right)\right\|_{L^2_w}&=\left\|\int_0^1 d\lambda\;(\widetilde{\vect q}_{i,T}-\vect q_{i,t})\cdot\nabla D^\alpha\vect E^C(\cdot-\widetilde{\vect q}_{i,T}+\lambda(\widetilde{\vect q}_{i,T}-\vect q_{i,t}))\right\|_{L^2_w}\\
    &\leq \int_0^1 d\lambda\;\left\|(\vect q_{i,t}-\widetilde{\vect q}_{i,T})\cdot\nabla D^\alpha\vect E^C(\cdot-\widetilde{\vect q}_{i,T}+\lambda(\widetilde{\vect q}_{i,T}-\vect q_{i,t}))\right\|_{L^2_w}.
  \end{align*}
 Therefore, for all $|\alpha|\leq 1$ we get
  \begin{align*}
    &\sum_{|\alpha|\leq 1}\left\|D^\alpha\left(\vect E^C(\cdot-\vect q_{i,T})-\vect E^C(\cdot-\widetilde{\vect q}_{i,T})\right)\right\|_{L^2_w}\\
    &\leq \|\vect q_{i,T}-\widetilde{\vect q}_{i,T}\|_{\bb R^3}\sup_{0\leq \lambda\leq 1}\sum_{|\beta|\leq 2}\left\|D^\beta\vect E^C(\cdot+\lambda(\vect q_{i,T}-\widetilde{\vect q}_{i,T}))\right\|_{L^2_w}.
  \end{align*}
  The estimate (\ref{eqn:coulomb est}), $0\leq \lambda\leq 1$ and the properties of $w\in\cal W$ yield
  \begin{align*}
    &\left\|D^\beta\vect E^C(\cdot-\widetilde{\vect q}_{i,T}+\lambda(\widetilde{\vect q}_{i,T}-\vect q_{i,t}))\right\|_{L^2_w}\leq \left(1+\namer{cw}\left\|\widetilde{\vect q}_{i,T}-\lambda(\widetilde{\vect q}_{i,T}-\vect q_{i,t})\right\|_{\bb R^3}\right)^\frac{\namer{pw}}{2} \left\|D^\beta\vect E^C\right\|_{L^2_w}\\
    &\leq (1+\namer{cw}(\|\vect q_{i,T}\|_{\bb R^3}+\|\widetilde{\vect q}_{i,T}\|_{\bb R^3})^\frac{\namer{pw}}{2}\constr{coulomb const}^{(\beta)}
  \end{align*}
  Hence, because of bound (\ref{eqn:vel_bound}), property (v) holds for
  \[
    \constr{c strong lipschitz}(|T|,\|\varphi\|_{\cal H_w},\|\widetilde \varphi\|_{\cal H_w}):=N\sum_{|\beta|\leq 2}(1+\namer{cw}(\|\pQ\varphi\|_{\bb R^3}+\|\pQ \widetilde\varphi\|_{\bb R^3}+2|T|)^\frac{\namer{pw}}{2}\constr{coulomb const}^{(\beta)}
  \]
(iii) holds by Theorem \ref{thm:maxwell_solutions}. (iv) Let $B_\tau(0)\subset\bb R^3$ be a ball of radius $\tau>0$ around the origin. For any $F\in D_w(\mA)$ we define $(\vect q_{i,T})_{1\leq i\leq N}:=\pQ M_L[\varphi](T,0)$. It holds
  \begin{align*}
    \|C^T[p,F]\|_{\cal F_w^n(B_\tau^c(0))}&\leq \sum_{i=1}^N\sum_{|\alpha|\leq n}\left\|D^\alpha\vect E^C(\cdot-\vect q_{i,T})\right\|_{L^2_w(B_\tau^c(0))}\\
    &\leq \sum_{i=1}^N\sum_{|\alpha|\leq n}\left(1+\namer{cw}\|\vect q_{i,T}\|\right)^{\frac{\namer{pw}}{2}}\left\|D^\alpha\vect E^C\right\|_{L^2_w(B_\tau^c(\vect q_{i,T}))}.
  \end{align*}
  We use again that the maximal velocity is smaller than one, i.e. $\|\vect q_{i,T}\|\leq\|\vect q_i^0\|+T$. Hence, for $\tau>\|\vect q_i^0\|+T$ define $r(\tau):=\tau-\|\vect q_i^0\|+T$ such that we can estimate the $L^2_w(B_\tau^c(\vect q_{i,T}))$ norm by the $L^2_w(B_{r(\tau)}^c(0))$ norm and yield
\[
    \sup_{F\in D_w(\mA)}\|C^T[p,F]\|_{\cal F_w^n(B_\tau^c(0))}\leq \sum_{i=1}^N\sum_{|\alpha|\leq n}\left(1+\namer{cw}\|\vect q_{i,T}\|\right)^{\frac{\namer{pw}}{2}}\left\|D^\alpha\vect E^C\right\|_{L^2_w(B_{r(\tau)}^c(0))}\xrightarrow[\tau\to\infty]{}0
\]  This concludes the proof.
\end{proof}
\begin{remark}
  When looking for global $\text{WF}_{\varrho}$ solutions, in view of (\ref{eqn:WF with boundary fields}) and (\ref{eqn:WF boundary fields}), the boundary fields can be seen as a good guess of how the charge trajectories $(\vect q_i^0,\vect p_i)_{1\leq i\leq N}$ continue outside of the time interval $[-T,T]$. Without much modification of Lemma \ref {lem:boundary fields not empty} one can also treat the Li\'enard-Wiechert fields of a charge trajectory which starts at $\vect q_{i,T}$ and has constant momentum $\vect p_{i,T}$ using the notation $(\vect q_{i,T},\vect p_{i,T})_{1\leq i\leq N}:=\pQP M_L[\varphi](T,0)$ (the result is the Lorentz boosted Coulomb field). Such boundary fields are also in $\widetilde{\cal A}^n_w\cap\cal A^\Lip_w$ since the derivative $\partial_{s}{\vect p}_{i,s}$ for $s\in[-T,T]$ can be expressed by $J$ which is locally Lipschitz continuous by \cite[Lemma 2.22]{bauer_maxwell_2010} \todo{check reference} and the $\text{ML-SI}_{\varrho}$ dynamics are well controllable on the interval $[-T,T]$; see (ii) of Theorem \ref{thm:globalexistenceanduniqueness}.
\end{remark}

We collect the needed estimates and properties of $S_T^{p,X^\pm}$ in the following three lemmas.

\begin{lemma}[$\cal F_w^n$ Estimates]\label{lem:estimates on hwn}
  For $n\in\bb N_0$ the following is true:
   \begin{enumerate}[(i)]
     \item For all $t\in\bb R$ and $F\in D_w(\mA^n)$ it holds that $\|\mW_t F\|_{\cal F_w^n}\leq e^{\gamma|t|}\|F\|_{\cal F_w^n}$.
     \item For all $\varphi\in \cal H_w$ there is a $\constl{cfj radius}^{(n)}\in\bounds$ such that
         \begin{align*}
           \|\mJ(\varphi)\|_{\cal F_w^n}\leq \constr{cfj radius}^{(n)}(\|\pQ \varphi\|_{\cal H_w}).
         \end{align*}
     \item For all $\varphi,\widetilde \varphi\in \cal H_w$ there is a $\constl{cj2 hwn}^{(n)}\in\bounds$ such that
         \begin{align*}
           \|\mJ(\varphi)-\mJ(\widetilde\varphi)\|_{\cal F_w^n}\leq \constr{cj2 hwn}^{(n)}(\|\varphi\|_{\cal H_w},\|\widetilde\varphi\|_{\cal H_w})\|\varphi-\widetilde\varphi\|_{\cal H_w}.
         \end{align*}
   \end{enumerate}
\end{lemma}
\begin{proof}
  (i) As shown in \cite[Lemma 2.19]{bauer_maxwell_2010}\todo{check reference}, $A$ on $D_w(A)$ generates a $\gamma$-contractive group $(W_t)_{t\in\bb R}$; cf. Definition \ref{def:Wt}. This property is inherited from $\mA$ on $D_{w}(\mA)$ which generates the group $(\mW_{t})_{t\in\bb R}$. Hence, $\mA$ and $\mW_t$ commute for any $t\in\bb R$ which implies for all $F\in D_w(\mA^n)$ that
  \begin{align*}
    \|\mW_tF\|_{\cal F_w^n}^2=\sum_{k=0}^n\|\mA^k \mW_t F\|_{\cal F_w}^2=\sum_{k=0}^n\|\mW_t\mA^k F\|_{\cal F_w}^2\leq e^{\gamma|t|}\sum_{k=0}^n\|\mA^k F\|_{\cal F_w}^2=e^{\gamma|t|}\|F\|_{\cal F_w^n}.
  \end{align*}

  For (ii) let $(\vect q_i,\vect p_i,\vect E_i,\vect B_i)_{1\leq i\leq N}=\varphi\in\cal H_w$. Using then the definition of $\mJ$, cf. Definition \ref{def:operator_J} and \ref{def:AWJ}, we find
  \begin{align*}
    \|\mJ(\varphi)\|_{\cal F_w^n}\leq\sum_{i=1}^{N}\sum_{k=0}^n\|(\nabla\wedge)^k \vect v(\vect p_i)\varrho_{i}(\cdot-\vect q_i)\|_{L^2_w}.
  \end{align*}
  By applying the triangular inequality one finds a constant $\constl{rot alpha const}$ such that
  \begin{align*}
    \|(\nabla\wedge)^k \vect v(\vect p_i)\varrho_{i}(\cdot-\vect q_i)\|_{L^2_w}\leq (\constr{rot alpha const})^n\sum_{|\alpha|\leq n}\|\vect v(\vect p_i)D^\alpha\varrho_{i}(\cdot-\vect q_i)\|_{L^2_w}\leq (\constr{rot alpha const})^n\sum_{|\alpha|\leq n}\|D^\alpha\varrho_{i}(\cdot-\vect q_i)\|_{L^2_w}
  \end{align*}
  whereas in the last step we used the fact that the maximal velocity is smaller than one. Using the properties of the weight function $w\in\cal W$, cf. Definition \ref{def:weighted spaces}, we conclude
  \begin{align*}
    \|D^\alpha\varrho_{i}(\cdot-\vect q_i)\|_{L^2_w}\leq (1+\namer{cw}\|\vect q_i\|)^{\frac{\namer{pw}}{2}}\|D^\alpha\varrho_{i}\|_{L^2_w}.
  \end{align*}
  Collecting these estimates we yield that claim (ii) holds for
  \begin{align*}
    \constr{cfj radius}^{(n)}(\|\pQ\varphi\|_{\cal H_w}):=(\constr{rot alpha const})^n\sum_{i=1}^N(1+\namer{cw}\|\vect q_i\|)^{\frac{\namer{pw}}{2}}\sum_{|\alpha|\leq n}\|D^\alpha\varrho_{i}\|.
  \end{align*}

  Claim (iii) is shown by repetitively applying estimate of \cite[Lemma 2.22]{bauer_maxwell_2010} \todo{check reference} on the right-hand side of
  \begin{align*}
    \|\mJ(\varphi)-\mJ(\widetilde\varphi)\|_{\cal F_w^n}\leq\sum_{k=0}^n\|A^k[J(\varphi)-J(\widetilde\varphi)]\|_{\cal H_w}
  \end{align*}
  which yields a constant $\constr{cj2 hwn}^{(n)}:=\sum_{k=0}^n\constl{cj2}^{(k)}(\|\varphi\|_{\cal H_w},\|\widetilde\varphi\|_{\cal H_w})$ where $\constr{cj2}\in\bounds$ is given in the proof of \cite[Lemma 2.22]{bauer_maxwell_2010}\todo{check reference}. This concludes the proof.
\end{proof}

\begin{lemma}[Properties of $S^{p,X^\pm}_T$]\label{lem:estimates for ST}
  Let $0<T<\infty$, $p\in\cal P$ and $X^\pm\in\cal A_w^n$ for $n\in\bb N$. Then it holds:
  \begin{enumerate}[(i)]
    \item There is a $\constl{ST radius const}\in\bounds$ such that for all $F\in \cal F_w^1$ we have
        \begin{align*}
          \|S_T^{p,X^\pm}[p,F]\|_{\cal F_w^n}\leq \constr{ST radius const}^{(n)}(T,\|p\|).
        \end{align*}
    \item $F\mapsto S^{p,X^\pm}_T[F]$ as $\cal F_w^1\to\cal F_w^1$ is continuous.
  \end{enumerate}
  If $X^\pm\in\cal A^\Lip_w$, it is also true that:
  \begin{enumerate}[(i)]
    \setcounter{enumi}{2}
    \item There is a $\constl{ST Lipschitz const res}\in\bounds$ such that for all $F,\widetilde F\in \cal F_w^1$ we have
        \begin{align*}
          \|S^{p,X^\pm}_T[F]-S^{p,X^\pm}_T[\widetilde F]\|_{\cal F_w^1}\leq T\constr{ST Lipschitz const res}(T,\|p\|,\|F\|_{\cal F_w},\|\widetilde F\|_{\cal F_w})\|F-\widetilde F\|_{\cal F_w}.
        \end{align*}
  \end{enumerate}
\end{lemma}
\begin{proof}
   Fix a finite $T>0$, $p\in\cal P$, $X^\pm\in\cal A_w^n$ for $n\in\bb N$. Before we prove the claims we preliminarily recall the relevant estimates of the $\text{ML-SI}_\varrho$ dynamics. Throughout the proof and for any $F,\widetilde F\in\cal F_w^n$ we use the notation
   \[
   D_w(A^n)\ni\varphi\equiv(p,F), \qquad D_w(A^n)\ni\widetilde\varphi\equiv(p,\widetilde F),\]
   and furthermore,
   \[\varphi_t:=M_L[\varphi](t,0),\qquad \widetilde\varphi_t:=M_L[\widetilde\varphi](t,0),\] 
   for any $t\in\bb R$. Recall the estimates given in (ii) of  Theorem \ref{thm:globalexistenceanduniqueness} which gives the following $T$ dependent upper bounds on these $\text{ML-SI}_\varrho$ solutions:
  \begin{align}\label{eqn:mlsi estimate 1}
     \sup_{t\in[-T,T]}\|\varphi_t-\widetilde\varphi_t\|_{\cal H_w}\leq \constr{apriori lipschitz}(T,\|\varphi\|_{\cal H_w},\|\widetilde\varphi\|_{\cal H_w})\|\varphi-\widetilde\varphi\|_{\cal H_w},
  \end{align}
  \begin{align}\label{eqn:mlsi estimate 2}
     \sup_{t\in[-T,T]}\|\varphi_t\|_{\cal H_w}\leq \constr{apriori lipschitz}(T,\|\varphi\|_{\cal H_w},0)\|\varphi\|_{\cal H_w} && \text{and} && \sup_{t\in[-T,T]}\|\widetilde\varphi_t\|_{\cal H_w}\leq \constr{apriori lipschitz}(T,\|\widetilde\varphi\|_{\cal H_w},0)\|\widetilde\varphi\|_{\cal H_w}.
  \end{align}
To prove claim (i) we estimate
  \begin{align*}
    \|S^{p,X^\pm}_T[F]\|_{\cal F_w^n}&\leq \left\|\frac{1}{2}\sum_\pm \mW_{\mp T}X^\pm_{\pm T}[p,F]\right\|_{\cal F_w^n} + \left\|\frac{1}{2}\sum_\pm\int_{\pm T}^0ds\;\mW_{-s}\mJ(\varphi_s)\right\|_{\cal F_w^n}
    =:\terml{uni bound 2}+\terml{uni bound 3},
  \end{align*}
  cf. Definition \ref{def:STX}. By the estimate given in (i) of Lemma \ref{lem:estimates on hwn} and the property (i) of Definition \ref{def:boundary fields} we find
  \begin{align*}
    \termr{uni bound 2}\leq \frac{1}{2}\sum_\pm\|\mW_{\mp T}X^\pm_{\pm T}[p,F]\|_{\cal F_w^n}\leq e^{\gamma T}\|X^\pm_{\pm T}[p,F]\|_{\cal F_w^n}\leq e^{\gamma T}\constr{c hwn norm}^{(n)}( T,\|\varphi\|_{\cal H_w}).
  \end{align*}
  Furthermore, using in addition the estimates (i)-(ii) of Lemma \ref{lem:estimates on hwn} we get a bound for the other term by
  \begin{align*}
    \termr{uni bound 3}\leq T e^{\gamma T}\sup_{s\in[-T,T]}\|\mJ(\varphi_s)\|_{\cal F_w^n}\leq T e^{\gamma T}\sup_{s\in[-T,T]}\constr{cfj radius}(\|\pQ\varphi_s\|_{\cal H_w})\leq T e^{\gamma T}\constr{cfj radius}(\|p\|+T)
  \end{align*}
  whereas the last step is implied by the fact that the maximal velocity is below one. These estimates prove claim (i) for
  \begin{align*}
    \constr{ST radius const}^{(n)}(T,\|\phi\|_{\cal H_w^n}) := e^{\gamma T}\left(\constr{c hwn norm}^{(n)}( T,\|p\|)+T\constr{cfj radius}(\|p\|+T)\right).
  \end{align*}

  Next we prove claim (ii). Therefore, we consider
  \begin{align*}
    \|S_T^{p,X^\pm}[F]-S_T^{p,X^\pm}[\widetilde F]\|_{\cal F_w^n}&\leq e^{\gamma T}\|X^\pm_{\pm T}[\varphi]-X^\pm_{\pm T}[\widetilde\varphi]\|_{\cal F_w^n}+Te^{\gamma T}\sup_{s\in[-T,T]}\left\|\mJ(\varphi_s)-\mJ(\widetilde\varphi_s)\right\|_{\cal F_w^n}\\&=:\terml{ST lipschitz 1}+\terml{ST lipschitz 2}
  \end{align*}
  where we have already applied (i) of Lemma \ref{lem:estimates on hwn}. Next we apply (iii) of Lemma \ref{lem:estimates on hwn} to the term \termr{ST lipschitz 2} and yield
  \begin{align*}
    \termr{ST lipschitz 2}\leq Te^{\gamma T}\sup_{s\in[-T,T]}\constr{cj2 hwn}^{(n)}(\|\varphi_s\|_{\cal H_w},\|\widetilde\varphi_s\|_{\cal H_w})\|\varphi_s-\widetilde\varphi_s\|_{\cal H_w}.
  \end{align*}
  Finally, by the $\text{ML-SI}_\varrho$ estimates (\ref{eqn:mlsi estimate 1}) and (\ref{eqn:mlsi estimate 2}) we have
  \begin{align}\label{eqn:ST lip est 2}
    \termr{ST lipschitz 2}\leq T\constl{apriori lipschitz pre}(T,\|p\|,\|F\|_{\cal F_w^n},\|\widetilde F\|_{\cal F_w^n})\|\varphi-\widetilde\varphi\|_{\cal H_w}
  \end{align}
  for
  \begin{align*}
    \constr{apriori lipschitz pre}(T,\|p\|,\|F\|_{\cal F_w^n},\|\widetilde F\|_{\cal F_w^n}) := & e^{\gamma T}\constr{cj2 hwn}^{(n)}\bigg(\constr{apriori lipschitz}(T,\|\varphi\|_{\cal H_w},0)\|\varphi\|_{\cal H_w},\constr{apriori lipschitz}(T,0,\|\widetilde\varphi\|_{\cal H_w})\|\varphi\|_{\cal H_w}\bigg)\times\\
    &\times\constr{apriori lipschitz}(T,\|\varphi\|_{\cal H_w},\|\widetilde\varphi\|_{\cal H_w}).
  \end{align*}
  By this estimate and (ii) of Definition \ref{def:boundary fields} the limit $\widetilde F\to F$ in $\cal F_w^1$ implies $S^{p,X^\pm}_T[\widetilde F]\to S^{p,X^\pm}_T[F]$ in $\cal F_w^1$ since here $\|\varphi-\widetilde\varphi\|_{\cal H_w}=\|F-\widetilde F\|_{\cal F_w}$. Hence, the claim (ii) is true.

  (iii) Let now $X^\pm\in\cal A^\Lip_w$. By (v) of Definition \ref{def:boundary fields} the term \termr{ST lipschitz 1} then behaves according to
  \begin{align*}
    \termr{ST lipschitz 1}&\leq T\constr{c strong lipschitz}^{(n)}(|T|,\|\varphi\|_{\cal H_w},\|\widetilde\varphi\|_{\cal H_w})\;\|\varphi-\widetilde\varphi\|_{\cal H_w}
  \end{align*}
  Together with the estimate (\ref{eqn:ST lip est 2}) this proves claim (ii) for
  \begin{align*}
    \constr{ST Lipschitz const res}^{(n)}(T,\|p\|,\|F\|_{\cal F_w},\|\widetilde F\|_{\cal F_w}):=\constr{c strong lipschitz}^{(n)}(|T|,\|\varphi\|_{\cal H_w},\|\widetilde\varphi\|_{\cal H_w})+\constr{apriori lipschitz pre}(T,\|p\|,\|F\|_{\cal F_w^n},\|\widetilde F\|_{\cal F_w^n})
  \end{align*}
  since in our case $\|\varphi-\widetilde\varphi\|_{\cal H_w}=\|F-\widetilde F\|_{\cal F_w}$.
\end{proof}

\begin{remark}\label{rem:uniqueness}
Let $p\in\cal P$, $X^{\pm}\in A^{\Lip}_{w}$. Then claim (iii) of Lemma \ref{lem:estimates for ST} has an immediate consequence:
 For sufficiently small $T$ the mapping $S^{p,X^{\pm}}_{T}$ has a unique fixed point, which follows by Banach's fixed point theorem.
Consider therefore $X^\pm\in\cal A^\Lip_w\subset\cal A^1_w$, then (i) of Lemma \ref{lem:estimates for ST}  states
  \begin{align*}
    \|S_T^{p,X^\pm}[p,F]\|_{\cal F_w^1}\leq \constr{ST radius const}^{(1)}(T,\|p\|)=:r.
  \end{align*}
  Hence, the map $S_T^{p,X^\pm}$ restricted to the ball $B_r(0)\subset\cal F_w^1$ with radius $r$ around the origin is a nonlinear self-mapping. Claim (iii) of Lemma \ref{lem:estimates for ST} states for all $T>0$ and $F,\widetilde F \in B_r(0)\subset D_w(\mA)$ that
  \begin{align*}
    \|S_T^{p,X^\pm}[F]-S_T^{p,X^\pm}[\widetilde F]\|_{\cal F_w^1}&\leq T \constr{ST Lipschitz const res}(T,\|p\|,\|F\|_{\cal F_w},\|\widetilde F\|_{\cal F_w})\|F-\widetilde F\|_{\cal F_w}\\
    &\leq T\constr{ST Lipschitz const res}(T,\|p\|,r,r)\|F-\widetilde F\|_{\cal F_w}.
  \end{align*}
  where we have also used that $\constr{ST Lipschitz const res}\in\bounds$ is a continuous and strictly increasing function of its arguments. Hence, for $T$ sufficiently small we have $T \constr{ST Lipschitz const res}(T,\|p\|,r,r)<1$ such that $S_T^{p,X^\pm}$ is a contraction on $B_r(0)\subset\cal F_w^1$. However, for larger $T$ we loose control on the uniqueness of the fixed point. This highlights an interesting aspect of dynamical systems. E.g. for the $\text{ML-SI}_{\varrho}$ dynamics it means that solutions are still uniquely characterized not only by Newtonian Data and fields $(p,F)\in D_{w}(A)$ at time $t=0$ but also by specifying Newtonian Cauchy data $p\in\cal P$ at time $t=0$ and fields $F$ at a different time $t=T$. The maximal $T$ will in general be inverse proportional to the Lipschitz constant of the vector field.
  \end{remark}

We come to the proof of Theorem \ref{thm:ST has a fixed point} where we shall use the following criterion for precompactness of sequences in $L^2_w$.

\begin{lemma}[Criterion for Precompactness]\label{lem:precompactness}
  Let $(\vect F_n)_{n\in\bb N}$ be a sequence in $L^2_w(\bb R^3,\bb R^3)$ such that
  \begin{enumerate}[(i)]
    \item The sequence $(\vect F_n)_{n\in\bb N}$ is uniformly bounded in $H_w^\triangle$, defined in (\ref{eqn:sobolev spaces}).
    \item $\lim_{\tau\to\infty}\sup_{n\in\bb N}\|\vect F_n\|_{L^2_w(B_\tau^c(0))}=0$.
  \end{enumerate}
  Then the sequence $(F_n)_{n\in\bb N}$ is precompact, i.e. it contains a convergent subsequence.
\end{lemma}
\begin{proof}
  The proof of this claim can be seen as a special case of \cite[Chapter 8, Proof of Theorem 8.6, p.208]{lieb_analysis_2001} and can be found in \cite{deckert_electrodynamic_2010}.
\end{proof}
Of course, one solely needs control on the gradient. However, the Laplace turns out to be more convenient for our later application of this lemma.

\begin{proof}[\textbf{Proof of Theorem \ref{thm:ST has a fixed point} (Existence of Conditional $\text{WF}_\varrho$ Solutions)}]
  Fix $p\in\cal P$.  Given a finite $T>0$, $p\in\cal P$ and $X^\pm\in\widetilde{\cal A}^{3}_w$ claim (i) of Lemma \ref{thm:the map ST} states  for all $F\in\cal F_w^1$
  \begin{align}\label{eqn:a3w radius}
    \|S_T^{p,X^\pm}[p,F]\|_{\cal F_w^1}\leq\|S_T^{p,X^\pm}[p,F]\|_{\cal F_w^3}\leq \constr{ST radius const}^{(3)}(T,\|p\|)=:r.
  \end{align}
  Let $K$ be the closed convex hull of $M:=\{S_T^{p,X^\pm}[F]\;|\; F\in \cal F_w^1\}\subset B_r(0)\subset\cal F_w^1$. By (ii) of Lemma \ref{thm:the map ST} we know that the map $S_T^{p,X^\pm}:K\to K$ is continuous as a map $\cal F_w^1\to\cal F_w^1$. Note that if $M$ were compact so would be $K$ and we could infer by Schauder's fixed point theorem the existence of a fixed point. 
  
  It remains to verify that M is compact. Therefore, let $(G_m)_{m\in\bb N}$ be a sequence in $M$. With the help of Lemma \ref{lem:precompactness} we shall show now that it contains an $\cal F_w^1$ convergent subsequence. By definition there is a sequence $(F_m)_{m\in\bb N}$ in $B_r(0)\subset\cal F_w^1$ such that $G_m:=S_T^{p,X^\pm}[F_m]$, $m\in\bb N$; note that $G_m$ is an element of $D_{w}(\mA^{\infty})$ and therefore also of $\cal F^{n}_{w}$ for any $n\in\bb N$. We define for $m\in\bb N$ the electric and magnetic fields 
  \begin{align*}
      (\vect E_i^{(m)},\vect B_i^{(m)})_{1\leq i\leq N}:=S_T^{p,X^\pm}[F_m].
  \end{align*}
  Recall the definition of the norm of $\cal F_w^n$, cf. Definition \ref{def:Fwn}, for some $(\vect E_i,\vect B_i)_{1\leq i\leq N}=F\in\cal F_w^n$ and $n\in\bb N$
  \begin{align}\label{eqn:hwn norm}
    \|F\|_{\cal F_w^n}^2=\sum_{k=0}^n\|\mA^k F\|_{\cal F_w}^2=\sum_{k=0}^n\sum_{i=1}^N\left(\|(\nabla\wedge)^k\vect E_i\|^2_{L^2_w}+\|(\nabla\wedge)^k\vect B_i\|^2_{L^2_w}\right).
  \end{align}
  Therefore, since $\mA$ on $D_w(\mA)$ is closed, $(G_m)_{m\in\bb N}$ has an $\cal F_w^1$ convergent subsequence if and only if all the sequences $((\nabla\wedge)^k\vect E_i^{(m)})_{m\in\bb N}$, $((\nabla\wedge)^k\vect B_i^{(m)})_{m\in\bb N}$ for $k=0,1$ and $1\leq i\leq N$ have a common convergent subsequence in $L^2_w$.

  To show that this is the case we first provide the bounds needed for condition (i) of Lemma \ref{lem:precompactness}. Estimate (\ref{eqn:a3w radius}) implies that
  \begin{align}\label{eqn:ST first bound}
    \sum_{k=0}^3\sum_{i=1}^N\left(\|(\nabla\wedge)^k\vect E^{(m)}_i\|^2_{L^2_w}+\|(\nabla\wedge)^k\vect B^{(m)}_i\|^2_{L^2_w}\right)=\|G_m\|^2_{\cal F_w^3}\leq r^2
  \end{align}
  for all $m\in\bb N$. Furthermore, by (ii) of Lemma \ref{thm:the map ST} the fields $(\vect E_i^{(m)},\vect B_i^{(m)})_{1\leq i\leq N}$ are the fields of a Maxwell solution at time zero, and hence, by Theorem \ref{thm:maxwell_solutions} fulfill the Maxwell constraints for $(\vect q^0_i,\vect p_i^0)_{1\leq i\leq N}:=p$ which read
  \begin{align*}
    \nabla\cdot\vect E_{i}^{(m)}=4\pi\varrho_{i}(\cdot-\vect q^0_i), && \nabla\cdot\vect B_i^{(m)}=0.
  \end{align*}
  Also by Theorem \ref{thm:maxwell_solutions}, $G_m$ is in $\cal F^{N}$ so that for every $k\in\bb N_0$
  \begin{align*}
    (\nabla\wedge)^{k+2}\vect E_i^{(m)}=4\pi(\nabla\wedge)^{k}\nabla\varrho_{i}(\cdot-\vect q_i^0)-\triangle(\nabla\wedge)^{k}\vect E_i^{(m)}, &&
    (\nabla\wedge)^{k+2}\vect B_i^{(m)}=-\triangle(\nabla\wedge)^{k}\vect B_i^{(m)}.
  \end{align*}
  Estimate (\ref{eqn:ST first bound}) implies for all $m\in\bb N$ that
  \begin{align*}
    &\sum_{k=0}^1\sum_{i=1}^N\left(\|\triangle(\nabla\wedge)^k\vect E^{(m)}_i\|^2_{L^2_w}+\|\triangle(\nabla\wedge)^k\vect B^{(m)}_i\|^2_{L^2_w}\right)\\
    &\leq
    2\sum_{k=0}^1\sum_{i=1}^N\left(\|(\nabla\wedge)^{k+2}\vect E^{(m)}_i\|^2_{L^2_w}+\|(\nabla\wedge)^{k+2}\vect B^{(m)}_i\|^2_{L^2_w}\right)+2\sum_{i=1}^N\|4\pi\nabla\varrho_{i}(\cdot-\vect q_i^0)\|_{L^2_w}\\
    &\leq 2r^2 + 8\pi\sum_{i=1}^N\left(1+\namer{cw}\left\|\vect q_i^0\right\|\right)^{\namer{pw}}\|\nabla\varrho_{i}\|^2_{L^2_w}
  \end{align*}
  where we made use of the properties of the weight $w\in\cal W$. Note that the right-hand does not depend on $m$. Therefore, all the sequences $((\nabla\wedge)^k\vect E_i^{(m)})_{m\in\bb N}$, $(\triangle(\nabla\wedge)^k\vect E_i^{(m)})_{m\in\bb N}$, $((\nabla\wedge)^k\vect B_i^{(m)})_{m\in\bb N}$, $(\triangle(\nabla\wedge)^k\vect B_i^{(m)})_{m\in\bb N}$ for $k=0,1$ and $1\leq i\leq N$ are uniformly bounded in $L^2_w$. This ensures condition (i) of Lemma \ref{lem:precompactness}.

  Second, we need to show that all the sequences $((\nabla\wedge)^k\vect E_i^{(m)})_{m\in\bb N}$, $((\nabla\wedge)^k\vect B_i^{(m)})_{m\in\bb N}$ for $k=0,1$ and $1\leq i\leq N$ decay uniformly at infinity to meet condition (ii) of Lemma \ref{lem:precompactness}. Define \ifarxiv{\linebreak}{}$(\vect E_{i,\pm T}^{(m),\pm},\vect B_{i,\pm T}^{(m),\pm})_{1\leq i\leq N}:=X^\pm_{\pm T}[p,F_m]$ for $m\in\bb N$ and denote the $i$th charge trajectory $t\mapsto(\vect q^{(m)}_{i,t},\vect p^{(m)}_{i,t}):=\pQP M_L[p,F_m](t,0)$ by $(\vect q^{(m)}_i,\vect p^{(m)}_i)$, $1\leq i\leq N$. Using Lemma \ref{thm:the map ST}(ii) and afterwards Theorem \ref{thm:maxwell_solutions} we can write the fields as
  \begin{align*}
    \begin{pmatrix}
      \vect E_i^{(m)}\\
      \vect B_i^{(m)}
    \end{pmatrix}
    &=\frac{1}{2}\sum_\pm M_{\varrho_{i}}[(\vect E^\pm_{i,\pm T},\vect E^\pm_{i,\pm T}),(\vect q^{(m)}_i,\vect p^{(m)}_i)](0,\pm T)\\
    &=\frac{1}{2}\sum_\pm\bigg[\begin{pmatrix}
    \partial_t & \nabla\wedge\\
    -\nabla\wedge & \partial_t
    \end{pmatrix}
    K_{t\mp T}*\begin{pmatrix}
      \vect E^{(m),\pm}_{i,\pm T}\\
      \vect B^{(m),\pm}_{i,\pm T}
    \end{pmatrix}
    + K_{t\mp T}*\begin{pmatrix}
      -4\pi \vect j^{(m)}_{i,\pm T}\\
      0
    \end{pmatrix}
    \\&\quad+ 4\pi \int_{\pm T}^{t} ds\; K_{t-s} * \begin{pmatrix}
      -\nabla & - \partial_s \\
      0 & \nabla\wedge
    \end{pmatrix}
    \begin{pmatrix}
      \rho^{(m)}_{i,s}\\
      \vect j^{(m)}_{i,s}
    \end{pmatrix}
    \bigg]_{t=0}
    =:\terml{uni decay 1}+\terml{uni decay 2}+\terml{uni decay 3}
  \end{align*}
  where $\rho^{(m)}_{i,t}:=\varrho_{i}(\cdot-\vect q^{(m)}_{i,t})$ and $\vect j^{(m)}_{i,t}:=\vect v(\vect p^{(m)}_{i,t})\rho_{i,t}$ for all $t\in\bb R$.

  We shall show that there is a $\tau^*>0$ such that for all $m\in\bb N$ the terms \termr{uni decay 2} and \termr{uni decay 3} are point-wise zero on $B_{\tau^*}^c(0)\subset\bb R^3$. Recalling the computation rules for $K_t$ from Lemma \ref{lem:Greens_function_dalembert} we calculate for term \termr{uni decay 2}
  \begin{align*}
    \|4\pi[K_{\mp T}*\vect j^{(m)}_{i,\pm T}](\vect x)\|_{\bb R^3}\leq 4\pi T\underset{\partial B_T{(\vect x)}}\fint d\sigma(y)\;|\varrho_{i}(\vect y-\vect q^{(m)}_{\pm T})|.
  \end{align*}
  The right-hand side is zero for all $\vect x\in\bb R^3$ such that $\partial B_T(\vect x)\cap\supp\varrho_{i}(\cdot-\vect q_{\pm T})=\emptyset$. Because the charge distributions have compact support there is a $R>0$ such that $\supp\varrho_{i}\subseteq B_R(0)$ for all $1\leq i\leq N$. Now for any $1\leq i\leq N$ and $m\in\bb N$ we have
  \begin{align*}
    \supp\varrho_{i}(\cdot-\vect q^{(m)}_{i,\pm T})\subseteq B_R(\vect q^{(m)}_{i,\pm T})\subseteq B_{R+T}(\vect q_i^0)
  \end{align*}
  since the supremum of the velocities of the charge $\sup_{t\in[-T,T],m\in\bb N}\|\vect v(\vect p_{i,t}^{(m)})\|$ is less than one. Hence, $\partial B_T(\vect x)\cap B_{R+T}(\vect q_i^0)=\emptyset$ for all $\vect x\in B_{\tau}^c(0)$ with $\tau>\|p\|+R+2T$.

  Considering \termr{uni decay 3} we have
  \begin{align}\label{eqn:uni est 2}
    \left\|4\pi \int_{\pm T}^{0} ds\; \left[K_{-s} * \begin{pmatrix}
      -\nabla & - \partial_s \\
      0 & \nabla\wedge
    \end{pmatrix}
    \begin{pmatrix}
      \rho^{(m)}_{i,s}\\
      \vect j^{(m)}_{i,s}
    \end{pmatrix}\right](\vect x)\right\|_{\bb R^6}\leq 4\pi\int_{\pm T}^0ds\;s\underset{\partial B_{|s|}(\vect x)}\fint d\sigma(y)\;\|\vect G(\vect y-\vect q^{(m)}_{s})\|_{\bb R^6}
  \end{align}
  where we used the abbreviation
  \begin{align*}
    \vect G(\vect y-\vect q^{(m)}_{s}):=\begin{pmatrix}
      -\nabla & - \partial_s \\
      0 & \nabla\wedge
    \end{pmatrix}
    \begin{pmatrix}
      \rho^{(m)}_{i,s}\\
      \vect j^{(m)}_{i,s}
    \end{pmatrix}(\vect y)
  \end{align*}
  and the computation rules for $K_t$ given in Lemma \ref{lem:Greens_function_dalembert}. As $\supp\vect G\subseteq\supp\varrho_{i}\subseteq B_R(0)$, the right-hand side of (\ref{eqn:uni est 2}) is zero for all $\vect x\in\bb R$ such that
  \begin{align*}
    \bigcup_{s\in[-T,T]}\left[\partial B_{|s|}(\vect x)\cap B_{R}(\vect q^{(m)}_{i,s})\right]=\emptyset.
  \end{align*}
Now the left-hand side above is a subset of
  \begin{align*}
    \cup_{s\in[-T,T]}\partial B_{|s|}(\vect x)\bigcap \cup_{s\in[-T,T]}B_{R}(\vect q^{(m)}_{i,s})
    \qquad \subseteq \qquad B_T(\vect x)\cap B_{R+T}(\vect q^0_i)
  \end{align*}
  which equals the empty set for all $\vect x\in B_{\tau}^c(0)$ with $\tau>\|p\|+R+2T$.

  Hence, setting $\tau^*:=\|p\|+R+2T$ we conclude that for all $\tau>\tau^*$ the terms \termr{uni decay 2} and \termr{uni decay 3} and all their derivatives are zero on $B_\tau^c(0)\subset\bb R^3$. That means in order to show that all the sequences $((\nabla\wedge)^k\vect E_i^{(m)})_{m\in\bb N}$, $((\nabla\wedge)^k\vect B_i^{(m)})_{m\in\bb N}$ for $k=0,1$ and $1\leq i\leq N$ decay uniformly at spatial infinity, it suffices to show
  \begin{align}\label{eqn:uniform estimate with e and b}
    \lim_{\tau\to\infty}\sup_{m\in\bb N}\sum_{k=0}^1\sum_{i=1}^N\left(\|(\nabla\wedge)^k\vect e_i^{(m)}\|_{L^2_w(B_\tau^c(0))}+\|(\nabla\wedge)^k\vect b_i^{(m)}\|_{L^2_w(B_\tau^c(0))}\right)=0.
  \end{align}
  for
  \begin{align*}
      \begin{pmatrix}
      \vect e^{(m)}_i\\
      \vect b^{(m)}_i
    \end{pmatrix}:=\termr{uni decay 1}=\begin{pmatrix}
    \partial_t & \nabla\wedge\\
    -\nabla\wedge & \partial_t
    \end{pmatrix}
    K_{t\mp T}*\begin{pmatrix}
      \vect E^{(m),\pm}_{i,\pm T}\\
      \vect B^{(m),\pm}_{i,\pm T}
    \end{pmatrix}\bigg|_{t=0}
  \end{align*}
  for $1\leq i\leq N$. Let $\vect F\in\cal C^\infty(\bb R^3,\bb R^3)$ and $\tau>0$. By computation rules for $K_t$ given in Lemma \ref{lem:Greens_function_dalembert}  we get
  \begin{align*}
    &\|\nabla\wedge K_{\mp T}*\vect F\|_{L^2_w(B_{\tau+T}^c(0))}= \|K_{\mp T}*\nabla\wedge\vect F\|_{L^2_w(B_{\tau+T}^c(0))}\leq \left\|T\underset{\partial B_T(0)}\fint d\sigma(y)\;\nabla\wedge\vect F(\cdot-\vect y)\right\|_{L^2_w(B_{\tau+T}^c(0))}\\
    &\leq T\underset{\partial B_T(0)}\fint d\sigma(y)\;\|\nabla\wedge\vect F(\cdot-\vect y)\|_{L^2_w(B_{\tau+T}^c(0))}\leq T\sup_{\vect y\in\partial B_T(0)}\|\nabla\wedge\vect F(\cdot-\vect y)\|_{L^2_w(B_{\tau+T}^c(0))}\\
    &\leq T\sup_{\vect y\in\partial B_T(0)}(1+\namer{cw}\|\vect y\|)^{\frac{\namer{pw}}{2}}\|\nabla\wedge\vect F(\cdot-\vect y)\|_{L^2_w(B_{\tau+T}^c(0))}
    \leq T(1+\namer{cw}T)^{\frac{\namer{pw}}{2}}\|\nabla\wedge\vect F\|_{L^2_w(B_{\tau}^c(0))},
  \end{align*}
  and
  \begin{align*}
    &\|\partial_t K_{t\mp T}|_{t=0}*\vect F\|_{L^2_w(B_{\tau+T}^c(0))}=\left\|\underset{\partial B_T(0)}\fint d\sigma(y)\;\vect F(\cdot-\vect y)+\frac{T^2}{3}\underset{B_T(0)}\fint d^3y\; \triangle\vect F(\cdot-\vect y)\right\|_{L^2_w(B_{\tau+T}^c(0))}\\
    &\leq \underset{\partial B_T(0)}\fint d\sigma(y)\;\|\vect F(\cdot-\vect y)\|_{L^2_w(B_{\tau+T}^c(0))}+\frac{T^2}{3}\underset{B_T(0)}\fint d^3y\; \|\triangle\vect F(\cdot-\vect y)\|_{L^2_w(B_{\tau+T}^c(0))}\\
    &\leq (1+\namer{cw}T)^{\frac{\namer{pw}}{2}}\|\vect F\|_{L^2_w(B_{\tau}^c(0))}+\frac{T^2}{3}(1+\namer{cw}T)^{\frac{\namer{pw}}{2}} \|\triangle\vect F\|_{L^2_w(B_{\tau}^c(0))}.
  \end{align*}
  Substituting $\vect F$ with $(\nabla\wedge)^k\vect E^{(m),\pm}_{i,\pm T}$ and $(\nabla\wedge)^k\vect B^{(m),\pm}_{i,\pm T}$ for $k=0,1$ and $1\leq i\leq N$ in the two estimates above yields
  \begin{align}\label{eqn:little e b}
    \begin{split}
      &\sum_{k=0}^1\sum_{i=1}^N\left(\|(\nabla\wedge)^k\vect e_i^{(m)}\|_{L^2_w(B_{\tau+T}^c(0))}+\|(\nabla\wedge)^k\vect b_i^{(m)}\|_{L^2_w(B_{\tau+T}^c(0))}\right)\\
      &\leq (1+\namer{cw}T)^{\frac{\namer{pw}}{2}}
      \bigg(
        \|(\nabla\wedge)^{k}\vect E^{(m),\pm}_{i,\pm T}\|_{L^2_w(B_{\tau}^c(0))}+\|(\nabla\wedge)^{k}\vect B^{(m),\pm}_{i,\pm T}\|_{L^2_w(B_{\tau}^c(0))}+\\
          &\hskip2.2cm+\frac{T^2}{3}\left(\|(\nabla\wedge)^{k}\triangle\vect E^{(m),\pm}_{i,\pm T}\|_{L^2_w(B_{\tau}^c(0))}+\|(\nabla\wedge)^{k}\triangle\vect B^{(m),\pm}_{i,\pm T}\|_{L^2_w(B_{\tau}^c(0))}\right)+\\
            &\hskip2.2cm+ T\left(\|(\nabla\wedge)^{k+1}\vect E^{(m),\pm}_{i,\pm T}\|_{L^2_w(B_{\tau}^c(0))}+\|(\nabla\wedge)^{k+1}\vect B^{(m),\pm}_{i,\pm T}\|_{L^2_w(B_{\tau}^c(0))}\right)
      \bigg).
    \end{split}
  \end{align}
  Now the boundary fields $X^\pm$ lie in $\widetilde{\cal A}^{3}_w$ which means that the fields $\vect E^{(m),\pm}_{i,\pm T}$ and $\vect B^{(m),\pm}_{i,\pm T}$ for $1\leq i\leq N$ fulfill the Maxwell constraints so that
  \begin{align*}
    \|(\nabla\wedge)^k\triangle\vect E^{(m),\pm}_{i,\pm T}\|_{L^2_w(B_{\tau}^c(0))} &= \|(\nabla\wedge)^{k+2}\vect E^{(m),\pm}_{i,\pm T}\|_{L^2_w(B_{\tau}^c(0))}+
    4\pi\|(\nabla\wedge)^{k}\nabla\varrho_{i}(\cdot-\vect q^{(m)}_{i,\pm T}\|_{L^2_w(B_{\tau}^c(0))}
  \end{align*}
  and
  \begin{align*}
    \|(\nabla\wedge)^k\triangle\vect B^{(m),\pm}_{i,\pm T}\|_{L^2_w(B_{\tau}^c(0))} &= \|(\nabla\wedge)^{k+2}\vect B^{(m),\pm}_{i,\pm T}\|_{L^2_w(B_{\tau}^c(0))}.
  \end{align*}
  Applying (iv) of Definition \ref{def:boundary fields} yields
  \begin{align*}
    \lim_{\tau\to\infty}\sup_{m\in\bb N}\sum_{j=0}^{3}\sum_{i=1}^N\left\|(\nabla\wedge)^{j}\vect E^{(m),\pm}_{i,\pm T}\|^2_{L^2_w(B_{\tau}^c(0))}+\|(\nabla\wedge)^{j}\vect B^{(m),\pm}_{i,\pm T}\|^2_{L^2_w(B_{\tau}^c(0))}\right)=\lim_{\tau\to\infty}\sup_{m\in\bb N}\|\chi^\pm_{\pm T}[p,F_m]\|^2_{\cal F_w^n(B_{\tau}^{c}(0)}=0
  \end{align*}
  because $F_m\in B_r(0)\subset\cal F_w^1$ for all $m\in\bb N$. Hence, (\ref{eqn:uniform estimate with e and b}) holds which, as we have shown, implies the uniform decay at spatial infinity of all the sequences $((\nabla\wedge)^k\vect E_i^{(m)})_{m\in\bb N}$, $((\nabla\wedge)^k\vect B_i^{(m)})_{m\in\bb N}$ for $k=0,1$ and $1\leq i\leq N$. This ensures condition (ii) of Lemma \ref{lem:precompactness}.

  Using the abbreviations $\vect E_i^{(m,k)}:=(\nabla\wedge)^k\vect E_i^{(m)}$ and  $\vect B_i^{(m,k)}:=(\nabla\wedge)^k\vect B_i^{(m)}$ for  $1\leq i\leq N$, $k=0,1$, and $m\in\bb N$ we summarize: The sequences $(\vect E_i^{(m,k)})_{m\in\bb N}$, $(\vect B_i^{(m,k)})_{m\in\bb N}$, $(\triangle\vect E_i^{(m,k)})_{m\in\bb N}$ and $(\triangle\vect B_i^{(m,k)})_{m\in\bb N}$ are all uniformly bounded in $L^2_w$ and decay uniformly at spatial infinity.

  Successively application of Lemma \ref{lem:precompactness} produces the common $\cal F_w^1$ convergent subsequence: Fix $1\leq i\leq N$. Let $(\vect E_i^{(m^0_l,0)})_{l\in\bb N}$ be the $L^2_w$ convergent subsequence of $(\vect E_i^{(m,0)})_{m\in\bb N}$ and $(\vect E_i^{(m^1_l,1)})_{l\in\bb N}$ the $L^2_w$ convergent subsequence of $(\vect E_i^{(m^0_l,1)})_{l\in\bb N}$. In the same way we proceed with the other indices $1\leq i\leq N$ and the magnetic fields, every time choosing a further subsequence of the previous one. Let us denote the final subsequence by $(m_l)_{l\in\bb N}\subset\bb N$. Then we have constructed sequences $(G_{m_l})_{l\in\bb N}$ as well as $(\mA G_{m_l})_{l\in\bb N}$ which are convergent in $\cal F_w^{0}$ and implies the convergence in $\cal F_w^1$. As $(G_m)_{m\in\bb N}$ was arbitrary, we conclude that every sequence in $M$ has an $\cal F_w^1$ convergent subsequence and therefore $M$ is compact which had to be shown. This concludes the proof.
\end{proof}

Having established the existence of a fixed point $F$ for any finite $T>0$, $p\in\cal P$ and $(X^\pm_{i,\pm T})_{1\leq i\leq N}=X^\pm\in\widetilde{\cal A}^3_w$, claim (iii) of Theorem \ref{thm:the map ST} states that the charge trajectories $t\mapsto(\vect q_{i,t},\vect p_{i,t})_{1\leq i\leq N}:=\pQP M_L[p,F](t,0)$ are in $\cal T^{p,X^\pm}_T$, i.e. that they are time-like charge trajectories that solve the conditional $\text{WF}_\varrho$ equations (\ref{eqn:bWF equation written out})-(\ref{eqn:WF with boundary fields}) for all times $t\in\bb R$. As discussed in the introduction it is interesting to verify that among those solutions we see truly advanced and delayed interactions between the charges.  This is the content of Theorem \ref{thm:existence of L} which we prove next. We introduce:

%
\begin{definition}[Partial $\text{WF}_\varrho$ solutions]
  For $p\in\cal P$ we define $\cal T_\WF^L$ to be the set of time-like charge trajectories $(\vect q_i,\vect p_i)_{1\leq i\leq N}\in\cal T_{\text{\clock}}^N$ which solve the $\text{WF}_\varrho$ equations (\ref{eqn:WF equation written out})-(\ref{eqn:WF fields def}) for times $t\in[-L,L]$ and have initial conditions $(\vect q_{i,t},\vect p_{i,t})_{1\leq i\leq N}|_{t=0}=p$. We shall call every element of $\cal T^L_\WF$ a partial $\text{WF}_{\varrho}$ solution on $[-L,L]$ for initial value $p$.
\end{definition}
In order to show that a conditional $\text{WF}_\varrho$ solution $(\vect q_i,\vect p_i)_{1\leq i\leq N}\in\cal T^{p,X^\pm}_T$ is also a partial $\text{WF}_\varrho$ solution we have to regard the difference between the WF fields produce by them:
\begin{align}\label{eqn:bwf wf diff}
  \begin{split}
    &M_{\varrho_{i}}[X^\pm_{i,\pm T},(\vect q_i,\vect p_i)](t,\pm T)-M_{\varrho_{i}}[\vect q_i,\vect p_i](t,\pm\infty)\\
    &=
      \begin{pmatrix}
      \partial_t & \nabla\wedge\\
      -\nabla\wedge & \partial_t
      \end{pmatrix}
      K_{t\mp T}*X^\pm_{i,\pm T}
      + K_{t\mp T}*\begin{pmatrix}
        -4\pi \vect v(\vect p_{i,\pm T})\varrho_{i}(\cdot-\vect q_{i,\pm T})\\
        0
      \end{pmatrix}\\
          &\quad- 4\pi \int_{\pm\infty}^{\pm T} ds\; K_{t-s} * \begin{pmatrix}
      -\nabla & - \partial_s \\
      0 & \nabla\wedge
    \end{pmatrix}
    \begin{pmatrix}
      \varrho_{i}(\cdot-\vect q_{i,s})\\
      \vect v(\vect p_{i,s})\varrho_{i}(\cdot-\vect q_{i,s})
    \end{pmatrix}.
  \end{split}
\end{align}
The equality holds by Definition \ref{def:Maxwell time evolution}, Theorem \ref{thm:maxwell_solutions} and (\ref{eqn:LW_fields}) in Theorem \ref{thm:LWfields}. Let $R>0$ be the smallest radius such that $\supp\varrho_{i}\subseteq B_{R}(0)$ for all $1\leq i\leq N$. Whenever there is an $L>0$ such that this difference is zero at least for times $t\in[-L,L]$ and in tubes of radius $R$ around the charge trajectories $(\vect q_i,\vect p_i)_{1\leq i\leq N}$ of a conditional $\text{WF}_\varrho$ solution, these trajectories form already a partial $\text{WF}_{\varrho}$ solution in $\cal T_\WF^L$.

 Suppose that the boundary fields $X^{\pm}_{i,\pm T}$ are given by the advanced and retarded Li\'enard-Wiechert fields of asymptotes which continue the conditional $\text{WF}_\varrho$ solution $(\vect q_i,\vect p_i)_{1\leq i\leq N}$ for times $|t|>T$ into the future and past, respectively. Only by looking at the geometry of the interaction, see Figure \ref{fig:wf diamond}, one may expect that the difference (\ref{eqn:bwf wf diff}) is zero in the intersection of all forward and backward light-cones of the space-time points $(-T,\vect q_{k,-T})$ and $(+T,\vect q_{k,+T})$ for all $1\leq k\leq N$, respectively. However, this intersection might be empty, either if $T$ is chosen too small compared to the maximal distance of the charges at time zero, or if the charges approach the speed of light sufficiently fast; see Figure \ref{fig:wf extreme} for an extreme case. As discussed, the properties of known $\text{WF}_\varrho$ solutions suggest that the latter case will never occur as their velocities are expected to be uniformly bounded away from the speed of light. In the particular case of the Coulomb boundary fields $C$, cf. Definition \ref{def:coulomb field}, we shall now show that, even without having such a uniform velocity estimate, for fixed $T>0$ there is always a suitable choice of Newtonian Cauchy data $p\in\cal P$ and non-trivial charge densities $\varrho_{i}\in\cal C^\infty_c$ such that partial $\text{WF}_\varrho$ solutions exist.

\begin{figure}[ht]
  \begin{center}
    \subfigure[\label{fig:wf diamond}]{
      \if\arxiv 1
        \includegraphics[scale=1]{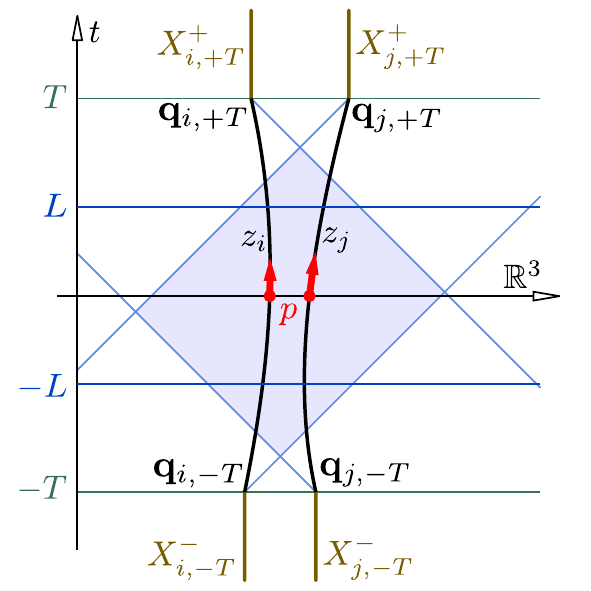}
      \else
        \includegraphics[scale=.7]{figures/wf_diamond}
      \fi
    }
    \subfigure[\label{fig:shadows}]{
      \if\arxiv 1
        \includegraphics[scale=1]{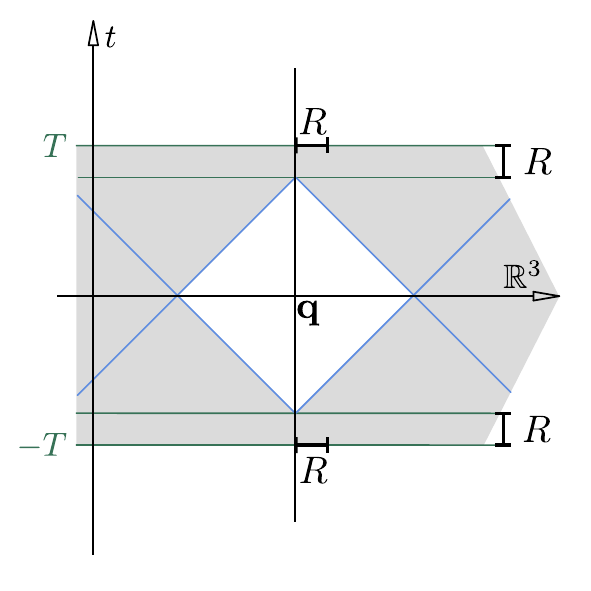}
      \else
        \includegraphics[scale=.7]{figures/shadows}
       \fi
    }
    \caption{(a) Choosing Li\'enard-Wiechert fields for $X^{\pm}_{i,\pm T}$, $1\leq i\leq N$, the difference (\ref{eqn:bwf wf diff}) between the conditional $\text{WF}_\varrho$ and partial $\text{WF}_\varrho$ solution vanishes inside the shaded (sheared diamond shaped) space-time region, which is given by the intersection of the forward and backward light-cones of $\vect q_{k,-T}$ and $\vect q_{k,+T}$ for $1\leq k\leq N$, respectively. (b) The non-shaded region visualizes the set of space-time points $(t,\vect x)$ such that $t\in(-T+R,T-R)$ and $\vect x\in B_{|t\mp T|-R}(\vect q)$ which is used in Lemma \ref{lem:shadows of boundary fields}.}
  \end{center}
\end{figure}

Observe that the difference (\ref{eqn:bwf wf diff}) is given by the free Maxwell time evolution of the boundary fields which at time $\pm T$ carry divergences at $\vect q_{k,\pm T}$ due to the Maxwell constraints. We shall now exploit the remarkable feature of the free Maxwell time evolution that justifies the discussed geometric picture in Figure \ref{fig:wf diamond}: Any initial field with a non-vanishing divergence will be evolved by the free Maxwell time evolution in a way that the forward and backward light-cones of the support of the divergence is cleared to zero. The next Lemma proves this explicitly in the case of the Coulomb field. Using Lorentz boosts the presented proof can easily be generalized to Coulomb fields of a moving charge with constant velocity, and with a bit more work it can be generalized further to Li\'enard-Wiechert fields of any strictly time-like charge trajectory.

\begin{lemma}[Shadows of the Boundary Fields and $\text{WF}_\varrho$ fields]\label{lem:shadows of boundary fields}
  Let $\vect q,\vect v\in\bb R^3$, $\varrho\in\cal C^\infty_c(\bb R^3,\bb R)$ such that $\supp\varrho\subseteq B_R(0)$ for some finite $R>0$, and let $t\mapsto(\vect q_{t},\vect p_{t})\in\cal T_{\text{\clock}}^N$ such that $\vect q_{t}|_{t=0}=\vect q$. Furthermore,$\vect E^C$ be the Coulomb field of a charge at rest at the origin
  \begin{align*}
    \vect E^C:=\int d^3z\; \varrho(\cdot-\vect z)\frac{\vect z}{\|\vect z\|^3}
  \end{align*}
  Then for $T>R$ the expressions
  \begin{align}\label{eqn:coulomb diff}
    \left[\begin{pmatrix}
    \partial_t & \nabla\wedge\\
    -\nabla\wedge & \partial_t
    \end{pmatrix}
    K_{t\mp T}*\begin{pmatrix}
      \vect E^C(\cdot-\vect q)\\
      0
    \end{pmatrix}
    + K_{t\mp T}*\begin{pmatrix}
      -4\pi \vect v\varrho(\cdot-\vect q)\\
      0
    \end{pmatrix}\right](\vect x)
  \end{align}
  and
	\begin{align}\label{eqn:wf relicts}
	  \int_{\pm\infty}^{\pm T} ds\; K_{t-s} * \begin{pmatrix}
      -\nabla & - \partial_s \\
      0 & \nabla\wedge
    \end{pmatrix}
    \begin{pmatrix}
      \varrho(\cdot-\vect q_{s})\\
      \vect v(\vect p_{s})\varrho(\cdot-\vect q_{s})
    \end{pmatrix}\;(\vect x)
	\end{align}
  equal zero
  for $t\in(-T+R,T-R)$ and $\vect x\in B_{|t\mp T|-R}(\vect q)$; see Figure \ref{fig:shadows}.
\end{lemma}
\begin{proof}
  Let $t\in(-T+R,T-R)$. With regard to the second term in (\ref{eqn:coulomb diff}) we compute
  \begin{align*}
    \left\|-4\pi \vect v\left[K_{t\mp T}*\varrho(\cdot-\vect q)\right](\vect x)\right\|&=4\pi\|\vect v\|\left|(t\mp T)\underset{\partial B_{|t\mp T|}(0)}\fint d\sigma(y)\;\varrho(\vect x-\vect y-\vect q)\right|\\
    &\leq4\pi\|\vect v\||t\mp T|\sup|\varrho|\underset{\partial B_{|t\mp T|}(\vect q)}\fint d\sigma(y)\;\charf{B_R(\vect x)}(\vect y)
  \end{align*}
  where we used Definition \ref{def:greens_dalembert} for $K_{t\mp T}$. Now $\vect x\in B_{|t\mp T|-R}(\vect q)$ implies $\partial B_{|t\mp T|}(\vect q)\cap B_R(\vect x)=\emptyset$, and hence, that the term above is zero.

  With regard to the first term we note that the only non-zero contribution is $\partial_t K_{t\mp T}*\vect E^C_i$ since $\nabla\wedge\vect E^C=0$. Lemma \ref{lem:Greens_function_dalembert} and in particular equation (\ref{eqn:dtKt_f}) give
  \begin{align}\label{eqn:coulomb derivative}
    &\left[\partial_t K_{t\mp T}*\vect E^C(\cdot-\vect q)\right](\vect x)=\underset{\partial{B_{|t\mp T|}}(0)}\fint d\sigma(y)\; \vect E^C(\vect x-\vect y-\vect q)+(t\mp T)\partial_t\underset{\partial{B_{|t\mp T|}}(0)}\fint d\sigma(y)\; \vect E^C(\vect x-\vect y-\vect q)\\
    &=\underset{\partial{B_{|t\mp T|}}(0)}\fint d\sigma(y)\; \vect E^C(\vect x-\vect y-\vect q) + \frac{(t\mp T)^2}{3}\underset{{B_{|t\mp T|}}(0)}\fint d^3y\; \triangle \vect E^C(\vect x-\vect y)=:\terml{coulomb 1}+\terml{coulomb 2}.
  \end{align}
  Using Lebesgue's theorem we start with
  \begin{align*}
    \termr{coulomb 1}&=\vect E^C(\vect x-\vect q)+\int_0^{|t\mp T|}ds\;\partial_s \underset{\partial B_s(0)}\fint d\sigma(y)\; \vect E^C(\vect x-\vect y-\vect q)\\
    &=\vect E^C(\vect x-\vect q)+\int_0^{|t\mp T|}dr\;\frac{r}{3} \underset{B_r(0)}\fint d^3y\; \triangle\vect E^C(\vect x-\vect y-\vect q).
  \end{align*}
  Furthermore, we know that $0=(\nabla\wedge)^2\vect E^C=\nabla(\nabla\cdot\vect E^C)-\triangle\vect E^C$ and $\nabla\cdot\vect E^C=4\pi\varrho$. So we continue the computation with
  \begin{align*}
    \termr{coulomb 1}&=\vect E^C(\vect x-\vect q)+\int_0^{|t\mp T|}dr\;\frac{r}{3} \underset{B_r(0)}\fint d^3y\; 4\pi\nabla\varrho(\vect x-\vect y-\vect q)\\
    &=\vect E^C(\vect x-\vect q)-\int_0^{|t\mp T|}dr\;\frac{1}{r^2} \underset{\partial B_r(0)}\int d\sigma(y)\; \frac{\vect y}{r}\varrho(\vect x-\vect y-\vect q)
  \end{align*}
  where we have used (\ref{eqn:coulomb derivative}) to evaluate the derivative and in addition used Stoke's Theorem. Note that the minus sign in the last line is due to the fact that $\nabla$ acts on $\vect x$ and not on $\vect y$. Inserting the definition of the Coulomb field $\vect E^C$ we finally get
  \begin{align*}
    \termr{coulomb 1}=\underset{B^c_{|t\mp T|}(0)}\int d^3y\; \varrho(\vect x-\vect y-\vect q)\frac{\vect y}{\|\vect y\|^3}.
  \end{align*}
  This integral is zero if, for example, $B^c_{|t\mp T|}(\vect q)\cap B_R(\vect x)=\emptyset$ and this is the case for $\vect x\in B_{|t\mp T|-R}(\vect q)$. So it remains to show that \termr{coulomb 2} also vanishes. Therefore, using $\triangle\vect E^C=4\pi\nabla\varrho$ as before, we get
  \begin{align*}
    \termr{coulomb 2}=-\underset{{\partial B_{|t\mp T|}}(0)}\int d\sigma(y)\; \frac{\vect y}{(t\mp T)^2}\varrho(\vect x-\vect y-\vect q).
  \end{align*}
  This expression is zero, for example, when $\partial B_{|t\mp T|}(\vect q)\cap B_R(\vect x)=\emptyset$ which is true for $\vect x\in B_{|t\mp T|-R}(\vect q)$. Hence, we have shown that for $t\in(-T+R,T-R)$ and $\vect x\in B_{|t\mp T|}(\vect q)$ the term (\ref{eqn:coulomb diff}) is zero.

Looking at the support of the integrand and the integration domain in term (\ref{eqn:wf relicts}) we find that for all $t\in(-T+R,T-R)$ it is zero for all $\vect x\in\bb R^{3}$ such that
\begin{align}
	\bigcup_{|s|>T}\left(\partial B_{|t-s|}(\vect x)\cap B_{R}(\vect q_{s})\right)=\emptyset. 
\end{align}
Hence, for $t\in(-T+R,T-R)$ and $\vect x\in B_{|t\mp T|}(\vect q)$ the term (\ref{eqn:wf relicts}) is also zero which concludes the proof.
\end{proof}

It remains to get a bound on the velocities of the charge trajectories within $[-T,T]$.

\begin{lemma}[Uniform Velocity Bound]\label{lem:uni vel bound}
  For finite $a,b$ there is a continuous and strictly increasing map $v^{a,b}:\bb R^+\to[0,1)$, $T\mapsto v_T^{a,b}$ such that
  \begin{align*}
    \sup\bigg\{\|\vect v(\vect p_{i,t})\|_{\bb R^3}\;\bigg|\;&t\in[-T,T],\|p\|\leq a,F\in\Ran S_T^{p,C},\|\varrho_{i}\|_{L^2_w}+\|w^{-1/2}\varrho_{i}\|_{L^2}\leq b,1\leq i\leq N\bigg\}\\
    \leq v_T^{a,b}<1.
  \end{align*}
  for $(\vect p_{i,t})_{1\leq i\leq N}:=\pP M_L[p,F](t,0)$ for all $t\in\bb R$.
\end{lemma}
\begin{proof}
  Recall the estimate (\ref{eqn:apriori lipschitz  no diff}) from the $\text{ML-SI}_\varrho$ existence and uniqueness Theorem \ref{thm:globalexistenceanduniqueness} which gives the following $T$ dependent upper bounds on these $\text{ML-SI}_\varrho$ solutions for all $\varphi\in D_w(A)$:
  \begin{align}
     \sup_{t\in[-T,T]}\|M_L[\varphi](t,0)\|_{\cal H_w}\leq \constr{apriori ml rho}\left(T,\|\varrho_{i}\|_{L^2_w},\|w^{-1/2}\varrho_{i}\|_{L^2}; 1\leq i\leq N\right)\; \|\varphi\|_{\cal H_w}.
  \end{align}
  Note further that by Lemma \ref{lem:estimates for ST} since $C\in\cal A^1_w$, there is a $\constr{ST radius const}^{(1)}\in\bounds$ such that fields $F\in\Ran S_T^{p,C}\in D_w(\mA^\infty)$ fulfill
  \begin{align*}
    \|F\|_{\cal F_w}\leq \constr{ST radius const}^{(1)}(T,\|p\|)\leq\constr{ST radius const}^{(1)}(T,a).
  \end{align*}
  Therefore, setting $c:=a+\constr{ST radius const}^{(1)}(T,a)$ we estimate the maximal momentum of the charges by
  \begin{align*}
    &\sup\bigg\{\|\vect p_{i,t}\|_{\bb R^3}\;\bigg|\;t\in[-T,T],\|p\|\leq a,F\in\Ran S_T^{p,C},\|\varrho_{i}\|_{L^2_w}+\|w^{-1/2}\varrho_{i}\|_{L^2}\leq b,1\leq i\leq N\bigg\}\\
    &\leq\sup\bigg\{\|\vect p_{i,t}\|_{\bb R^3}\;\bigg|\;t\in[-T,T],\varphi\in D_w(A),\|\varphi\|_{\cal H_w}\leq c,\|\varrho_{i}\|_{L^2_w}+\|w^{-1/2}\varrho_{i}\|_{L^2}\leq b,1\leq i\leq N\bigg\}\\
    &\leq \constr{apriori ml rho}\left(T,b,b,\right)c=:p_T^{a,b}<\infty.
  \end{align*}
  Now, since $\constr{apriori lipschitz}$ as well as $\constr{ST radius const}^{(1)}$ are in $\bounds$ the map $T\mapsto p_T^{a,b}$ as $\bb R^+\to\bb R^+$ is continuous and strictly increasing. We conclude that claim is fulfilled for the choice
  \begin{align*}
    v_T^{a,b}:=\frac{p_T^{a,b}}{\sqrt{m^2+(p_T^{a,b})^2}}<1.
  \end{align*}
\end{proof}

With this we can prove our last result, i.e. Theorem \ref{thm:existence of L}.

\begin{proof}[\textbf{Proof of Theorem \ref{thm:existence of L} (True $\text{WF}_{\varrho}$ Interaction)}]
   Let $F^{*}$ be a fixed point $F^{*}=S_T^{p,C}[F^{*}]$ which exists by Theorem \ref{thm:ST has a fixed point}. Define the charge trajectories $(\vect q_i,\vect p_i)_{1\leq i\leq N}$ by \ifarxiv{\linebreak}{}$t\mapsto(\vect q_{i,t},\vect p_{i,t})_{1\leq i\leq N}:=\pQP M_L[p,F^{*}](t,0)$. By the fixed point properties of $F^{*}$ we know that $(\vect q_i,\vect p_i)_{1\leq i\leq N}\in\cal T^{p,C}_T$ and therefore solve the conditional $\text{WF}_\varrho$ equations (\ref{eqn:bWF equation written out})-(\ref{eqn:WF with boundary fields}) for Newtonian Cauchy data $p$ and boundary fields $C$. In order to show that the charge trajectories $(\vect q_i,\vect p_i)_{1\leq i\leq N}$ are also in $\cal T^L_\WF$ for the given $L$ we need to show that the difference (\ref{eqn:bwf wf diff}), which equals
\begin{align*}
  \begin{split}
    &M_{\varrho_{i}}[X^\pm_{i,\pm T},(\vect q_i,\vect p_i)](t,\pm T)-M_{\varrho_{i}}[\vect q_i,\vect p_i](t,\pm\infty)\\
    &=
      \begin{pmatrix}
      \partial_t & \nabla\wedge\\
      -\nabla\wedge & \partial_t
      \end{pmatrix}
      K_{t\mp T}*X^\pm_{i,\pm T}
      + K_{t\mp T}*\begin{pmatrix}
        -4\pi \vect v(\vect p_{i,\pm T})\varrho_{i}(\cdot-\vect q_{i,\pm T})\\
        0
      \end{pmatrix}\\
          &\quad- 4\pi \int_{\pm\infty}^{\pm T} ds\; K_{t-s} * \begin{pmatrix}
      -\nabla & - \partial_s \\
      0 & \nabla\wedge
    \end{pmatrix}
    \begin{pmatrix}
      \varrho_{i}(\cdot-\vect q_{i,s})\\
      \vect v(\vect p_{i,s})\varrho_{i}(\cdot-\vect q_{i,s})
    \end{pmatrix},
  \end{split}
\end{align*}
   is zero for times $t\in[-L,L]$ at least for all points $\vect x$ in a tube around the position of the $j\neq i$ charge trajectories. Lemma \ref{lem:shadows of boundary fields} states that this expression is zero for all $t\in(-T+R,T-R)$ and $\vect x\in B_{|t\mp T|-R}(\vect q_{i,{\pm T}})$. So it is sufficient to show that the charge trajectories spend the time interval $[-L,L]$ in this particular space-time region. Clearly, the position $\vect q_{i}^0$ at time zero is in $B_{T-R}(\vect q_{i,\pm T})$. Hence, we estimate the earliest exit time $L$ of this space-time region of a charge trajectory $j$, i.e. the time when the $j$th charge trajectory leaves the region $B_{|L\mp T|-R}(\vect q_{i,\pm T})$.   By Lemma \ref{lem:uni vel bound} the charges can in the worst case move apart from each other with velocity $v^{a,b}_T$ during the time interval $[-T,T]$. Putting the origin at $\vect q_i^0$ we can compute the exit time $L$ by
   \begin{align*}
     -v^{a,b}_T T=\|\vect q_j^0-\vect q_i^0\|+2R+v^{a,b}_T L -(T-L)
   \end{align*}

   This gives the lower bound $L:=\frac{(1-v^{a,b}_T)T-\namer{qmax}-2R}{1+v_T^{a,b}}>0$.
   Since $v_{T}^{a,b}<1$ we have $(1-v^{a,b}_T)T>0$. Hence, for sufficiently small $R$ and $\namer{qmax}$ it is true that $L>0$ which concludes the proof.
\end{proof}

\vskip1cm

\noindent\emph{G. Bauer}\\ FH M\"unster\\
Bismarckstra\ss e 11, 48565 Steinfurt, Germany\\

\noindent\emph{D.-A. Deckert}\\
Department of Mathematics, University of California Davis\\
One Shields Avenue, Davis, California 95616, USA\\

\noindent \emph{D. D\"urr}\\
Mathematisches Institut der LMU M\"unchen\\
Theresienstra\ss e 39, 80333 M\"unchen, Germany

\end{document}